\documentclass{article}

\usepackage{tikz}
\usepackage{amsmath}
\usepackage{amssymb, amsthm, amsfonts, tikz-cd, mathrsfs,mathtools,stmaryrd,enumitem, algorithmic,float,comment,tikz}
\usetikzlibrary{backgrounds}
\usetikzlibrary{decorations.pathreplacing}
\usepackage[colorlinks,citecolor=blue,linkcolor=blue,urlcolor=red]{hyperref}
\usepackage{fullpage}
\usepackage{xcolor}
\usepackage{bbm}
\usepackage[ruled,vlined,linesnumbered]{algorithm2e}
\usepackage{thm-restate}

\usepackage[capitalize]{cleveref}
\crefname{theorem}{Theorem}{Theorems}
\crefname{lemma}{Lemma}{Lemmas}
\crefname{claim}{Claim}{Claims}
\crefname{observation}{Observation}{Observations}
\newtheoremstyle{dotless}{}{}{\itshape}{}{\bfseries}{}{ }{}
\theoremstyle{dotless}
\newtheorem{theorem}{Theorem}[section]
\newtheorem{lemma}[theorem]{Lemma}

\newtheorem{claim}[theorem]{Claim}
\newtheorem{observation}[theorem]{Observation}

\newtheoremstyle{dotlessdef}{}{}{\normalfont}{}{\bfseries}{}{ }{}
\theoremstyle{dotlessdef}
\newtheorem{definition}[theorem]{Definition}

\newcommand{\graph}{\text{\normalfont Graph}}
\newcommand{\eps}{\varepsilon}
\newcommand{\dist}{\text{\normalfont dist}}

\usepackage{todonotes}

\title{Covering Approximate Shortest Paths with DAGs}
\author{
Sepehr Assadi\thanks{\url{sepehr@assadi.info}. This work was initiated while the author was at Rutgers University. Sepehr Assadi is supported in part by a Sloan Research Fellowship, an
NSERC Discovery Grant, and a Faculty of Math
Research Chair grant.}
\\University of Waterloo 
\and
Gary Hoppenworth\thanks{\url{garytho@umich.edu}. This work was supported by NSF:AF 2153680.}\\University of Michigan
\and 
Nicole Wein\thanks{\url{nswein@umich.edu}. This work was initiated while the author was at Rutgers University supported by a grant to DIMACS from the Simons Foundation (820931), and continued while the author was at the Simons Institute.}\\University of Michigan}
\date{}

\begin{document}

\maketitle

\begin{abstract}
    We define and study analogs of \emph{probabilistic tree embedding} and \emph{tree cover} for directed graphs. We define the notion of a \emph{DAG cover} of a general directed graph $G$: a small collection $D_1,\dots D_g$ of DAGs so that for all pairs of vertices $s,t$, some DAG $D_i$ provides low distortion for $\dist(s,t)$; i.e. $ \dist_G(s, t) \le \min_{i \in [g]} \dist_{D_i}(s, t) \leq \alpha \cdot \dist_G(s, t)$, where $\alpha$ is the distortion. 
    
    As a trivial upper bound, there is a DAG cover with $n$ DAGs and $\alpha=1$ by taking the shortest-paths tree from each vertex. When each DAG is restricted to be a subgraph of $G$, there is a simple matching lower bound (via a directed cycle) that $n$ DAGs are necessary, even to preserve reachability. Thus, we allow the DAGs to include a limited number of additional edges not from the original graph.

    When $n^2$ additional edges are allowed, there is a simple upper bound of two DAGs and $\alpha=1$. Our first result is an almost-matching lower bound that even for $n^{2-o(1)}$ additional edges, at least $n^{1-o(1)}$ DAGs are needed, even to preserve reachability. However, the story is different when the number of additional edges is $\tilde{O}(m)$, a natural setting where the sparsity of the DAG collection nearly matches that of the original graph. Our main upper bound is that there is a near-linear time algorithm to construct a DAG cover with  $\tilde{O}(m)$ additional edges, polylogarithmic distortion, and only $O(\log n)$ DAGs. This is similar to known results for undirected graphs: the well-known FRT probabilistic tree embedding implies a tree cover where both the number of trees and the distortion are logarithmic. Our algorithm also extends to a certain probabilistic embedding guarantee. Lastly, we complement our upper bound with a lower bound showing that achieving a DAG cover with no distortion and $\tilde{O}(m)$ additional edges requires a \emph{polynomial} number of DAGs. 
\end{abstract}
\pagenumbering{gobble}

\newpage
 \tableofcontents

\vfill

\pagenumbering{gobble}
\pagebreak
 \pagenumbering{arabic}

\section{Introduction}

\emph{Probabilistic tree embedding}, first explicitly introduced by Alon, Karp, Peleg, and West~\cite{MR1313480}, and the related notion of \emph{tree cover}, are powerful graph primitives with wide-ranging applications. The goal of a probabilistic tree embedding is to embed an ($n$-vertex, $m$-edge) undirected graph into a distribution over trees with low expected distance distortion. (Formally, for each vertex pair $s,t$ in the original graph $G$, (1)  every tree $T$ in the distribution $\mathcal{D}$ satisfies $\dist_G(s,t)\leq \dist_T(s,t)$, and (2) $\mathbb{E}_{T\sim\mathcal{D}}[\dist_T(s,t)]\leq\alpha\cdot\dist_G(s,t)$ where $\alpha$ is the \emph{distortion}.) Tree covers provide a somewhat more relaxed guarantee wherein the goal is to obtain a small collection of trees so that for each pair $u,v$ of vertices, \emph{some} tree in the collection provides low distortion for $\dist(s,t)$.

Famously, building upon prior work~\cite{karp19892k,MR1313480,MR1450616,MR1731572}, Fakcharoenphol, Rao, and Talwar~\cite{MR2087946}
obtained the asymptotically optimal bound of $O(\log n)$ distortion for probabilistic tree embeddings. Fast computation of this construction has been extensively studied~\cite{MR3685766,mendel2009fast}, as well as extensions to
various settings such as 
dynamic~\cite{MR4262508},
online~\cite{MR1450616,MR4141277,bhore2024online},
distributed~\cite{khan2008efficient,MR3284084}, parallel~\cite{blelloch2012parallel,MR3882586},
hop-constrained~\cite{filtser2022hop,MR4398848},
and derandomized~\cite{MR1731567}.
Another line of work has focused on subgraph versions of probabilistic tree embedding, where the distribution is over subtrees of the input graph~\cite{elkin2008lower, abraham2008nearly,abraham2012using,abraham2020ramsey}. 
Probabilistic tree embedding and its many variations have enjoyed a strikingly broad range of applications such as k-median, buy-at-bulk network design~\cite{blelloch2012parallel},
generalized Steiner forest, minimum routing cost spanning tree, multi-source shortest paths~\cite{khan2008efficient}, distance oracles~\cite{MR3685766}, linear systems~\cite{MR3238960}, minimum bisection~\cite{MR3042131},
oblivious routing~\cite{MR2582666,harrelson2003polynomial,racke2002minimizing},
metric labeling~\cite{MR2145145}, and
group Steiner tree~\cite{MR1783249}.

A probabilistic tree embedding implies a tree cover with the same distortion and $O(\log n)$ trees, by sampling trees from the distribution. For tree cover there are also additional trade-offs between distortion and number of trees~\cite{MR1157580, awerbuch1994buffer,thorup2005approximate, MR3984837,MR4446771}: for any integer $k\geq 1$, one can achieve distortion $2k-1$ with $O(n^{1/k}\log^{1-1/k}n)$ trees. This result has applications, for instance, to routing~\cite{awerbuch1994buffer,MR1157580,awerbuch1991efficient} and
distributed directories~\cite{awerbuch1995online}.
Tree covers have also been fruitful for metrics with geometric structure such as Euclidean space, doubling metrics, ultrametrics, and planar graphs~\cite{chang2023covering,arya1995euclidean,MR4446771,MR3549603,MR2183280,MR2293956}.
Applications in these settings include, for instance, routing~\cite{gupta2005traveling,kahalon2022can}, 
spanners~\cite{arya1995euclidean,kahalon2022can,MR4490062}, 
and distance labeling~\cite{gupta2005traveling}. 

While both probabilistic tree embeddings and tree covers have experienced wide adoption, the success of these techniques has been limited to \emph{undirected} graphs. 
The contribution of the current work is to define and prove bounds for \emph{directed analogs} of these primitives. 

Before elaborating on the details, we discuss the wider context of this work.

\paragraph{Graph Simplification for Distance Preservation.}

This work falls under the umbrella of \emph{graph simplification} for distance preservation, where the goal is to obtain a simplified representation of a graph $G$ while (approximately) preserving its distances. In addition to probabilistic tree embeddings and tree covers, a number of other such structures exist, including spanners, emulators, distance preservers, hopsets, and distance oracles (see e.g.~the survey~\cite{MR4103322}). However, distances in directed graphs suffer from a relatively sparse toolkit of structural primitives. 
Each of the above structures can be defined for both undirected and directed graphs, however they are either provably nonexistent for directed graphs (e.g.~spanners, distance oracles), or generally less-understood (e.g.~distance preservers, hopsets). For hopsets on directed graphs, there has been recent significant progress~\cite{MR4415092, MR4538206, MR4720288, MR4699339, hoppenworth2025new}, though there still remain polynomial gaps between upper and lower bounds (in contrast to undirected graphs~\cite{ElkinN19,HuangP19}). 

Furthermore, this lack of structural primitives for directed graphs is perhaps a contributing factor to our gaps in understanding of basic algorithmic problems regarding distances and reachability in directed graphs. As a few examples, there are large gaps in our understanding of even the single-source reachability problem in various common settings such as streaming, distributed, and parallel; in addition, the following basic problems are less-understood for directed graphs than undirected: diameter approximation~\cite{MR3478403}, disjoint shortest paths~\cite{MR4262445}, not-shortest path~\cite{MR2095359}. See Section 1.3 of~\cite{MR4640631} for more detail. 
The current work can be viewed as a step towards remedying the lack of structural primitives for distances in directed graphs.

\paragraph{DAGs.}
Our goal is to define a directed analog of probabilistic tree embedding and tree cover. We will use directed acyclic graphs (DAGs) as our directed analog of trees. This is a natural choice, not only because DAGs are among the most natural directed analogs of trees,\footnote{Another natural directed analog of a tree is an arborescence, which is not meaningful in this context because a trivial lower bound shows that for a directed complete bipartite graph $A\rightarrow B$, one would need $n$ arborescences to preserve even the reachability for all pairs of vertices.} but also because many problems are much easier on DAGs than general directed graphs. This is important because a central purpose of graph simplification for distance preservation is to get better algorithms for distance-related problems, by first simplifying the graph, and then applying an algorithm to the simplified graph. This approach only works when the simplified graph indeed allows for more efficient algorithms. 

There are many examples of distance-related problems that are much easier on DAGs than general directed graphs, but we will specify a few. Single-source shortest paths with negative weights in DAGs in linear time is an undergraduate exercise, while for general directed graphs it  is a notorious problem that has witnessed several recent breakthroughs~\cite{MR4537239,MR4720280,MR4764788,huang2024faster}, and remains open. For distance preservers, the best known results on DAGs are better than those for general directed graphs \cite{bodwin2017linear}. For hopsets in DAGs, upper bounds for the simpler problem of shortcut  sets tend to extend quite easily, whereas this is not the case for general directed graphs~\cite{MR4415092,MR4538206}. 

\paragraph{Defining the problem.}
For simplicity, we will focus our initial discussion on directed analogs of tree cover (rather than probabilistic tree embedding).
Our goal is to obtain a small collection of DAGs such that for each ordered pair $s,t$ of vertices, some DAG in the collection provides low distortion for $\dist(s,t)$. 

Our first observation is that there is a trivial upper bound of $n$ DAGs with no distortion, because one can always return the collection of shortest paths DAGs, one from each of the $n$ vertices. Ultimately, we would like to do much better, and achieve a bound similar to what~\cite{MR2087946} yields: polylogarithmic bounds on both the distortion and the number of DAGs. 

Suppose we require each DAG in the collection to be a subgraph of the original graph. Then, a simple construction shows that no non-trivial upper bound is possible: Consider a directed cycle. Suppose $s$ and $t$ are adjacent vertices where $s$ appears right after $t$ on the cycle. Then, the only way to preserve the reachability from $s$ to $t$ is to include the unique path from $s$ to $t$ as a DAG in the collection; note that this path includes every edge in the cycle except $(t,s)$. Thus, to have finite distortion, the DAG collection must include such a path for all $n$ pairs $s,t$ of adjacent vertices on the cycle. %includes fewer than $n$ DAGs, there must exist such a pair $u,v$ where no DAG contains a path from $v$ to $u$, so the distortion is infinite.

Thus, we need to relax the problem to allow the DAGs to contain extra edges not in the original graph. If we allow the DAGs to have arbitrarily many edges, this trivializes the problem: There is a simple upper bound achieving only two DAGs and no distortion: Consider an arbitrary ordering of the vertices. Construct one DAG by including all $n \choose 2$ forward edges with respect to the ordering, where the weight of each edge $(u,v)$ is the distance $\dist(u,v)$ in the original graph. Construct the other DAG in the same way, but for the opposite ordering of vertices. This way, for each pair of vertices $s,t$, at least one of the two DAGs witnesses no distortion of the distance $\dist(s,t)$. Beyond trivializing the problem, allowing arbitrarily many extra edges complicates the graph, which is the opposite of our goal of graph simplification. 

Thus, we allow the DAGs to include a \emph{limited} number of edges not from the original graph. 
%This addition of edges is in a similar spirit to shortcut set and hopsets (albeit with a different goal). \nnote{I debated whether to include the previous sentence, not sure if it helps}
The number of additional edges can be expressed in terms of $n$ and/or $m$. %either in terms of $n$ or $m$; perhaps surprisingly, this distinction is crucial for what bounds are attainable. 
Now, we have arrived at the main definition of the problem:

\begin{definition}[DAG cover] Let $G$ be an $n$-node, $m$-edge weighted, directed graph. We define a DAG cover of $G$ with $f = f(m, n)$ additional edges and $g = g(m, n)$ DAGs as a collection $D_1, \dots, D_g$ of DAGs such that
\begin{itemize}
    \item each DAG $D_i$ is defined over the vertex set $V(G)$ of $G$, and
    \item the collection of DAGs contains at most $f$ additional edges not in $G$, i.e., $|(\cup_iE(D_i)) \setminus E(G)| \leq f$.\footnote{We could instead require that each DAG $D_i$ uses at most $f$ additional edges (i.e., $|E(D_i) \setminus E(G)| \leq f$). When $g$ is small, these two definitions are similar. Since we are most interested in the setting where $g = n^{o(1)}$, we ignore this distinction. } 
\end{itemize}

We say that a DAG cover $D_1, \dots, D_g$ of $G$ is \emph{reachability-preserving} if for all nodes $s, t \in V(G)$, $s$ can reach $t$ in $G$ if and only if there exists a DAG $D_i$ such that $s$ can reach $t$ in $D_i$.\\

We say a DAG cover $D_1, \dots, D_g$ of $G$ is \emph{$\alpha$-distance-preserving} if for all nodes $s, t \in V(G)$, 
$$
\dist_G(s, t) \leq \min_{i \in [g]} \dist_{D_i}(s, t) \leq  \alpha \cdot \dist_G(s, t).
$$
\end{definition}

\subsubsection*{Our Results.}

Our first result is essentially the strongest possible impossibility result for the regime where the number of additional edges is in terms of $n$: Even reachability-preserving DAG covers with almost $n^2$ additional edges require almost $n$ DAGs (where $n$ DAGs is a trivial upper bound):

\begin{restatable}{theorem}{fnlower}
        There exists a family of $n$-node directed graphs $G$ such that any reachability-preserving DAG cover of $G$ with at most $n^{2-o(1)}$ additional edges requires at least $n^{1 - o(1)}$ DAGs.  
    \label{thm:dag_cover_n_setting}
\end{restatable}

An implication of \cref{thm:dag_cover_n_setting} is to rule out a certain approach for obtaining hopsets from shortcut sets. \cite{MR4415092} obtained a breakthrough shortcut set construction and extended it to a hopset construction for DAGs with the same bounds; they also extended it to general directed graphs, but with worse bounds. Subsequently,~\cite{MR4538206} obtained an involved algorithm yielding a hopset for general directed graphs with bounds matching the original bound of~\cite{MR4415092}. If there existed a $(1+\eps)$-distance-preserving DAG cover with $\tilde{O}(n)$ additional edges and a polylogarithmic number of DAGs, then one could skip the involved algorithm of \cite{MR4538206}, and obtain this result as an immediate corollary of~\cite{MR4415092}. However, \cref{thm:dag_cover_n_setting} implies that such an approach is too good to be true.

\Cref{thm:dag_cover_n_setting} effectively rules out almost any non-trivial construction when the number of additional edges is in terms of $n$. However, as an aside, we show that the term ``almost" here is crucial; there is indeed a non-trivial construction when the number of edges is slightly less than $n^2$ additional edges: 
%We complement \cref{thm:dag_cover_n_setting} with an upper bound that with slightly less than $n^2$ additional edges, it is possible to achieve a sub-polynomial number of DAGs. 
\begin{restatable}{theorem}{fnupper}
    Every $n$-node directed graph $G$ admits an exact distance-preserving DAG cover of $G$ with  $O\left( \frac{n^2 (\log \log n)^4}{\log ^2 n } \right)$ additional edges and $n^{o(1)}$ DAGs. 
    \label{thm:f_n_upper}
\end{restatable}

Ultimately, the achievable bounds are unsatisfactory when the number of added edges is expressed in terms of $n$, so we turn our attention to expressing this quantity in terms of $m$. In particular, in line with our goal of graph simplification, we pay special attention to the natural regime of $\tilde{O}(m)$ additional edges, where the sparsity of the DAG collection is nearly the same as that of the original graph.\footnote{We use $\tilde{O}(\cdot)$ to hide logarithmic factors.} Perhaps surprisingly, this distinction between $n$ and $m$ is crucial for obtaining better bounds. We see a separation between the $n$ and $m$ regimes by observing that every graph $G$ admits a reachability-preserving DAG cover with $O(m)$ additional edges and only two DAGs.
\begin{observation}
    Every graph $G$ admits a reachability-preserving DAG cover with $O(m)$ additional edges and only two DAGs.
    \label{obs:m_reach}
\end{observation}
\begin{proof}[Proof sketch of Observation \ref{obs:m_reach}]
Let $[1, k]$ denote the set of strongly connected components (SCCs) in graph $G$.  For each SCC $i$ of $G$, pick an arbitrary ordering of its vertices, and let $f_i$ and $\ell_i$ be the first and last vertices in this ordering. Define one of the two DAGs as a directed path through each SCC according to its ordering, unioned with the following edges: for each SCC $i$, (1) an edge from $\ell_i$ to $f_j$ for all $j$ such that $G$ has an edge from SCC $i$ to SCC $j$, and (2) an edge to $f_i$ from $\ell_j$ for all $j$ such that $G$ has an edge from SCC $j$ to SCC $i$. Define the other DAG as a directed path through each SCC according to the reverse of its ordering. This is a reachability-preserving DAG cover because the first DAG preserves reachability for vertices in different SCCs, while for vertices in the same SCC, one of the two DAGs preserves reachability depending on which vertex comes first in the SCC's ordering. 
\end{proof}
\begin{figure}[H]
\begin{center}
    \includegraphics[width=0.80\textwidth]{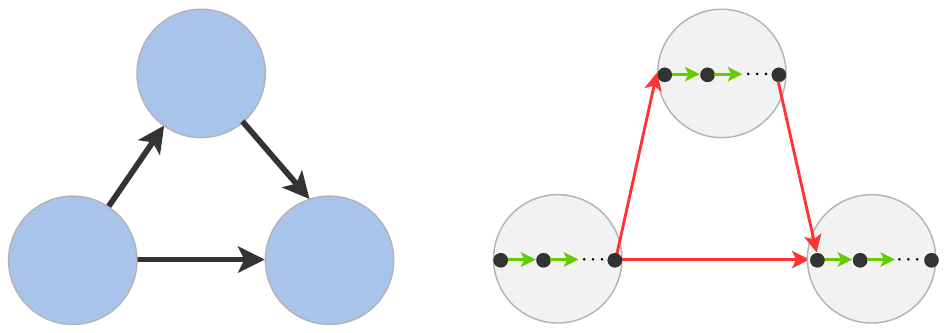}\caption{On the left is a directed graph $G$, with strongly connected components shown in blue. On the right is one of the two DAGs constructed in our reachability-preserving DAG cover of $G$ in Observation \ref{obs:m_reach}. The other DAG in our  cover can be obtained by reversing all the (green) edges between vertices in the same SCC of $G$, and removing all of the (red) edges between vertices in different SCCs of $G$. 
    }\label{fig:diam}
    \end{center}
\end{figure}

Thus, the interesting regime for $\tilde{O}(m)$ additional edges is (approximate)-distance-preserving DAG covers.
\begin{center}
\emph{Question: How many DAGs are required for $\alpha$-distance-preserving DAG covers with $\tilde{O}(m)$ additional edges? For which values of $\alpha$ is it polynomial versus polylogarithmic?}
\end{center}

We make progress on this question from both the upper and lower bound fronts: the answer is \emph{logarithmic} for certain polylogarithmic values of $\alpha$, and the answer is \emph{polynomial} for exact distances ($\alpha=1$). Our bounds are stated in~\cref{thm:dag_cover_intro,thm:lb_m}.

Our upper bound in~\cref{thm:dag_cover_intro} achieves similar bounds to what~\cite{MR1450616} yields for undirected graphs: polylogarithmic bounds on both the distortion and the number of DAGs. Furthermore, the fact that only $\tilde{O}(m)$ edges are added means that the sparsity of the DAG collection is asymptotically the same as that of the original graph, which satisfies our original goal of graph simplification. Additionally, such a DAG cover can be computed in near-linear time.

\begin{restatable}{theorem}{dagcover}
  \label{thm:dag_cover}
\label{thm:dag_cover_intro}
    Let $G$ be an $n$-node weighted, directed graph with polynomially-bounded positive integer edge weights. Then there exists an $O(\log^3n \cdot  \log\log n)$-approximate DAG cover with $O((m+n) \log^3 n)$ additional edges and $O(\log n)$ DAGs. Moreover, this DAG cover can be computed w.h.p.  in $\widetilde{O}(m)$ time.  
\end{restatable}

To complement our upper bound, and to show a separation between reachability versus exact shortest paths for $\tilde{O}(m)$ additional edges, we obtain the following polynomial lower bound on the number of DAGs.

\begin{restatable}{theorem}{lbm}\label{thm:lb_m}
There exists a family of $n$-node, $m$-edge  directed graphs $G$ with polynomially-bounded positive integer edge weights such that any exact distance-preserving DAG cover of $G$ with $m^{1+\varepsilon}$ additional edges  requires $\Omega(n^{1/6})$ DAGs, for a sufficiently small constant $\varepsilon > 0$.      
\end{restatable}

So far we have focused on DAG covers as an analog to tree covers, but there is also a corollary of \cref{thm:dag_cover_intro} that analogizes probabilistic tree embedding. Recall that the guarantee of a probabilistic tree embedding is that for each pair of vertices, the expected distortion of their distance over the distribution of trees, is bounded. It is impossible to achieve a precisely analogous guarantee for a distribution over DAGs because in any DAG, for every pair of vertices $s,t$, either $\dist(s,t)$ or $\dist(t,s)$ is infinite. Instead, we define a distribution $\mathcal{D}$ over DAGs so that for all pairs of nodes $(s,t)$ in the transitive closure of $G$, if a DAG is chosen from $\mathcal{D}$, the reachability from $s$ to $t$ is preserved with probability $1/2$, and conditioned on that, the expected distortion is polylogarithmic:
\begin{restatable}{theorem}{embed}\label{thm:embed}
      Let $G$ be an $n$-node weighted, directed graph with positive integer edge weights and polynomial aspect ratio. Then there exists a distribution $\mathcal{D}$ of DAGs over vertex set $V(G)$ satisfying the following properties:
      \begin{enumerate}
          \item  The support of distribution $\mathcal{D}$ contains at most $O((m+n) \log^3 n)$ additional edges.
          \item Let $D \sim \mathcal{D}$ be a DAG sampled from distribution $\mathcal{D}$. For all pairs of nodes $(s, t) \in TC(G)$, $s$ can reach $t$ in $D$ with probability $1/2$. Moreover,
          $$
            \mathbb{E}[\dist_D(s, t) \mid s \leadsto_D t] = O(\log^3 n \cdot \log \log n \cdot \dist_G(s, t)),
            $$
            where $s \leadsto_D t$ denotes the event that $s$ can reach $t$ in $D$. 
        \item We can sample a DAG from distribution $\mathcal{D}$ in $\widetilde{O}(m)$ time.
      \end{enumerate}
\end{restatable}

Finally, we return to the case of DAG covers with \emph{no} additional edges, for graphs with bounded diameter. We note that the above trivial lower bound of $n$ DAGs does not apply for graphs of low diameter, since the construction is simply a cycle of diameter $n-1$. In fact, there is a simple upper bound for low diameter graphs: an exactly distance-preserving DAG cover whose number of DAGs is logarithmic in $n$ but factorial in the diameter $D$. To see this, randomly order the vertices in $G$ and construct a DAG that includes every edge in $G$ that goes ``forward'' with respect to the ordering. A given shortest path is included in this DAG if its (at most $D+1$) vertices are ordered consistently with the random ordering, which occurs with probability $\geq 1/(D+1)!$. Then, the number of random trials (and thus number of DAGs) needed to include all $\Theta(n^2)$ shortest paths in some DAG with high probability is $O((D+1)! \cdot \log n)$.

Is exponential dependence on $D$ necessary? For our final result, we show that it is:

\begin{restatable}{theorem}{diam}\label{thm:diam}
    There exists a family of $n$-node directed graphs $G$ with diameter $O(\log n)$ such that any reachability-preserving DAG cover of $G$ with no additional edges requires $n$ DAGs.
\end{restatable}

\paragraph{Additional Related Work.}
A recently developing area of research focuses on understanding the \emph{structure} of shortest paths in both directed and undirected graphs~\cite{amiri2020disjoint,MR4695168, MR4640631, MR4430290, MR4520059,MR3909594,balzotti2022non,MR4481424}. See Section 1.4 of~\cite{MR4695168} for more details. Our work can be viewed as a contribution to this undertaking.

Additionally, the following papers are related to the wider context of our work: embedding directed planar graphs into directed $\ell_1$~\cite{MR4399708}, an extension of treewidth to DAG-width~\cite{MR2249395}, and shortest-path preservers with minimum aspect ratio~\cite{MR4695168}.

\section{Technical Overview}

\subsection{Lower Bounds} First we will outline some of the ideas behind~\cref{thm:dag_cover_n_setting} (our lower bound that any reachability-preserving DAG cover with almost $n^2$ additional edges requires almost $n$ DAGs), and then we will turn to~\cref{thm:lb_m} (our lower bound that any distance-preserving DAG cover with $m^{1+\varepsilon}$ additional edges requires $\Omega(n^{1/6})$ DAGs), 
which builds upon these ideas. 

\paragraph{Number of additional edges in terms of $n$.} Our starting point is a family of graphs whose variants are commonly used for obtaining lower bounds for related structures including hopsets/shortcut sets, spanners/emulators, and reachability/distance preservers, first used by~\cite{hesse2003directed}. Each graph in this family is a DAG that contains a \emph{large} collection of \emph{long} paths that are pairwise \emph{edge-disjoint}, and each path is the \emph{unique} path between its endpoints. 

Of course, such a graph does not yield a lower bound for DAG covers because the graph itself is already a DAG. Instead, the intuition is that we can modify the above ``base'' graph so that a constant fraction of pairs of intersecting paths are ``incompatible'', i.e.,~cannot appear together in the same DAG. To do so, we replace every vertex $v$ in the base graph with two vertices $v_1$, $v_2$ that have a bidirectional edge between them. Then, we randomize which of these two vertices constitute the entry and exit points of each path passing through $v$. That is, for every edge in the base graph that enters (respectively, exits) $v$, we assign it to enter (respectively, exit) either $v_1$ or $v_2$ with probability $1/2$ each. Then, for any pair of paths that intersect at $v$, the probability that they traverse the edge $(v_1,v_2)$ in opposite directions is 1/4. If this happens, the union of these two paths contains a cycle, making them incompatible. A probabilistic analysis of all intersections between paths yields a lower bound on the number of DAGs needed to accommodate every path. 

We are not done, however, because there are added edges. In particular, these added edges could bypass the $(v_1,v_2)$ bidirectional edge, allowing two paths that intersect at $v$ to become compatible. A key property of the base graph is that each pair of paths intersects on at most one vertex. This means that any added edge can only ``help'' one path. 
Furthermore, the parameters of the base graph are set so that the number of paths is a constant factor larger than the number of added edges. This means that a constant fraction of the paths are not helped by any added edge. Then, we can again perform a probabilistic analysis to calculate a bound on the number of DAGs needed to accommodate all of the non-helped paths. To obtain this bound, it is important that the paths are long enough (a property of the base graph) because for a DAG to accommodate a path, the path needs to randomly chose entry/exit points so that it is consistent with the DAG at \emph{every} step.

This analysis allows us to obtain essentially the strongest possible impossibility result for the regime where the number of additional edges is in terms of $n$: reachability-preserving DAG covers with almost $n^2$ additional edges require almost $n$ DAGs.

\paragraph{Number of additional edges in terms of $m$.} 
One issue towards obtaining a lower bound for $m^{1+\eps}$ added edges is the previously mentioned upper bound of a reachability-preserving DAG cover with two DAGs. That is, in contrast to the case of $n^{2-o(1)}$ added edges, any lower bound for $m^{1+\eps}$ added edges must hold only for exact or approximate shortest paths, and not for reachability. 

A related issue is that if the SCCs are of polylogarithmic size, there is an upper bound for exact-distance-preserving DAG covers: For each SCC $S$ and each vertex $v\in S$, add (appropriately weighted) edges from $v$ to the union of the out-neighbors of $S$, and add edges to $v$ from the union of in-neighbors of $S$. This DAG handles vertex pairs in two different SCCs, and a polylogarithmic number of additional DAGs handle vertex pairs in the same SCC. 

Therefore, to achieve a lower bound, the SCCs must be super-polylogarithmic in size. Specifically, each of our SCCs will be a clique of size $c = n^{\delta}$ 
with bidirectional edges. Similar to our construction for $n^{2-o(1)}$ added edges, every vertex in the original graph is replaced with a clique, and every incoming and outgoing edge to/from the clique picks a \emph{random} clique vertex to enter or exit, respectively. Now, a pair of paths is only incompatible if each traverses the same clique edge in opposite directions. This occurs much more rarely than in our previous construction with cliques of size 2, but we show that this construction can still be used towards getting a polynomial lower bound on the number of DAGs.

However, there is a major issue with the construction so far: in the previous construction it was important for the number of paths to be larger than the number of added edges (since each added edge could only ``help'' one path). However, since the paths in the base graph are edge-disjoint, the number of paths is necessarily less than the number $m^{1+\eps}$ of added edges. For this reason, we need to relax the edge-disjointness condition. This has been done in prior work  (first by~\cite{hesse2003directed}) for other problems (e.g.~for shortcut sets with $m$ added edges), by taking a certain type of \emph{product} of two copies of the base graph. We use such a product graph, which allows pairs of paths to overlap on a single edge (rather than just a vertex). Using this product graph as our new base graph, we perform the previously described clique-replacement step. However, the use of a product graph significantly complicates the proof for two reasons:

(1)  Since pairs of paths can overlap on an edge in the product graph, there are two types of added edges: ``long'' ones that only help one path as before, and ``short'' ones between adjacent cliques, which could help \emph{many} paths that all share an edge between those two cliques. 

(2) In our previous construction (for $n^{2-o(1)}$ added edges), we used the fact that each path picks its entry/exit points ($v_1$ or $v_2$) \emph{independently}. This is no longer the case. These choices are made independently for each edge, but since pairs of paths can share an edge, there are correlations between whether or not paths are consistent with a fixed DAG.

To circumvent this correlation issue, we restrict our attention to \emph{edge-disjoint} sets of paths going through the same clique. For such paths, their entry-exit points to/from the clique are chosen independently and we use this to calculate their probability of compatibility. Such a set of paths going through a given clique can be viewed as a \emph{matching} in an appropriate auxiliary graph. We are interested in lower bounding the size of such matchings over all of the cliques (see Section \ref{subsec:match} for details). 

We can think of this scenario as a game played with an adversary: First we randomly choose the trajectories of all of the paths through the cliques. Then the adversary, after seeing the random choices, is allowed to help paths using ``long'' edges, and help \emph{parts} of paths using ``short'' edges. The goal of the adversary is to minimize the matching sizes in the non-helped portions of the graph. 

Due to the intricacy of this scenario, to argue about the structure of the non-helped parts of the graph, we need to open the black box of the construction of the base graph and product graph, which are based on constructions from discrete geometry. In particular, we need to prove a special ``expansion'' property of these graphs (Lemmas \ref{lem:base_degree} and \ref{lem:prod_degree}).

\subsection{Upper Bound} In this section we will describe some of the ideas behind~\cref{thm:dag_cover_intro}, our $O(\log^3 n \cdot \log \log n)$-approximate DAG cover with $\tilde{O}(m)$ additional edges and $O(\log n)$ DAGs.
Our starting point is the technique of~\cite{MR1450616} for probabilistic tree embedding in undirected graphs, which is the elegant and well-known technique of a recursively applying a \emph{low-diameter decomposition} (LDD). In a low-diameter decomposition, we delete edges from the input graph according to a random process, so that afterwards each resulting connected component has bounded diameter, and the probability that any given path survives has a particular inverse relation to its length. Then to build a tree, we can add edges to connect these connected components into a star, where each edge weight is the diameter of the original graph (so that no distances are shortened), and we can recurse on each connected component.

LDDs for directed graphs are also known~\cite{MR4537239,bringmann2023negative,bringmann2025near}; in particular, we use Lemma 1.2 of~\cite{MR4537239}. Instead of each connected component having bounded diameter, each \emph{strongly} connected component (SCC) has bounded diameter. The other guarantee is similar to the undirected case: the probability that any given path survives has a particular inverse relation to its length. We recursively apply this LDD to each SCC. While this process works smoothly, the issue is that it is not clear how to use this recursive LDD to obtain a DAG cover. We highlight a few considerations: 

(1) In the undirected case, the above random process results in a single tree. In contrast, in any DAG at least half of all ordered vertex pairs have infinite distance, so we need to construct multiple DAGs. To do so, we need to establish a partition of the ordered vertex pairs to determine which distances will be preserved in each DAG we construct, and which distances will be neglected.

(2) Consider a vertex pair $(s,t)$ and a DAG that approximately preserves $\dist(s,t)$. The topological order of the DAG may not respect the order of vertices a given $st$-path, which could include many segments that are ``backwards'' with respect to the topological order of the DAG. To resolve this issue, we need to add a limited number of edges to \emph{bypass} these backwards segments for all such $s,t$.

(3) If we were to recurse on multiple SCCs to build our DAG collection, each SCC would return multiple DAGs (otherwise some distances would be infinite). Then, we would need to combine these DAGs in such a way that globally all distances are preserved (approximately). If we consider all possible ways to combine these DAGs, we risk a combinatorial explosion of combinations, leading to a large number of DAGs. To address this issue, instead of using recursion as a black box, we take a more global approach and define the process so that we construct a total of only two DAGs. (And then, similar to the undirected version, we repeat the random process $O(\log n)$ times to obtain the final collection of DAGs).

The rough idea of our algorithm is as follows. Consider a pair of vertices $(s,t)$, and consider the moment during the construction of the recursive LDD at which an edge is deleted, causing $\dist(s,t)$ to become infinite. There are two cases: (1) right before this moment, $s$ and $t$ were in the same SCC, or (2) they were in different SCCs. 
In case 1, if $s$ and $t$ were in the same SCC $S$, 
then the LDD gives us a bound on $\dist(s,t)$. Similar to the undirected version, we would like to add edges of weight diameter$(S)$ to create a star-like structure to approximately preserve $\dist(s,t)$. To make a star-like DAG for each SCC, we use a folklore shortcut set construction, which creates a DAG of diameter 2 with $\widetilde{O}(|S|)$ added edges.

In case 2, if $s$ and $t$ were in different SCCs, this means that $s$ appears before $t$ in the topological order of the evolving graph resulting from the recursive LDD. However, some edge on an $st$-path was deleted, disconnecting $t$ from $s$. We would like to bypass this edge, by adding a limited number of new edges that are consistent with the topological order emerging from the recursive LDD. We do this roughly via another star-like DAG for each SCC, as well as adding additional edges to link the SCCs together.  This linking of the SCCs is done in a roughly similar way to the previously described reachability-preserving DAG cover with $\tilde{O}(m)$ added edges. This is the part of the construction that requires $\tilde{O}(m)$ added edges, as opposed to $f(n)$ edges.

To approximately preserve all distances, we execute both of the above cases for every SCC and every level of the LDD. Importantly, we cannot construct a separate DAG for every level of the LDD because we need a low-distortion path from $s$ to $t$ to be contained in \emph{a single DAG}. That is, if an $st$-path requires multiple ``bypass'' edges from multiple levels of the LDD then all of these bypass edges need to appear in the \emph{same} DAG. A key challenge is to ensure that the union of \emph{all} of the added edges over all levels genuinely forms only two DAGs, and no cycles are formed from interference between edges from different levels.

\section{Notation}
\label{sec:prelims}
Throughout this paper, we will be working exclusively with directed graphs. For a directed graph $G$, we let $TC(G) \subseteq V(G) \times V(G)$ denote the transitive closure of  $G$. For a node $v \in V(G)$, we will use $\deg_G(v)$ to denote the number of edges incident to $v$ in $G$. Additionally, we will use $\text{indeg}_G(v)$ and $\text{outdeg}_G(v)$ to denote the the number of incoming edges and outgoing edges, respectively, that are incident to $v$ in $G$. 

We will frequently use $\pi$ to denote a directed path and $\Pi$ to denote a collection of paths. 
We will let $V(\pi)$ and $E(\pi)$ denote the nodes and edges of path $\pi$, respectively. Given a path $\pi$ and a set of nodes $S$, we will let $\pi \cap S$ denote the sequence of nodes in $V(\pi) \cap S$, listed in the order they appear in path $\pi$. We will interpret the ordered set $\pi \cap S$ as a (possibly non-contiguous) subpath of $\pi$ restricted to the nodes in $S$. Note that the edges $E(\pi \cap S)$ of subpath $\pi \cap S$ may not be a subset of the set edges $E(\pi)$ of the original path $\pi$; however, instead we can guarantee that the edges in $\pi \cap S$ are contained in the transitive closure of $\pi$, i.e.,  $E(\pi \cap S) \subseteq TC(\pi)$. 

Given a collection of nodes $V$ and a collection of paths $\Pi$ where the paths in $\pi$ may contain nodes that are not in $V$, we will define a corresponding graph which we denote as $\graph(V, \Pi)$ over the vertex set $V$. This graph will be the (non-disjoint) union of the set of paths in $\Pi$ restricted to the nodes in $V$. Formally, we let $\graph(V, \Pi)$ be the graph 
$$
\graph(V, \Pi) = \left(V, \quad \bigcup_{\pi \in \Pi} E(\pi \cap V) \right).
$$

\section{DAG covers with $f(n)$ additional edges}
In this section, we present a lower bound and an upper bound for DAG covers with $f(n)$ additional edges, where function $f(n)$ depends only on the number of nodes $n$ in the input graph. All of our bounds in this section are tight up to $n^{o(1)}$ factors.

\subsection{A lower bound for DAG covers with $n^{2-o(1)}$ additional edges}
\label{subsec:fn_lb}
The goal of this subsection is to prove the following theorem.
\fnlower*

Our lower bound construction will make use of the following lemma, which is immediate from existing lower bounds on reachability preservers in~\cite{MR3775910,MR4716717}.

\begin{lemma}[cf. Theorem 13 of \cite{MR4716717}]
\label{lem:rp_lbs}
There exists an $n$-node directed acyclic graph $G = (V, E)$ and a collection of directed paths $\Pi$ in $G$ with the following properties:
\begin{enumerate}
    \item $|\Pi| = n^{2-o(1)}$,
    \item each path $\pi \in \Pi$ has at least $|\pi| \geq 10\log n$ edges, 
    \item paths in $\Pi$ are pairwise edge-disjoint, and
    \item each path $\pi \in \Pi$ is the unique directed path between its endpoints in $G$. 
\end{enumerate}
\end{lemma}

Let $G = (V, E)$ and $\Pi$ be the $n$-node directed graph and associated collection of paths referenced in Lemma \ref{lem:rp_lbs}, respectively. We define a random graph $G^*$ as follows:
\begin{enumerate}
    \item Replace each node $v$ in $V$ with a copy of the directed clique graph $K_2$ with bidirectional edges. Denote the clique replacing node $v$ as $K_2^v$. 
    \item Fix an edge $(u, v) \in E$. Let $u^* \in V(K_2^u)$ be a node  sampled uniformly at random from $K_2^u$. Likewise, let $v^* \in V(K_2^v)$ be a node sampled uniformly at random from $K_2^v$. Add edge $(u^*, v^*)$ to graph $G^*$.  Repeat this random process for every edge in $E$. 
\end{enumerate}
This completes the construction of random graph $G^* = (V^*, E^*)$. We now define an associated set of paths $\Pi^*$ in $G^*$. Each path in $\Pi^*$ will be an image of a path in $\Pi$. Each edge $(u, v) \in E$ corresponds to a unique edge $(u^*, v^*)$ in $E^*$ between cliques $K_2^u$ and $K_2^v$. Let $\phi:E \mapsto E^*$ denote this injective mapping. Now fix a path $\pi \in \Pi$, and let $\pi = (v_1, \dots, v_k)$. For each edge $(v_i, v_{i+1}) \in E(\pi)$, let $e_i := \phi((v_i, v_{i+1}))$. Note that $e_i := (x_i, y_i) \subseteq V(K_2^{v_i}) \times V(K_2^{v_{i+1}})$, so $y_i, x_{i+1} \in V(K_2^{v_{i+1}})$, and in particular, either $y_i = x_{i+1}$ or $(y_i, x_{i+1}) \in E(K_2^{v_{i+1}})$. Then we define the image path $\pi^*$ in $G^*$ as 
$$
\pi^* = (x_1, y_1, x_2, y_2, \dots, x_i, y_i, \dots, x_{k-1}, y_{k-1}).
$$
We will use $\psi(\pi)$ to denote the image path $\pi^*$ of $\pi$ in $G^*$. Given a subset $\Pi' \subseteq \Pi$, we will use $\psi[\Pi']$ to denote the image of set $\Pi'$ under function $\psi$, i.e., $\psi[\Pi'] = \{\psi(\pi) \mid \pi \in \Pi'\}$.
We define $\Pi^*$ as the collection of image paths $\psi(\pi)$ of paths $\pi$ in $\Pi$, so
$\Pi^* = \psi[\Pi].$

We now define a family of DAGs $\mathcal{D}$ associated with our distribution of random graphs $G^*$ that will be useful in our analysis.

\paragraph{DAG Family $\mathcal{D}$.} We begin by defining a graph $\hat{G} = (V^*, \hat{E})$  over the same vertex set as $G^*$, and containing every possible edge in $G^*$. Formally, we define edge set $\hat{E}$ to be 
$$
\hat{E} = \{ V(K_2^u) \times V(K_2^v) \mid (u, v) \in E \} \cup \{ E(K_2^v) \mid  v \in V \}.
$$
Note that $G^* \subseteq \hat{G}$ for each graph $G^*$ in the support of our random graph distribution.

Our DAG family $\mathcal{D}$ will be the set of all edge-maximal acyclic subgraphs of $\hat{G}$. 
% For analysis purposes, it will be useful to define a family of DAGs $\mathcal{D}$ which roughly correspond to DAGs  total orderings of the vertex set $V(G^*)$. 
We define a DAG $D$ of $\hat{G}$ as follows. For each clique $K_2^v$, fix a total order on its $2$ vertices, and delete the edge in $K_2^v$ that does not respect this total order. After repeating this process for all cliques $K_2^v$, the resulting graph will be a DAG, since the original graph $G$ is a DAG. Define $\mathcal{D}$ to be the family of DAGs of $G^*$ that can be constructed from the above process. Note that since there are only $2$ choices of total orderings for each clique $K_2$, it follows that the size of $\mathcal{D}$ is 
$$
|\mathcal{D}| = 2^{|V|} = 2^n.
$$

For each DAG $D \in \mathcal{D}$, the graph $D \cap G^*$ is an acyclic subgraph of random graph $G^*$.\footnote{We use $D \cap G^*$ to  denote the graph with vertex set $V^*$ and edge set $ E(D) \cap E(G^*)$.} This allows us to reason about DAG covers of $G^*$ using the DAGs in graph family $\mathcal{D}$. We note that while $G^*$ is a random graph, the DAGs in $\mathcal{D}$ are fixed and depend only on graph $G$. 
The next claim  follows directly from the structure of graph $G^*$ as a DAG of strongly connected components $K_2^v$. 

\begin{claim}
For each acyclic subgraph $D$ of $G^*$, there is DAG $D' \in \mathcal{D}$ such that $D \subseteq  D' \cap G^*$.
    \label{clm:subdag}
\end{claim}
\begin{proof}
    Fix an acyclic subgraph $D$ of $G^*$. 
    DAG $D$ has an associated total ordering $\sigma$ of the vertices $V(G^*)$ of $G^*$. This total ordering induces a total ordering $\sigma[V(K_2^v)]$ on the vertices $V(K_2^v)$ for each $v \in V(G^*)$.    
    By our construction of family $\mathcal{D}$, there exists a DAG $D' \in \mathcal{D}$ such that $D'$ contains every edge in $E(K_2^v)$ that respects total order $\sigma[V(K_2^v)]$ for all $v \in V(G^*)$. Then for every edge $e \in E(D)$ with $e \in E(K_2^v)$ for some $v \in V(G^*)$, it follows that $e \in E(D') \cap E(G^*)$. Additionally, note that every edge in $E(G^*)$ between distinct strongly connected components $K_2^u$ and $K_2^v$ in $G^*$, where $u \neq v$, is preserved in $D'$. We conclude that $D \subseteq D' \cap G^*$. 
\end{proof}

We now verify that every DAG $D \in \mathcal{D}$ inherits the unique directed path property of our original graph $G$ (Property 4 of Lemma \ref{lem:rp_lbs}).

\begin{claim}
    Let $D \in \mathcal{D}$ be a DAG in $\mathcal{D}$. 
    % , and let $D' = D \cup E'$ be a graph obtained from $D$  after adding a set of additional edges $E'$ in the transitive closure of $D$, i.e., $E' \subseteq TC(D)$.  
    Fix  an $s^* \leadsto t^*$ path $\pi^* \in \Pi^*$. 
    % Suppose that $$E' \cap (V(\pi^*) \times V(\pi^*)) = \emptyset.$$
    Then   $s^*$ can reach  $t^*$ in $D \cap G^*$  only if $\pi^* \subseteq D \cap G^*$.  
    \label{clm:uniqueness}
\end{claim}
\begin{proof}
Recall that path $\pi^* \in \Pi^*$ is the  image of an $s \leadsto t$ path $\pi \in \Pi$. 
    This claim will follow straightforwardly from the property that path $\pi$ is the unique path between $s$ and $t$ in graph $G$. Consider an arbitrary $s^* \leadsto t^*$ path $\pi^*_1$ in $D \cap G^*$. Contract each clique $K_2^v$ in $G^*$ into a single node $v$. The resulting graph will be identical to our original graph $G$. Let $\pi'$ (respectively, $\pi_1'$) denote the image of path $\pi^*$ (respectively, path $\pi^*_1$)  after this contraction operation. Since $\pi$ is the unique $s \leadsto t$ path in $G$, we conclude that $\pi' = \pi'_1 = \pi$.

    As a consequence, we know that paths $\pi^*$ and $\pi_1^*$ in $D \cap G^*$  travel through the same sequence of cliques $K_2^v$ in $D \cap G^*$. 
    Let $K_2^{v_1}, \dots, K_2^{v_k}$ denote the sequence of cliques passed through by paths $\pi^*$ and $\pi_1^*$ in $G^*$, where $v_1 = s$ and $v_k = t$. By our construction of $G^*$, there is exactly one  edge between clique $K_2^{v_i}$ and $K_2^{v_{i+1}}$ in $G^*$. Formally,  
    $$
    |E(G^*) \cap (V(K_2^{v_i}) \times V(K_2^{v_{i+1}}))| = 1,
    $$
    for all $i \in [1, k-1]$. As a consequence, paths $\pi^*$ and $\pi_1^*$ must use exactly the same edge in $V(K_2^{v_i}) \times V(K_2^{v_{i+1}})$ between cliques $K_2^{v_i}$ and $K_2^{v_{i+1}}$. Finally, since paths $\pi^*$ and $\pi_1^*$ are contained in DAG $D$, we conclude that paths $\pi^*$ and $\pi_1^*$ are identical when restricted to any clique $K_2^{v_i}$. Formally,  $$\pi^* \cap V(K_2^{v_i}) =  \pi_1^*\cap V(K_2^{v_i}),$$
    for all $i \in [1, k]$. 
    We conclude that $\pi_1^* = \pi^*$, a contradiction.
\end{proof}

Our eventual goal is to argue that with high probability, for each DAG $D \in \mathcal{D}$, the DAG subgraph $D \cap G^*$ of $G^*$  cannot preserve reachability between many pairs of vertices in $G^*$. We begin with the following simple claim about the probability a DAG $D \cap G^*$ preserves a short path in $G^*$ (see \Cref{fig:clique_lb} for a visualization). 

\begin{figure}[h]
\centering
\includegraphics[width=0.75\textwidth]{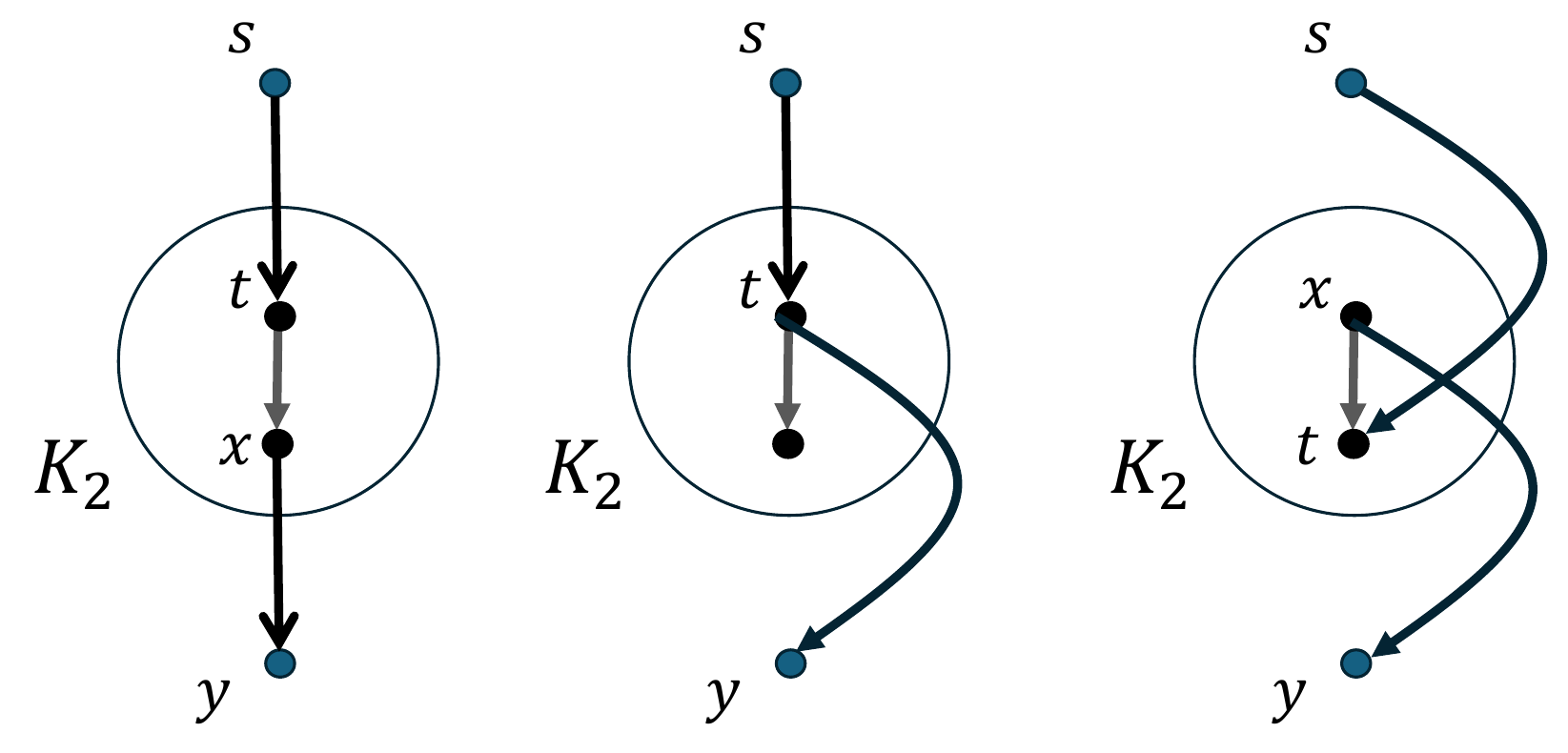} \label{fig:clique_lb}
\caption{A visualization of the different cases in \Cref{clm:small_paths}. These cases correspond to three possibilities of how edges entering and exiting a clique $K_2^v$ in $G^*$ are attached to vertices in $K_2^v$.}
\end{figure}

\begin{claim}
\label{clm:small_paths}
Fix a DAG $D \in \mathcal{D}$.  
Fix two edges $(u, v), (v, w) \in E$ in graph $G$, and let $\phi(u, v) = (s, t)$ and $\phi(v, w) = (x, y)$. 
Then path $\pi^* = (s, t, x, y)$ is contained in DAG $D \cap G^*$ with probability at most $3/4$.
\end{claim}
\begin{proof}
We observe that edges $(s, t)$ and $(x, y)$ are always contained in graph $D \cap G^*$.  Then path $\pi^*$ is contained in $D \cap G^*$ if and only if $t = x$ or $t <_D x$, where $<_D$ is the total order of DAG $D$.

Let $V(K_2^v) = {v_1^*, v_2^*}$, where $v_1^* <_D v_2^*$. Then the probability that $t = x$ is
$$
\Pr[t = x] = \Pr[t = x = v_1^* ] + \Pr[t = x = v_2^* ] = 1/2.
$$
Likewise, the probability that $t <_D x$ is
$$
\Pr[t <_D x] = \Pr[t = v_1^* \text{ and } x = v_2^*] = 1/4. 
$$
Then path $\pi^*$ is contained in $D \cap G^*$ with probability at most $3/4$.  
\end{proof}

In Claim \ref{clm:small_paths}, we showed that certain short paths in $G^*$ survive in $D \cap G^*$ for a fixed DAG $D \in \mathcal{D}$ with probability $3/4$. In the following claim, we will use Claim \ref{clm:small_paths} to argue that $D \cap G^*$ contains many paths in $\Pi^*$ with very low probability.

\begin{claim}
Fix a DAG $D \in \mathcal{D}$, and fix a subset $\Pi^*_1 \subseteq \Pi^*$ of size $|\Pi_1^*| = k$. The probability that $D \cap G^*$ contains all $k$ paths in $\Pi_1^*$ as subgraphs is at most $\left(\frac{3}{4}\right)^{10k \log n}$.  
\label{claim:dag_cover}
\end{claim}
\begin{proof}
Fix a path $\pi^* \in \Pi_1^*$ that is the image of a path $\pi \in \Pi$. Let $e_i$ be the $i$th edge of path $\pi$ for $i \in [1, |\pi|]$, and let $(u_i, v_i) = \phi(e_i)$. Fix an $i \in [1, |\pi|]$, and let $A_i$ be the event that path $(u_i, v_i, u_{i+1}, v_{i+1})$ is contained in DAG $D \cap G^*$.
We will need the following observations:
\begin{itemize}
    \item Path $\pi^*$ is contained in DAG $D \cap G^*$  only if event $A_i$ holds for every $i \in [1, |\pi|]$. 
    \item By Claim \ref{clm:small_paths}, $\Pr[A_i] \leq 3/4$. 
    \item Events $A_i$ and $A_j$ are independent when $|i-j| \ge 2$. 
\end{itemize}
Combining these observations, we get that
$$
\Pr[\pi^* \subseteq D \cap G^*] \leq \Pr[\cap_i A_i] \leq \left(\frac{3}{4}\right)^{|\pi|/2 } \leq \left( \frac{3}{4} \right)^{5 \log n},
$$
where the final inequality follows from the fact that $|\pi| \geq 10\log n$ by Lemma \ref{lem:rp_lbs}.
Also by Lemma \ref{lem:rp_lbs}, the paths in $\Pi$ are pairwise-edge disjoint. This implies that for distinct paths $\pi^*_1, \pi^*_2 \in \Pi^*$, the event that $\pi^*_1 \subseteq D \cap G^*$ is independent of the event that $\pi^*_2 \subseteq D \cap G^*$. We conclude that
$$\Pr[\text{$\pi^* \subseteq D \cap G^*$ for all $\pi^* \in \Pi_1^*$}] = \prod_{\pi^* \in \Pi_1^*}\Pr[\pi^* \subseteq D \cap G^*] \leq \left( \frac{3}{4} \right)^{5k \log n}.$$
\end{proof}

We will now combine the very strong probability bound proved in Claim \ref{claim:dag_cover} with a union bound to argue that no DAG subgraph $D \cap G^*$, where $D \in \mathcal{D}$, can contain many paths in $\Pi^*$ as subgraphs.

\begin{lemma}
    With nonzero probability, for each DAG $D \in \mathcal{D}$, graph $D \cap G^*$ contains at most $n$ distinct paths in $\Pi^*$ as subgraphs. 
    \label{lem:path_coverage}
\end{lemma}
\begin{proof}
By \cref{claim:dag_cover} and the union bound, the probability that there exists a DAG $D \in \mathcal{D}$ and a set $\Pi' \subseteq \Pi^*$ of size $|\Pi'| = k$ such that each path in $\Pi'$ is a subgraph of $D \cap G^*$ is at most
$$
|\mathcal{D}| \cdot {|\Pi^*| \choose k} \cdot \left( \frac{4}{3} \right)^{-5k \log n} \leq 2^n \cdot n^{2k} \cdot \left( \frac{4}{3} \right)^{-5k \log n} < 1,
$$
when $k \geq n$.  Note that if a graph $D \cap G^*$ contains greater than $k$ distinct paths in $\Pi^*$ as subgraphs, then it must contain a set $\Pi' \subseteq \Pi^*$ of size $|\Pi'| = k$, so the claim follows. 
\end{proof}

From now on, we treat graph $G^*$ as a fixed graph that satisfies \Cref{lem:path_coverage}.  This is without loss of generality, since \Cref{lem:path_coverage} holds with nonzero probability. Intuitively, it should be difficult to construct a DAG cover for graph $G^*$ because any DAG subgraph of $G^*$ can only preserve at most $n$ paths in $\Pi^*$ by \Cref{lem:path_coverage}.  
We are ready to prove \cref{thm:dag_cover_n_setting}. 

\fnlower*
\begin{proof}
Let $D_1, \dots, D_g$ be a reachability-preserving DAG cover of $G^*$ with at most $|\Pi^*|/2 = n^{2-o(1)}$ additional edges. Let $E'$ be the set of additional edges used in the DAG cover. 
From the definition of a DAG cover, we require that $E' \subseteq TC(G^*)$. 
In particular, we may assume without loss of generality $E' \subseteq TC(G^*) \setminus E(G^*)$. This implies that every edge $e \in E'$ is between distinct clique subgraphs $K_2^u$ and $K_2^v$ in $G^*$, where $u \neq v$. 

We define a subset $\Pi_1^* \subseteq \Pi^*$ of paths in $\Pi^*$ as follows:
$$
\Pi_1^* = \{ \pi^* \in \Pi^* \mid TC(\pi^*) \cap E' = \emptyset  \}.
$$
We observe that $|\Pi_1^*| \geq |\Pi^*| - |E'| \geq |\Pi^*|/2$, since every edge $e$ in $E'$ is contained in $TC(\pi^*)$  for at most one path $\pi^* \in \Pi^*$ by Property 3 of Lemma \ref{lem:rp_lbs}, and since $|E'| \le |\Pi^*|/2$.

Fix an $i \in [1, g]$. By \Cref{clm:subdag}, there exists a DAG $D_i' \in \mathcal{D}$ such that $D_i - E' \subseteq D_i' \cap G^*$. Additionally, fix  an  $s^* \leadsto t^*$ path $\pi^* \in \Pi_1^*$  such that $s^*$ can reach $t^*$ in $D_i$.  Since $\pi^* \in \Pi_1^*$, we know that $s^*$ can reach $t^*$ in $D_i - E'$.
Since $D_i - E' \subseteq D_i' \cap G^*$, we know that $s^*$ can reach $t^*$ in $D_i' \cap G^*$ as well. Finally, by \Cref{clm:uniqueness}, this implies that path $\pi^*$ is contained in graph $D_i' \cap G^*$, i.e., $\pi^* \subseteq D_i' \cap G^*$.

We have shown that for each $i \in [1, g]$, there exists a DAG $D_i' \in \mathcal{D}$ such that for every $s^* \leadsto t^*$ path $\pi^* \in \Pi_1^*$, if $s^*$ can reach $t^*$ in DAG $D_i$,  then path $\pi^*$ is a subgraph of  $D_i' \cap G^*$. However, by \Cref{lem:path_coverage}, there are at most $n$ distinct paths in $\Pi^*$ that are subgraphs of $D_i' \cap G^*$. We conclude that each DAG $D_i$ in our DAG cover preserves the reachability between the endpoints of at most $n$ distinct paths in $\Pi^*_1$. 

There are $|\Pi_1^*| = |\Pi^*|/2 = n^{2-o(1)}$ distinct paths in $\Pi_1^*$. Then the number of DAGs in our DAG cover is at least
$$
g \geq \frac{|\Pi_1^*|}{n} = \frac{n^{2-o(1)}}{n} = n^{1-o(1)}.
$$
% By Claim \ref{clm:subdag}, DAG $D_i$ 
% $$
% D_i \subseteq (D_i' \cap G^*) \cup E',
% $$
% where $D_i' \in \mathcal{D}$  and $E'_i \subseteq E'$. By Claim \ref{clm:uniqueness}, $s^*$ can reach $t^*$ in $D_i'$ only if $\pi^* \subseteq D_i'$.
% Equivalently, $\pi^*$ is the unique $s^* \leadsto t^*$ path in $D_i'$. 
% We will need the stronger claim that $s^*$ can reach $t^*$ in $D_i = D_i' \cup E'_i$ only if $\pi^* \subseteq D_i'$. We make the following two observations: 
% \begin{enumerate}
%     \item  We may assume without loss of generality that $E_i' \subseteq TC(D_i')$, since DAG $D_i' \in \mathcal{D}$ is edge-maximal. 
%     \item Since $\pi^* \in \Pi_1^*$, we know that $TC(\pi^*) \cap E_i' = \emptyset$.
% \end{enumerate}
% By the first observation, every $s^* \leadsto t^*$ path in $D_i$ must contain only vertices in $V(\pi^*)$, by Claim \ref{clm:uniqueness}. By the second observation, there are no edges in $E_i'$ between vertices in $\pi^*$. We conclude that $\pi^*$ remains the unique $s^* \leadsto t^*$ path in $D_i$. In particular, we must have that $\pi^* \subseteq D_i'$.
 % By Lemma \ref{lem:path_coverage} and the above discussion, every DAG $D_i$ in our DAG cover $\{D_i\}_{i \in [1, g]}$ contains at most $n$ distinct paths in $\Pi_1^*$. 
\end{proof}

\subsection{An upper bound for DAG covers with $O\left( \frac{n^2(\log \log n)^4}{\log^2 n} \right)$ additional edges}

The goal of this section is to prove the following theorem.

\fnupper*

We will rely on the following theorem about the existence of small exact hopsets due to \cite{berman2010finding}.

\begin{lemma}[Theorem 5.1 of  \cite{berman2010finding}]
\label{lem:hopset}
    Let $G$ be an $n$-node directed, weighted graph. For every $d \in [1, n]$, there exists a collection of edges $H \subseteq TC(G)$ satisfying the following properties:
    \begin{enumerate}
        \item $|H| = O(n^{2}/d^2)$, and
        \item for every pair of nodes $(s, t) \in TC(G)$, there exists a $s \leadsto t$ path in $G \cup H$ with at most $d$ edges. 
    \end{enumerate}
\end{lemma}

With this lemma, we are ready to prove Theorem \ref{thm:f_n_upper}. 

\begin{proof}[Proof of Theorem \ref{thm:f_n_upper}]
    Given an input graph $G$ on $n$ nodes, construct a collection of edges $H$ as in Lemma \ref{lem:hopset} with parameter $d = \frac{\log n}{(\log \log n)^2}$. Let $\Sigma$ denote the collection of all total orders on the vertices $V(G)$ of $G$. Let $\sigma_1, \dots, \sigma_k \in \Sigma$ be  total orders of $V(G)$ sampled from $\Sigma$ uniformly with replacement.
    We choose $k$ to be $k = (d+1)! \cdot 10 \log n$. 
    For each $\sigma_i$, let $D_i \subseteq G \cup H$ be the DAG obtained by keeping only the edges in $G\cup H$ that respect total order $\sigma_i$. 

    Fix a pair of nodes $(s, t) \in TC(G)$. Let $\pi$ be an $s \leadsto t$ shortest path in $G \cup H$ with at most $d$ edges. Such a path must exist by Lemma \ref{lem:hopset}.  Path $\pi$ survives in DAG $D_i$ if total ordering $\sigma_i$ respects the ordering of the $|\pi|+1 \leq d+1$ vertices of path $\pi$. Then path $\pi$ survives in $D_i$ with probability at least $\frac{1}{(d+1)!}$. Then the probability that $\pi$ does not survive in any of our $k$ total orders is at most
    $$
    \left( 1 - \frac{1}{(d+1)!} \right)^k = \left( 1 - \frac{1}{(d+1)!} \right)^{(d+1)! \cdot 10 \log n} \leq e^{-\frac{(d+1)!}{(d+1)!} \cdot 10 \log n} \leq n^{-5}.
    $$
    By union bounding over all $O(n^2)$ pairs of nodes $(s, t) \in TC(G)$, we conclude that DAG cover $\{D_i\}_{i \in [1, k]}$ is an exact distance-preserving DAG cover of $G$, with high probability. Finally, we observe that our DAG cover contains 
    $$
    k = (d+1)! \cdot 10 \log n = 2^{\Theta(d \log d)} 10 \log n = 2^{O\left(\frac{\log n}{\log \log n}\right)} 10 \log n = n^{o(1)}
    $$
    DAGs.
\end{proof}

% \newpage

\section{An upper bound for DAG covers with $\widetilde{O}(m)$ additional edges}

The goal of this section is to prove the following theorem.

\dagcover*

An essential tool we will need to prove Theorem \ref{thm:dag_cover} is the directed low-diameter decomposition introduced in~\cite{MR4537239} and improved recently in~\cite{bringmann2025near}. The key properties of this low-diameter decomposition are summarized in Theorem 2.1 of \cite{bringmann2025near}, which we restate below. 

\begin{lemma}[Theorem 2.1 of \cite{bringmann2025near}]
\label{lem:ldd}
    Let $G = (V, E, w)$ be an $n$-node weighted, directed graph with positive integer edge weights. Let $d$ be a positive integer such that $d = O(\text{\normalfont poly}(n))$. Then there exists a set of edges $E' \subseteq E$ with the following properties:
    \begin{enumerate}
        \item each SCC of $G \setminus E$ has (weak) diameter at most $d$ in $G$. 
        \item For every edge $e \in E$, 
        $$
        \Pr[e \in E'] \leq \alpha \cdot  \frac{w(e)}{d},
        $$
        where $\alpha = O(\log n \cdot \log \log n)$. 
        \item  The set $E'$ and the SCCs of $G \setminus E'$ can be computed in $\widetilde{O}(m)$ time.
    \end{enumerate}
\end{lemma}

For the remainder of this section, we use $\alpha$ to denote the coefficient $\alpha$ in property 2 of \Cref{lem:ldd}.

\subsection{Construction of our DAG Cover $\mathcal{D}$}
\label{subsec:const_D}
Let $G$ be an $n$-node, weighted, directed graph with positive integer edge weights and maximum weight $W$. The key step of our construction will be a hierarchical decomposition of $G$ obtained by repeatedly applying Lemma \ref{lem:ldd}. 
Let $G_0 = G$, and let $\mathcal{F}_0$ denote the collection of all strongly connected components in $G_0$. We will construct set family $\mathcal{F}_{i+1}$ and graph $G_{i+1}$ from  graph $G_i$, as follows. 
\begin{enumerate}
    \item Apply Lemma \ref{lem:ldd} to graph $G_i$ with integer parameter $d_i = nW \cdot 2^{-i-1}$. Let $E_i \subseteq E(G_i)$ be the set of edges specified by Lemma \ref{lem:ldd}.
    \item Let $G_{i+1} = G_i - E_i$, and let $S^1_{i+1}, \dots, S^{k_{i+1}}_{i+1} \subseteq V(G)$ denote the SCCs of $G_{i+1}$. We let $\mathcal{F}_{i+1}$ denote the collection of SCCs of graph $G_{i+1}$.
    \item Let $z = \lceil \lg (n W) \rceil$.
    We will terminate our recursion after computing graph $G_z$ and collection $\mathcal{F}_z$. 
\end{enumerate}
Note that since $d_z < 1$, the graph $G_z$ will be a DAG. We will use $D^*$ to denote graph $G_z$. The following claim summarizes the properties of the sets $\mathcal{F}_i$ inherited from Lemma \ref{lem:ldd}.

\begin{claim}
The following properties of the $\mathcal{F}_i$'s and $E_i$'s are inherited from Lemma \ref{lem:ldd}:
\begin{enumerate}
    \item Set family $\mathcal{F}_i$ is the collection of SCCs of graph $G_i$. 
    \item Each set $S \in \mathcal{F}_i$ has weak diameter at most $2d_i$ in $G_{i-1}$ (and consequently, in $G$ as well). 
    \item For every edge $e \in E(G_i)$, 
    $$
    \Pr[e \in E_i] \leq \alpha \cdot \frac{w(e)}{d_i}
    $$
\end{enumerate}
\label{clm:fam}
\end{claim}

We will need several preliminary notations and claims before we can construct the DAGs associated with the set family of $\mathcal{F}_i$'s. 
Fix an $i \in [0, z]$. We define a total order $<_i$ on the SCCs $\mathcal{F}_i$ in $G_i$. Informally, this will correspond to a specific topological order of the condensation graph of $G_i$.  
Fix two SCCs  $S, T \in \mathcal{F}_i$, where $S \neq T$. 
If there exists a node $s \in S$ and a node $t \in T$ such that $s$ can reach $t$ in graph $G_i$, then we let $S <_i T$. If there exists distinct sets $S', T' \in \mathcal{F}_{i-1}$ satisfying $S \subseteq S'$ and $T \subseteq T'$ such that $S' <_{i-1} T'$, then we let $S <_i T$, as well. 
We observe that if there exists distinct $S, T \in \mathcal{F}_i$ such that $S <_iT$ and $T <_i S$, then that implies a contradiction of Property 1 of Claim \ref{clm:fam}. We conclude that $<_i$ is a partial order on the elements of $\mathcal{F}_i$. Finally, we convert partial order $<_i$ to a total order by completing the ordering arbitrarily. 
\begin{claim}
    Relation $<_i$ is a total ordering of set family $\mathcal{F}_i$.
    \label{clm:total_order}
\end{claim}

We will now define a total order $<_D$ on the vertex set $V(G)$ that respects every total order $<_i$ for all $i \in [0, z]$. For each node $v \in V(G)$ and each $i \in [0, z]$, we let $S_i(v) \in \mathcal{F}_i$ denote the unique set in $\mathcal{F}_i$ containing node $v$. Given $s, t \in V(G)$, we will let $s <_D t$ if 
$$
s \neq t \qquad \text{ and } \qquad S_i(s) <_i S_i(t),
$$
for some $i \in [0, z]$. We will let $<_D$ be an arbitrary total order that satisfies this property for all $s \neq t \in V(G)$. We will use $s \leq_D t$ to denote that either $s <_D t$ or $s = t$. 
\begin{claim}
    Relation $<_D$ is a total ordering of vertex set $V(G)$. 
\end{claim}
\begin{proof}
    Suppose towards contradiction that there exist nodes $ s, t \in V(G)$, with $s \neq t$, such that $s <_D t$ and $t <_D s$. Then without loss of generality, $S_i(s) <_i S_i(t)$ and $S_j(t) <_j S_j(s)$ for some $i \leq j$. Since $S_j(s) \subseteq S_i(s)$ and $S_j(t) \subseteq S_i(t)$ and $S_i(s) <_i S_i(t)$, we conclude that $S_j(s) <_j S_j(t)$ by the definition of $<_j$. This contradicts Claim \ref{clm:total_order}, so we conclude that $<_D$ is a total ordering of $V(G)$. 
\end{proof}

For each set $S \in \mathcal{F}_i$, we  associate two \textit{representative} nodes $r_1(S), r_2(S) \in S$.
We choose $r_1(S)$ to be the first node in $S$ under total order $<_{D}$, and we choose $r_2(S)$ to be the last node in $S$ under total order $<_{D}$. 
For each node $v \in V(G)$ and each $i \in [0, z]$, we define the level $i$ representative nodes of $v$, denoted as $r_i^1(v)$ and $r_i^2(v)$, to be 
$$
r_i^1(v) = r_1(S_i(v))  \quad \text{ and } \quad r_i^2(v) = r_2(S_i(v)).
$$
In particular, when $i = z$, we have that $r_z^1(v) = r_z^2(v) = v$, since $D^*$ is a DAG.  
For each edge $(u, v) \in E(G)$, we define the \textit{level} of edge $(u, v)$, denoted as $\ell(u, v)$,  to be the smallest integer $i \in [0, z]$ such that $S_i(u) \neq S_i(v)$.

In addition to the above notation, we will need the following claim, which roughly states that we can approximately preserve distances between  nodes $s, t$ in an SCC $S \in \mathcal{F}_i$ when $s <_{D} t$, using a small collection of edges $H$ that respect total order $<_D$.

\begin{claim}
\label{clm:hopedge}
    Let $S \in \mathcal{F}_i$ be a set in family $\mathcal{F}_i$. Then there exists a set of directed, weighted edges $H \subseteq S \times S$ of size $|H|=O(|S| \log |S|)$ with the following properties:
    \begin{enumerate}
    \item Every edge $(u, v)$ in $H$ respects total order $<_{D}$.
    \item Every edge $(u, v)$ in $H$  has weight $w(u, v) = 2d_i$. 
    \item For all $s, t \in S$ with $s <_{D} t$,
    $$
     \dist_G(s, t) \leq \dist_{H}(s, t) \leq 4d_i.
    $$
    \item  Set $H$ can be computed in $O(|H|)$ time.
    \end{enumerate}
\end{claim}
\begin{proof}
Let $S = \{s_1, \dots, s_k\}$ be the nodes in $S$, ordered with respect to $<_D$. 
We define a path $\pi$ that respects $<_D$ as follows:
$$
\pi = (s_1, \dots, s_k).
$$
We assign each edge in $\pi$ weight $d_i$. Initially, we let $H = E(\pi)$. 
By \cite{MR4415092}, for a directed path $\pi$ on $k$ nodes, there exists a collection of directed edges $H'$ such that 
\begin{itemize}
    \item $|H'| = O(k \log k)$ and $H'$ can be computed in $O(|H'|)$ time, 
    \item $H'$ is contained in the transitive closure of path $\pi$, and
    \item For every pair of nodes $s, t$ in path $\pi$ such that node $s$ occurs on $\pi$ before node $t$, there exists an $s \leadsto t$ path $\pi'$ in $\pi \cup H'$ with $|\pi'| \leq 2$ edges. 
\end{itemize}
Then we let $H = E(\pi) \cup H'$, with every edge in $H$ assigned weight $2d_i$. Properties 1, 2, and 4 of our claim are clearly satisfied. We now verify Property 3 for all $s, t \in S$ with $s <_D t$:
\begin{itemize}
    \item By Property 2, $\dist_H(s, t) \geq 2d_i$. Moreover, $2d_i \geq \dist_G(s, t)$ by Claim \ref{clm:fam}, so $\dist_G(s, t) \leq \dist_H(s, t)$. 
    \item For each pair of nodes $s, t$ in path $\pi$ such that $s <_D t$, there exists an $s \leadsto t$ path $\pi'$ in $H$ with $|\pi'| \leq 2$ edges. Then $\dist_H(s, t) \leq w(\pi') \leq 4d_i$, by Property 2. \qedhere
\end{itemize}
\end{proof}

We  construct two DAGs $D_1$ and $D_2$ associated with our  collection of   $\mathcal{F}_i$'s. DAG $D_1$ will respect total order $<_{D}$, while DAG $D_2$  will respect the reverse of $<_D$. 
We  construct DAG $D_1$ as follows.

\paragraph{Construction of DAG $D_1$.}
\begin{enumerate}
    \item The vertex set of $D_1$ will be the same as $G$, i.e., $V(D_1) = V(G)$. 
    \item For each edge $(u, v) \in E(D^*)$, each $i, j \in [0, z]$ such that $\min(i, j) \geq \ell(u, v)$, and each $k, k' \in \{1, 2\}$, add the edge $(r_i^{k}(u), r_j^{k'}(v))$ to $D_1$. Assign  edge $(r_i^{k}(u),  r_j^{k'}(v))$  weight $w_{G}(u, v) + 2d_i + 2d_j$.  Since $r_z^1(v) = r_z^2(v) = v$ for all $v \in V(G)$, this procedure also adds edges of the form $(u, r_i^j(v))$ and $(r_i^j(u), v)$, for all $i \geq \ell(u, v)$ and $j \in \{1, 2\}$.      Additionally, we add the edges in $E(D^*)$ with their original weights in $G$ to DAG $D_1$.

    \item Fix an index $i \in [0, z]$, and fix a set $S_i^j \in \mathcal{F}_i$, where $j \in [1, k_{i}]$. Let $H_i^j \subseteq S_i^j \times S_i^j$ be the set of directed, weighted edges specified in Claim \ref{clm:hopedge} with respect to set $S_i^j$. We will add the set of  edges
    $$
    H = \bigcup_{i \in [0, z], j \in [1, k_i]} H_i^j
    $$
    to $D_1$. 
\end{enumerate}
This completes the construction of $D_1$. If there are parallel edges in $D_1$ from node $u$ to node $v$, then we keep only the lowest weight edge from $u$ to $v$. We will quickly verify that  $D_1$ is a DAG.
\begin{claim}
   Graph $D_1$ is a DAG. 
   \label{clm:dag}
\end{claim}
\begin{proof}
Note that every edge in $H$ respects total order $<_{D}$ by Claim \ref{clm:hopedge}. We will now show that every edge we add to $D_1$ in Step 2 respects $<_{D}$. Fix an edge $(r_i^{x}(u), r_j^{y}(v))$ added to $D_1$ in Step 2, for some $i \leq j \in [0, z]$, $x, y \in \{1, 2\}$, and $(u, v) \in E(D^*)$.  Since $i \geq \ell(u, v)$, it follows that $S_i(u) \neq S_i(v)$. Additionally, since edge $(u, v) \in E(D^*) \subseteq E(G_i)$, it follows that
$S_i(u) <_i S_i(v)$. Then since $r_i^x(u) \in S_i(u)$ and $r_j^{y}(v) \in S_j(v) \subseteq S_i(v)$, we conclude that $$r_i^x(u) <_D r_j^{y}(v),$$ as desired.
\end{proof}

\paragraph{Construction of DAG $D_2$.}
\begin{enumerate}
    \item The vertex set of $D_2$ will be the same as $G$, so $V(D_2) = V(G)$. 
    \item Fix an index $i \in [0, z]$, and fix a set $S_i^j \in \mathcal{F}_i$, where $j \in [1, k_{i}]$. Let $H_i^j \subseteq S_i^j \times S_i^j$ be the set of directed, weighted edges specified in Claim \ref{clm:hopedge} with respect to set $S_i^j$. Let $H \subseteq V(G) \times V(G)$ be the directed, weighted set of edges
    $$
    H = \bigcup_{i \in [0, z], j \in [1, k_i]} H_i^j.
    $$
    Let $H^R$ denote the set of weighted edges obtained by reversing the orientation of every edge in $H$. We add the edges in $H^R$ to $D_2$.
\end{enumerate}
This completes the construction of graph $D_2$. Note that for every edge $(u, v)$ in graph $D_2$, we have that $(v, u) \in E(D_1)$. 
Then by Claim \ref{clm:dag}, graph $D_2$ is also a DAG. 

\paragraph{Construction of DAG Cover $\mathcal{D}$.}
We will construct our DAG cover $\mathcal{D}$ by repeating the following (random) procedure $10 \log n$ times:
\begin{enumerate}
    \item Construct a collection of sets $\{\mathcal{F}_i\}_{i\in [0, z]}$. (Note that each set  $\mathcal{F}_i$ is constructed randomly and inherits
    the probabilistic guarantees of Lemma \ref{lem:ldd}.)
    \item Construct DAGs $D_1$ and $D_2$ using the collection of sets $\{\mathcal{F}_i\}_{i\in [0, z]}$, and add $D_1$ and $D_2$ to $\mathcal{D}$. 
\end{enumerate}

\subsection{Size and Time Analysis of our DAG Cover $\mathcal{D}$}

In this section, we prove that our DAG cover $\mathcal{D}$ has the size claimed in Theorem \ref{thm:dag_cover}, and we prove that our DAG cover can be constructed in the time claimed in  Theorem \ref{thm:dag_cover}. 

\paragraph{Size Analysis.} In Step 2 of the construction of DAG $D_1$, we add at most $O(z^2)$ edges to $D_1$ for every edge $(u, v) \in E(D^*) \subseteq E(G)$. 
Then Step 2 of the construction of DAG $D_1$ contributes $O(mz^2) = O(m \log^2(nW))$ edges to $D_1$. 
By Claim \ref{clm:hopedge}, the third step of the construction of DAG $D_1$ contributes at most
$$
 |H| \leq \sum_{i \in [0, z], j \in [1, k_i]} \left| H_i^j \right| \leq O(\log n) \cdot \sum_{i \in [0, z]} \sum_{j \in [1, k_i]} |S_i^j| = O(\log n) \cdot \sum_{i \in [0, z]} n = O(n \log(nW) \log n)
$$
edges. Then we add at most
$$
|E(D_1)| + |E(D_2)| =  O(m \log^2 n + n\log^2 n)
$$
edges to $D_1$ and $D_2$, when $W = O(\text{\noindent poly}(n))$. 
Since we repeatedly construct random DAGs $D_1$ and $D_2$ exactly $10\log n$ times, the number of additional edges and DAGs claimed in Theorem \ref{thm:dag_cover} is established.

\paragraph{Running Time Analysis.} To construct our hierarchical decomposition $\{\mathcal{F}\}_i$, we make $z$ calls to Lemma \ref{lem:ldd}. This takes $\widetilde{O}(\log(nW) \cdot m)$ time.
The second step of the construction of DAG $D_1$ takes $O(mz^2) = O(m \log^2 (nW))$ time. To construct all sets of edges $H_i^j$, where $i \in [0, z]$ and $j \in [1, k_i]$, it takes time $O(\sum_{i \in [0, z], j \in [1, k_i]} | H_i^j |) = O(n  \log (nW) \log n)$ by Claim \ref{clm:hopedge} and our earlier bound on $\cup_{i, j}|H_i^j|$ in our size analysis. 
We conclude that the constructions of DAGs $D_1$ and $D_2$ take $\widetilde{O}(m)$ time when $W = O(\text{\noindent poly}(n))$. Since we repeatedly construct random DAGs $D_1$ and $D_2$ exactly $10\log n$ times, the running time claimed in Theorem \ref{thm:dag_cover} is established.

\subsection{Initial Distortion Analysis}

In this section, we make some initial progress towards proving that our DAG cover $\mathcal{D}$ approximately preserves distances in $G$.  First, we verify that for every DAG $D \in \mathcal{D}$, distances in $D$ are at least distances in $G$. 

\begin{claim}
    For each DAG $D \in \mathcal{D}$ and nodes $s, t \in V(G)$,
    $$
    \dist_G(s, t) \leq \dist_D(s, t). 
    $$
\end{claim}
\begin{proof}
    We will prove that for every edge $(u, v) \in E(D)$, 
    $$
    \dist_G(u, v) \leq w_D(u, v),
    $$
    where $w_D$ is the weight function associated with DAG $D$. 
    This will immediately imply the stated claim. We split our proof into cases based on which step of the construction we added edge $(u, v)$  to DAG $D$.
    \begin{itemize}
        \item Edge $(u, v)$ was added in Step 2 of the construction of DAG $D_1$. In this case, edge $(u, v)$ is of the form
        $$
        (u, v) = (r_i^x(u^*), r_j^{y}(v^*)),
        $$
        for some $i, j \in [0, z]$ such that $\min(i, j) \geq \ell(u^*, v^*)$, $x, y \in \{1, 2\}$ and edge $(u^*, v^*) \in E(D^*)$.
        Recall that $w_D(u, v) = w_{G}(u^*, v^*) + 2d_i + 2d_j$. 
        Then
        \begin{align*}
            \dist_G(u, v) & = \dist_G(r_i^x(u^*), r_j^{y}(v^*)) \\
            & \leq \dist_G(r_i^x(u^*), u^*) + \dist_G(u^*, v^*) + \dist_G(v^*, r_j^{y}(v^*)) & \text{ by the triangle inequality} \\
            & \leq \dist_G(r_i^x(u^*), u^*) + w_G(u^*, v^*) + \dist_G(v^*, r_j^{y}(v^*)) \\
            & \leq 2d_i + w_G(u^*, v^*) + 2d_j & \text{ by \Cref{clm:fam}} \\
            & = w_D(u, v),
        \end{align*}
        as desired.
        \item Edge $(u, v)$ was added in Step 3 of the construction of DAG $D_1$ or in Step 2 of the construction of DAG $D_2$. We will need the following observations:
        \begin{itemize}
            \item Since edge $(u, v) \in H \cup H^R$, there exists a set $S_i^j \in \mathcal{F}_i$, where $i \in [0, z]$ and $j \in [1, k_i]$, such that $u, v \in S_i^j$. 
            \item By Property 2 of Claim \ref{clm:hopedge}, $w(u, v) = 2d_i$. 
            \item By  Claim \ref{clm:fam}, set $S_i^j$ has weak diameter at most $2d_i$ in $G$.
        \end{itemize}
        Then
        \[
        \dist_G(u, v) \leq 2d_i  = w_D(u, v). \hfill \qedhere
        \]
    \end{itemize}
\end{proof}

We have shown that distances in our DAG cover are at least distances in $G$. 
We now establish the distortion upper bound guarantee of our DAG cover. 
Let  $D_1$ and $D_2$ be the random DAGs constructed in subsection \ref{subsec:const_D}, with associated set family $\{\mathcal{F}_i\}_i$ and graph family $\{G_i\}_i$. 
The key step in establishing our distortion upper bound will be to prove the following lemma.
\begin{restatable}{lemma}{expected}
    For all $s, t \in V(G)$,
    $$
    \mathbb{E}[\min(\dist_{D_1}(s, t), \dist_{D_2}(s, t))] = O(\log^3n \cdot \log \log n) \cdot \dist_G(s, t).
    $$
    \label{lem:exp_err}
\end{restatable}

Once we prove this lemma, our claimed distortion upper bound in Theorem \ref{thm:dag_cover} will follow from a simple application of Markov's inequality.  The remainder of this section and section \ref{subsec:induction} will be devoted to developing the necessary claims and lemmas to prove Lemma \ref{lem:exp_err}; we will finally prove Lemma \ref{lem:exp_err} in section \ref{subsec:final}.

First, we will show that paths in graph $G_i$ have a natural structure with respect to the SCCs $\mathcal{F}_i$ of $G_i$. 

\begin{claim}
    \label{clm:path_struct}
    Let $\pi$ be an $s \leadsto t$ path in $G_i$ for some $s, t \in V(G)$ and $i \in [0, z]$. There exists a sequence of sets
    $S_1, \dots, S_k \subseteq \mathcal{F}_i$ in set family $\mathcal{F}_i$ with the following properties:
\begin{enumerate}
    \item Path $\pi$ contains exactly one edge $e_j = (u_j, v_j)$ of the form $e_j \in S_{j} \times S_{j+1}$ for all $j \in [1, k-1]$. Moreover,
    \begin{itemize}
        \item  $\pi[s, u_1] \subseteq S_1$,
        \item $\pi[v_j, u_{j+1}] \subseteq S_{j+1}$ for all $j \in [1, k-2]$, and 
        \item $\pi[v_{k-1}, t] \subseteq S_{k}$. 
    \end{itemize}
    \item  $S_{j_1} \neq S_{j_2}$ for all $j_1, j_2 \in [1, k]$ such that $j_1 \neq j_2$.
\end{enumerate}
\end{claim}
\begin{proof}
Recall that $\mathcal{F}_i$ is the set of all strongly connected components in $G_i$.  Let $u, v \in \pi$ be nodes such that $u$ comes before $v$ in path $\pi$. Suppose that $u, v \in \pi \cap S$ for some $S \in \mathcal{F}_i$. This implies that $\pi[u, v] \subseteq S$ because $S$ is an SCC in $G_i$. Then for every SCC $S \in \mathcal{F}_i$, set $S$ intersects path $\pi$ in a (possibly empty) contiguous subpath of $\pi$. Let sequence $S_1, \dots, S_k \subseteq \mathcal{F}_i$  be the SCCs in $\mathcal{F}_i$ that have nonempty intersection with $\pi$, written so that if $j_1 < j_2 \in [1, k]$ then subpath $\pi \cap S_{j_1}$ comes before subpath $\pi \cap S_{j_2}$ in path $\pi$. 
\end{proof}

We will now show that if nodes $s$ and $t$ are contained in the same set $S$ in set family $\mathcal{F}_i$, then we can use the edges in $H$ and $H^R$ to upper bound the distances between $s$ and $t$ in $D_1$ and $D_2$.

\begin{claim}
    \label{clm:single_scc}
    Fix an index $i \in [0, z]$
    and nodes $s, t \in S$ for some set $S \in \mathcal{F}_i$ in set family $\mathcal{F}_i$.
    \begin{itemize}
        \item If $s <_D t$, then
        $$
        \dist_{D_1}(s, t) \leq  4d_i.
        $$
        \item if $s >_D t$, then
        $$
        \dist_{D_2}(s, t) \leq 4d_i.
        $$
    \end{itemize}
    \end{claim}
    \begin{proof}
        This claim will follow directly from Claim \ref{clm:hopedge} and our construction of DAGs $D_1$ and $D_2$. Let $H_S \subseteq S \times S$ be the directed, weighted edges specified in Claim \ref{clm:hopedge} with respect to set $S \in \mathcal{F}_i$.
        By construction, $H_S \subseteq E(D_1)$. 
        \begin{itemize}
            \item If $s <_D t$, then 
        $$
        \dist_{D_1}(s, t) \leq \dist_{H_S}(s, t) \leq 4d_i,$$
        by Property 3 of Claim \ref{clm:hopedge}. 
        \item Otherwise, $s >_D t$. Let $H_S^R$ be the set of weighted, directed edges obtained by reversing the orientations of the edges in $H_S$. By construction, $H_S^R \subseteq H^R \subseteq E(D_2)$. Then since $s >_D t$, we can again argue by Property 3 of  Claim \ref{clm:hopedge} that
        \[
        \dist_{D_2}(s, t) \leq  \dist_{H_S^R}(s, t) =   \dist_{H_S}(t, s) \leq 4d_i.\hfill\qedhere
        \] 
        \end{itemize}     
    \end{proof}

For an event $A$, let $\mathbbm{1}[A]$ denote the random variable that takes value $1$ when event $A$ holds, and value $0$ otherwise. The following lemma makes some progress towards proving Lemma \ref{lem:exp_err}, by achieving an upper  bound  on the expected distortion between $s$ and $t$ in DAG $D_2$, multiplied by the binary random variable $\mathbbm{1}[s >_D t]$. Notice that if $s <_D t$, then $\dist_{D_2}(s, t) = \infty$, so we cannot expect to get a reasonable upper bound on $\mathbb{E}[\dist_{D_2}(s, t)]$ in general; this is why we upper bound $\mathbb{E}[\dist_{D_2}(s, t) \cdot \mathbbm{1}[s >_D t]]$ instead.

\begin{lemma}
\label{lem:s_geq_t}
For all $s, t \in V(G)$,
    $$
\mathbb{E}[\dist_{D_2}(s, t) \cdot \mathbbm{1}[s >_D t]] = O(\log^2 n \cdot \log \log n) \cdot \dist_G(s, t).
$$
\end{lemma}
\begin{proof}
     Fix a shortest $s \leadsto t$ path $\pi$ in $G$.
     Let $X$ be the random variable $X = \dist_{D_2}(s, t)$. 
     Let $A$ be the event that $s >_D t$. Let $B_i$ be the event that path $\pi$ is contained in graph $G_i$, for all $i \in [0, z]$.

     Consider the scenario where event $A \cap B_i$ occurs, i.e., $s >_D t$, and path $\pi$ is contained in $G_i$. Then there must exist an SCC $S \in \mathcal{F}_i$ of $G_i$ such that $\pi \subseteq G_i[S]$. If this is not the case, then path $\pi$ intersects with two or more SCCs in $G_i$, implying that $s <_D t$, a contradiction. Then we can apply Claim \ref{clm:single_scc} to argue that if event $A \cap B_i$ occurs, then 
    $
    \dist_{D_2}(s, t) \leq 4d_i.
    $
    In particular, this implies that $$\mathbb{E}[X \mid A \cap B_i \cap \overline{B}_{i+1}] \leq 4d_i.$$

    Observe that $\Pr[B_0] = 1$ since $\pi \subseteq G = G_0$.  Then we can obtain the following inequality
    $$
    \mathbb{E}[X \mid A] \leq  \sum_{i = 0}^{z}  \mathbb{E}[X \mid A \cap B_i \cap \bar{B}_{i+1}] \Pr[B_i \cap \bar{B}_{i+1} \mid A].
    $$
    Additionally, we observe that by Claim \ref{clm:fam} and the union bound, for all $i \in [0, z]$,
    $$
    \Pr[\Bar{B}_{i+1} \mid B_{i}] \leq 
     \sum_{e \in E(\pi)} \alpha \cdot  \frac{w(e)}{d_i}  =  \frac{\alpha}{d_i}  \cdot \dist_G(s, t),
    $$
    since $\pi$ is an $s \leadsto t$ shortest path in $G$. 
    Putting our observations together, 
    \begin{align*}
        \mathbb{E}[X \cdot \mathbbm{1}[s >_D t]] & = \mathbb{E}[X \mid A]\Pr[A] \\
         & \leq \left( \sum_{i = 0}^{z}  \mathbb{E}[X \mid A \cap B_i \cap \bar{B}_{i+1}] \Pr[B_i \cap \bar{B}_{i+1} \mid A] \right)  \Pr[A]  \\
        & = \sum_{i = 0}^{z}  \mathbb{E}[X \mid A \cap B_i \cap \bar{B}_{i+1}] \Pr[B_i \cap \bar{B}_{i+1} \cap A]   \\
        & \leq \sum_{i = 0}^{z}  \mathbb{E}[X \mid A \cap B_i \cap \bar{B}_{i+1}] \Pr[ \bar{B}_{i+1} \mid B_i]   \\
        & \leq \sum_{i = 0}^{z}  2d_i \cdot \Pr[ \bar{B}_{i+1} \mid B_i]   \\  
        & \leq \sum_{i = 0}^{z}  2d_i \cdot \frac{\alpha}{d_i} \cdot \dist_G(s, t) \\
        & = (z+1) \cdot 2\alpha \cdot \dist_G(s, t)\\ &= O(\log^2 n \cdot \log\log n) \cdot \dist_G(s, t).\hfill\qedhere
    \end{align*}    
\end{proof}

We will need the following technical lemma, which we prove using a similar argument as in Lemma \ref{lem:s_geq_t}. 
Recall that for any index $X \in [0, z]$, $d_X = nW2^{-X-1}$ is the diameter of the directed low-diameter decomposition that we apply to graph $G_X$. This lemma roughly states that an $(s, t)$-shortest path $\pi$ survives to diameter $\widetilde{O}(1) \cdot \dist_G(s, t)$ in expectation.

\begin{lemma} \label{lem:d_X}
Let $s, t \in V(G)$ be nodes such that $s$ can reach $t$ in $G$. Let $\pi$ be an $s \leadsto t$ shortest path in $G$. Let $H \subseteq G$ be a subgraph of $G$ such that $\pi \subseteq H$.  For each $i \in [0, z]$ and subgraph $J \subseteq G$, let $C_i^J$ be the event that $G_i = J$.  Let $X \in [0, z]$ be the (random-valued) index such that $\pi \subseteq G_X$ and $\pi \not \subseteq G_{X+1}$. For each $i \in [0, z]$,
     \begin{align*}
     \mathbb{E}\left[ d_X |  C_i^H \right]& \leq  ( z - i + 1) \cdot \alpha \cdot \dist_G(s, t), 
    \end{align*}
where $\alpha = O(\log n \cdot \log \log n)$ is the coefficient from the directed LDD in \Cref{lem:ldd}.
\end{lemma}
\begin{proof}

Fix a pair of nodes $s, t \in V(G)$ such that $s$ can reach $t$ in $G$. Let $\pi$ be an $s \leadsto t$ shortest path in $G$, and let $H \subseteq G$ be a subgraph of $G$ such that $\pi \subseteq H$. We will prove Lemma \ref{lem:d_X} for a fixed index $i \in [0, z]$.

For each $j \in [0, z]$, let $B_j$ be the event that $\pi \subseteq G_j$. 
We can obtain the following equality:
    $$
    \mathbb{E}[d_X \mid C_i^H] =  \sum_{j = i}^{z}  d_j \cdot  \Pr[B_j \cap \bar{B}_{j+1}  \mid  C_i^H].
    $$
    Additionally, we observe that by Claim \ref{clm:fam} and the union bound, for all $j \in [i, z]$,
    $$
    \Pr[\Bar{B}_{j+1} \mid B_j \cap C_i^H] \leq 
     \sum_{e \in E(\pi)}  \frac{w(e) \cdot \alpha}{d_j}   =  \frac{\alpha}{d_j} \cdot \dist_G(s, t),
    $$
    since $\pi$ is an $s \leadsto t$ shortest path in $G$.  Putting our observations together,
    \begin{align*}
        \mathbb{E}[d_X \mid C_i^H ] & = \sum_{j=i}^z d_j \cdot \Pr[B_j \cap \bar{B}_{j+1} \mid  C_i^H] \\
         & =  \sum_{j = i}^{z}  d_j \cdot  \Pr[\bar{B}_{j+1} \mid B_{j} \cap C_i^H] \Pr[B_j \mid  C_i^H]  \\  
         & \le  \sum_{j = i}^{z}  d_j \cdot  \Pr[\bar{B}_{j+1} \mid B_{j} \cap C_i^H]  \\ 
        & \leq \sum_{j = i}^{z} d_j \cdot  \frac{\alpha}{d_j}  \cdot \dist_{G}(s, t)   \\
        & \leq  (z-i+1) \cdot  \alpha \cdot \dist_G(s, t).
    \end{align*}     
\end{proof}

In Lemma \ref{lem:s_geq_t}, we gave an upper bound on the expectation of random variable $\dist_{D_2}(s, t) \cdot \mathbbm{1}[s >_D t]$, for a pair of nodes $s, t \in V(G)$. This proof made heavy use of the fact that this random variable took value $0$ when $s <_D t$. Our next goal will be to give an upper bound on the expectation of the random variable $\dist_{D_1}(s, t) \cdot \mathbbm{1}[s <_D t]$. 
We will need to introduce some additional notation first.

\paragraph{Additional Notation.} 
For all nodes $s, t \in V(G)$ such that $s$ can reach $t$ in $G$ (i.e., $(s, t) \in TC(G)$), we fix an $s \leadsto t$ shortest path $\pi_{s, t}$ in $G$. We may assume without loss of generality that our collection of paths $\{\pi_{s, t}\}_{(s, t) \in TC(G)}$ is \textit{consistent}. Formally, this means that for any two paths $\pi_{s, t}$ and $\pi_{s', t'}$ in our collection, and for any two vertices $u, v \in V(G)$, if $u, v \in \pi_{s, t} \cap \pi_{s', t'}$ and node $u$ precedes node $v$ in both $\pi_{s, t}$ and $\pi_{s', t'}$, then we have that $\pi_{s, t}[u, v] = \pi_{s', t'}[u, v]$. 

For all nodes $s, t \in V(G)$ such that $s$ can reach $t$ in $G$, and for each $i \in [0, z]$, let $B_i^{s, t}$ be the event that path $\pi_{s, t}$ is contained in $G_i$. For every subgraph $H \subseteq G$, let $C_i^H$ be the event that $G_i = H$. 
\\

The following lemma formally states what we will prove about the expected size of the random variable $\dist_{D_1}(s, t) \cdot \mathbbm{1}[ s <_D t]$.

\begin{restatable}{lemma}{induction}
    Let $s, t \in V(G)$, and let $i \in [0, z]$. Let $H \subseteq G$ be a subgraph of $G$ such that $\pi_{s,t } \subseteq H$. Then
$$
\mathbb{E}\left[\dist_{D_1}(s, t) \cdot \mathbbm{1}[s <_D t] \mid  C_i^H  \right] \leq c \cdot (z-i+1)^2 \cdot \alpha \cdot \dist_G(s, t), 
$$
for a sufficiently large constant $c > 1$.
\label{lem:induction}
\end{restatable}

We will prove Lemma \ref{lem:induction} by an induction argument in section \ref{subsec:induction}.  After  proving  Lemma \ref{lem:induction}, we finish the proof of Theorem \ref{thm:dag_cover} in section \ref{subsec:final}.

\subsection{Induction proof of Lemma \ref{lem:induction}}

\label{subsec:induction}

In this section we will prove Lemma \ref{lem:induction} using an induction argument. Recall that for all nodes $s, t \in V(G)$ such that $s$ can reach $t$ in $G$, and for each $i \in [0, z]$,  $B_i^{s, t}$ denotes the event that path $\pi_{s, t}$ is contained in $G_i$, where $\pi_{s, t}$ is a $s \leadsto t$ shortest path in $G$ specified in the previously. Likewise, for every subgraph $H \subseteq G$,  $C_i^H$ denotes the event that $G_i = H$. 

\induction*

We will prove Lemma \ref{lem:induction} by induction on $k = z-i+1$, starting with $k=1$ and incrementing to $k = z+1$. Let $c>1$ denote the induction constant claimed in Lemma \ref{lem:induction}. We will specify the precise value of $c$ when we finish the induction step of our proof in Lemma \ref{lem:inductive_step}. 
For each $k \in [1, z+1]$, we define an induction statement $S_k$ as follows. 

\paragraph{Induction Statement $S_k$.} \textit{Let $i = z-k+1$.  Let $s, t \in V(G)$. Let $H \subseteq G$ be a subgraph of $G$ such that $\pi_{s, t} \subseteq H$. Then}
$$
\mathbb{E}\left[\dist_{D_1}(s, t) \cdot \mathbbm{1}[s <_D t] \mid  C_{i}^H  \right] \leq c \cdot k^2 \cdot \alpha \cdot \dist_G(s, t).
$$

We can prove statement $S_1$ with relative ease. 

\begin{claim}[Base Case: $S_1$] \label{clm:base_induction}
    Induction Statement $S_1$ is true.   
\end{claim}
\begin{proof}
     When $k = 1$, we have that $i = z-k+1 = z$. If event $C_i^H$ holds, then path $\pi_{s, t}$ is contained in graph $H = G_z = D^*$. Since we added the edges in $E(D^*)$ to DAG $D_1$ with their original edge weights, it follows that
    $$
    \dist_{D_1}(s, t) \leq \dist_{D^*}(s, t) = \dist_G(s, t).
    $$
    We conclude that 
    $$
    \mathbb{E}[\dist_{D_1}(s, t) \cdot \mathbbm{1}[s <_D t] \mid C_z^H] \leq \mathbb{E}[\dist_{D_1}(s, t)\mid  C_z^H] \leq  \dist_G(s, t) \leq c \cdot \alpha \cdot \dist_G(s, t),
    $$
    as desired.
\end{proof}

For the remainder of this section, we will let $k$ and $i$ be fixed indices such that $k \in [2, z+1]$ and $i = z-k+1 \in  [0, z-1]$. Additionally, for the remainder of this section we will assume that Induction Statement $S_{k-1}$ holds.
By the end of this section, we will prove  Lemma \ref{lem:inductive_step}, which states that if $S_{k-1}$ is true, then $S_k$ is true as well.  

\begin{restatable}{lemma}{inductivestep} \label{lem:inductive_step}
    If statement $S_{k-1}$ holds, then statement $S_k$ holds as well. 
\end{restatable}

Proving Lemma \ref{lem:inductive_step} will complete the inductive step of our proof of Lemma \ref{lem:induction}.
The remainder of this section will be devoted towards eventually proving  Lemma \ref{lem:inductive_step}.

\paragraph{Proving Lemma \ref{lem:inductive_step}.} 
Our  goal is to prove the upper bound  $$\mathbb{E}[\dist_{D_1}(s, t) \cdot \mathbbm{1}[s <_D t] \mid  C_i^H] \leq c \cdot k^2 \cdot \alpha \cdot \dist_G(s, t) $$ 
claimed in statement $S_k$, completing the proof of Lemma \ref{lem:inductive_step}. 
As a warm up, we will prove that statement $S_k$ holds in the special case where $s$ and $t$ lie in the same SCC in $H$. 

\begin{lemma} \label{lem:warmup}
Let $s, t \in V(G)$, and let $H \subseteq G$ be a subgraph of $G$ such that $\pi_{s,t} \subseteq H$ and nodes $s$ and $t$ lie in the same SCC in $H$.  If statement $S_{k-1}$ holds, then
    $$
    \mathbb{E}[\dist_{D_1}(s, t) \cdot \mathbbm{1}[s <_D t] \mid  C_i^H] \leq (c \cdot (z-i)^2 +c_1)  \cdot \alpha \cdot \dist_G(s, t),
    $$
    where $c_1 \geq 1$ is a sufficiently large universal constant that does not depend on the induction constant $c$. In particular, if $c \geq c_1$, then it follows that  
    $$
    \mathbb{E}[\dist_{D_1}(s, t) \cdot \mathbbm{1}[s <_D t] \mid  C_i^H] \leq c \cdot k^2  \cdot \alpha \cdot \dist_G(s, t).
    $$
\end{lemma}
\begin{proof}
 We define a random variable $X$ to be
$$X = \dist_{D_1}(s, t) \cdot \mathbbm{1}[s <_D t].$$

We will make use of the following equality in our proof.
\begin{multline} \label{warm up main ineq}
    X = \dist_{D_1}(s, t) \mathbbm{1}[s <_D t] = \dist_{D_1}(s, t) \mathbbm{1}[s <_D t] (\mathbbm{1}[B_{i+1}^{s, t}] + \mathbbm{1}[\bar{B}_{i+1}^{s, t}]) \\
    = \dist_{D_1}(s, t) \mathbbm{1}[s <_D t] \mathbbm{1}[B_{i+1}^{s, t}] + \dist_{D_1}(s, t) \mathbbm{1}[s <_D t] \mathbbm{1}[\bar{B}_{i+1}^{s, t}]
\end{multline}
Let $Y_1 = \dist_{D_1}(s, t) \mathbbm{1}[s <_D t] \mathbbm{1}[B_{i+1}^{s, t}]$, and let $Y_2 = \dist_{D_1}(s, t) \mathbbm{1}[s <_D t] \mathbbm{1}[\bar{B}_{i+1}^{s, t}]$. Since $X = Y_1 + Y_2$, we can upper bound $\mathbb{E}[X \mid  C_i^H]$ by upper bounding $\mathbb{E}[Y_1 \mid  C_i^H]$ and $\mathbb{E}[Y_2 \mid C_i^H]$ separately. 

We begin by upper bounding $\mathbb{E}[Y_1 \mid C_i^H]$. We will achieve this by applying the induction statement $S_{k-1}$. We observe the following sequence of inequalities:
    \begin{align*}
        &  \mathbb{E}[Y_1 \mid  C_i^H]  \\
        & = \mathbb{E}[ \dist_{D_1}(s, t) \cdot \mathbbm{1}[s <_D t] \cdot \mathbbm{1}[B_{i+1}^{s, t}] \mid C_i^H] & \\
        & \leq  \mathbb{E}[ \dist_{D_1}(s, t) \cdot \mathbbm{1}[s <_D t]  \mid  B_{i+1}^{s, t} \cap C_i^H]. &  
    \end{align*}

We are almost ready to apply the induction hypothesis. Let $\mathcal{H}$ denote the family of all subgraphs $H' \subseteq H$ of $H$ such that  $\pi_{s, t} \subseteq H'$. 
Let $H^* \in \mathcal{H}$ be the graph in $\mathcal{H}$ that maximizes $\mathbb{E}[\dist_{D_1}(s, t) \cdot \mathbbm{1}[s <_D t] \mid C_{i+1}^{H^*}]$, i.e., let
$$
H^* = \text{\noindent argmax}_{H' \in \mathcal{H}}  \mathbb{E}\left[\dist_{D_1}(s, t) \cdot \mathbbm{1}[s <_D t] \mid  C_{i+1}^{H'}\right].
$$

Notice that if event  $B_{i+1}^{s, t} \cap C_{i}^{H}$ occurs, then there exists a subgraph $H' \in \mathcal{H}$ such that event $C_{i+1}^{H'}$ occurs as well. Then by our choice of graph $H^*$, we obtain the following inequality. 
\begin{equation} \label{eq H star}
\mathbb{E}[ \dist_{D_1}(s, t)  \cdot \mathbbm{1}[s <_D t]  \mid  B_{i+1}^{s, t} \cap C_{i}^H] \leq \mathbb{E}[\dist_{D_1}(s, t) \cdot \mathbbm{1}[s <_D t] \mid   C_{i+1}^{H^*}]
\end{equation}

Continuing with our previous sequence of inequalities, 
\begin{align*}
    & \mathbb{E}[Y_1  \mid C_i^H]  & \\
    & \leq   \mathbb{E}[ \dist_{D_1}(s, t)  \cdot \mathbbm{1}[s <_D t]  \mid  B_{i+1}^{s, t} \cap C_i^H]  & \\
    & \leq \mathbb{E}[\dist_{D_1}(s, t) \cdot \mathbbm{1}[s <_D t] \mid  C_{i+1}^{H^*}]  & \text{ by inequality \ref{eq H star}} \\
    & \leq c \cdot (z - i)^2 \cdot \alpha \cdot \dist_G(s, t). & \text{ by statement $S_{k-1}$} 
\end{align*}
We conclude from the above sequence of inequalities that
 \begin{equation} \label{warmup eq c1}
     \mathbb{E}[Y_1 \mid  C_i^H] \leq c \cdot (z-i)^2 \cdot \alpha \cdot \dist_G(s, t). 
 \end{equation}

Next we upper bound  $\mathbb{E}[Y_2 \mid C_i^H]$. We will achieve this using Claim \ref{clm:single_scc}. We will also use of the following inequality which follows from the fact that $\pi_{s, t} \subseteq H$, combined with Claim \ref{clm:fam} and the union bound.
\begin{equation} \label{new eq warm up survive next level}
      \Pr[\Bar{B}_{i+1}^{s, t} \mid C_i^H] \leq 
     \sum_{e \in E(\pi_{s, t})}  \frac{w(e) \cdot \alpha}{d_i}  =  \frac{\alpha}{d_i}\cdot \dist_G(s, t)
\end{equation}

Let $A$ be the event that $s <_D t$. We observe the following sequence of inequalities:
\begin{align*}
  &    \mathbb{E}[Y_2 \mid  C_i^H] & \\
  & =   \mathbb{E}[\dist_{D_1}(s, t) \cdot \mathbbm{1}[A] \cdot \mathbbm{1}[\bar{B}_{i+1}^{s, t}] \mid   C_i^H ]          &  \\
  & =  \mathbb{E}[\dist_{D_1}(s, t)  \mid A \cap  \bar{B}_{i+1}^{s, t} \cap C_i^H ] \cdot \Pr[\bar{B}_{i+1}^{s, t} \mid A \cap  C_i^H] \cdot \Pr[A \mid  C_i^H] &  \\
  &  =  \mathbb{E}[\dist_{D_1}(s, t)  \mid A \cap  \bar{B}_{i+1}^{s, t} \cap C_i^H ] \cdot \Pr[\bar{B}_{i+1}^{s, t} \cap A \mid   C_i^H]  &  \\
  & \leq  \mathbb{E}[\dist_{D_1}(s, t)  \mid A \cap  \bar{B}_{i+1}^{s, t} \cap C_i^H ] \cdot \Pr[\bar{B}_{i+1}^{s, t} \mid   C_i^H]  &  \\
    & \leq  4d_i \cdot \Pr[\bar{B}_{i+1}^{s, t} \mid   C_i^H]  & \text{ by Claim \ref{clm:single_scc}} \\
    & \leq 4d_i \cdot \frac{\alpha}{d_i}  \cdot \dist_G(s, t) & \text{ by inequality \ref{new eq warm up survive next level}} \\
    & = 4\alpha \cdot \dist_G(s, t).
\end{align*} 
    We conclude from the above sequence of inequalities  that
    \begin{equation} \label{warm up eq c2}
        \mathbb{E}[Y_2 \mid  C_i^H] \leq c_1 \cdot \alpha \cdot \dist_G(s, t),
    \end{equation}
    for  $c_1 = 4$. Putting everything together, we get
\begin{align*}
    \mathbb{E}[X \mid C_i^H] & = \mathbb{E}[Y_1 \mid  C_i^H] + \mathbb{E}[Y_2 \mid  C_i^H] & \text{ by equation \ref{warm up main ineq}} \\
    & \leq (c \cdot (z-i)^2 +c_1) \cdot \alpha \cdot \dist_G(s, t). & \text{ by inequalities \ref{warmup eq c1} and \ref{warm up eq c2}} 
\end{align*}
 In particular, if $c \geq c_1$, then it follows that  
    $$
    \mathbb{E}[\dist_{D_1}(s, t) \cdot \mathbbm{1}[s <_D t] \mid  C_i^H] \leq c \cdot (z-i+1)^2 \cdot \alpha \cdot \dist_G(s, t)  = c \cdot k^2  \cdot \alpha \cdot \dist_G(s, t).
    $$
\end{proof}

By Lemma \ref{lem:warmup}, if nodes $s$ and $t$ lie in the same SCC in subgraph $H \subseteq G$, then statement $S_k$ holds. Therefore, when proving Lemma \ref{lem:inductive_step} we may assume without loss of generality that $s$ and $t$ do not lie in the same SCC in graph $H$.
For the remainder of this section, fix a pair of nodes $s, t \in V(G)$, and let $H \subseteq G$ be a fixed subgraph of $G$ such that $\pi_{s, t} \subseteq H$ and nodes $s$ and $t$ do not lie in the same SCC in $H$. We define a random variable $X$ to be
$$X = \dist_{D_1}(s, t) \cdot \mathbbm{1}[s <_D t].$$
We recall that our goal in proving Lemma \ref{lem:inductive_step} is to show that
$$
\mathbb{E}[X \mid  C_i^H] \leq c \cdot k^2 \cdot \alpha \cdot \dist_G(s, t). 
$$

Suppose that event $ C_i^H$ holds. Let  $S_1, \dots, S_{q} \in \mathcal{F}_i$, $q> 1$,  be the  sequence of SCCs of $H$ described in Claim \ref{clm:path_struct} with respect to path $\pi_{s, t}$ in graph $H$. Let $\{(u_j, v_j)\}_{j \in [1, q - 1]}$ be the associated set of edges, where edge $(u_j, v_j)$ is the unique edge in the intersection $E(\pi_{s, t}) \cap (S_j \times S_{j+1})$.
We observe that $(u_j, v_j) \in E(D^*) \subseteq E(D_1)$ for all $j \in [1, q-1]$.

We define vertices $x_j, y_j \in V(G)$ for all $j \in [0, q-2]$. Let $x_0 = s$ and $y_{q-2} = t$. For each $j \in [1, q-2]$, let $y_{j-1}= v_j$ and let $x_{j} = v_j$, so that $x_j = y_{j-1}$. See \Cref{fig:path} for a visualization of path $\pi_{s, t}$, the SCCs $\{S_i\}_i$, the edges $\{(u_i, v_i)\}_i$, and the vertices $\{x_i, y_i\}_i$.

\begin{figure}[h] 
\centering
\includegraphics[width=1\textwidth]{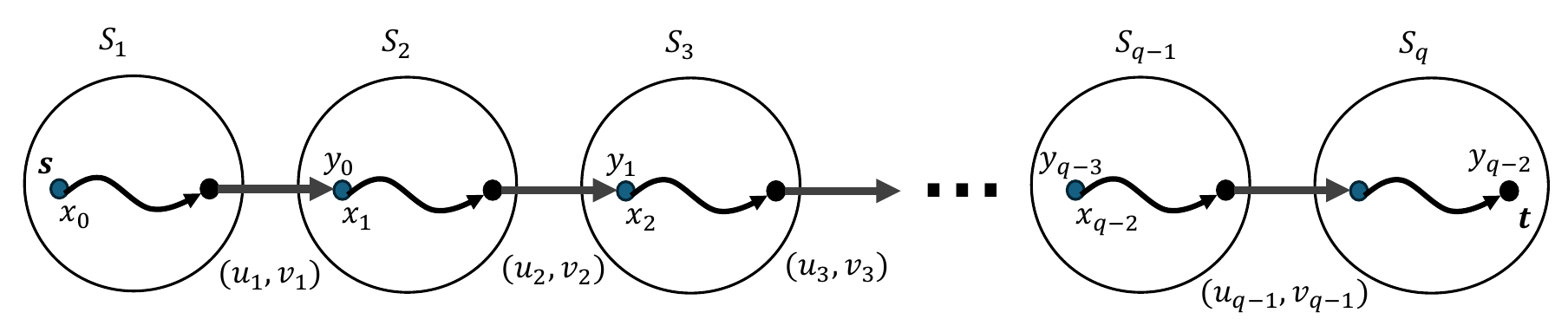}
\caption{A visualization of path $\pi_{s, t}$ in subgraph $G_i$ of $G$. We  upper bound the distance from $s$ to $t$ in DAG $D_1$ by upper bounding the distance from $x_i$ to $y_i$ in DAG $D_1$ for all $i \in [1, q-2]$.  } \label{fig:path}
\end{figure}

For each $j \in [0, q-2]$, let $X_j$ be the random variable 
$$
X_j = \dist_{D_1}(x_j,  y_j). 
$$
We will  use  the following upper bound on $\mathbb{E}[X \mid  C_i^H]$ in our proof of Lemma \ref{lem:inductive_step}. 
\begin{equation} \label{new eq X bound}
    \mathbb{E}[X \mid  C_i^H] \leq \mathbb{E}\left[ \sum_{j=0}^{q-2} \dist_{D_1}(x_j, y_{j}) \cdot \mathbbm{1}[s <_D t] \mid  C_i^H \right] \leq \sum_{j=0}^{q-2} \mathbb{E}[X_j   \mid  C_i^H ]
\end{equation}
Inequality \ref{new eq X bound} tells us that to upper bound $\mathbb{E}[X \mid  C_i^H]$, it is sufficient to upper bound $\mathbb{E}[X_j \mid  C_i^H]$ for each $j \in [0, q-2]$. This will be our next goal. 

\begin{lemma} \label{lem new X j}
If statement $S_{k-1}$ holds for a sufficiently large induction constant $c$, then for each $j \in [0, q-2]$, 
    $$
    \mathbb{E}[X_j \mid  C_i^H] \leq c \cdot k^2 \cdot \alpha \cdot  \dist_G(x_j, y_j)
    $$
\end{lemma}
\begin{proof}
 Fix an index $j \in [0, q-2]$.  To upper bound $\mathbb{E}[X_{j} \mid  C_i^H]$,  we essentially need to handle four different cases based on the ordering of vertices $x_j, u_{j+1}, v_{j+1}, $ and $y_j$ under $<_D$.  
 We briefly outline these four cases and the different inequalities we will use to handle them, before proceeding with the formal upper bound.

Let $\pi_1 := \pi_{x_j, u_{j+1}} = \pi_{s, t}[x_j, u_{j+1}]$, and let $\pi_2 := \pi_{v_{j+1}, y_j} = \pi_{s, t}[v_{j+1}, y_j]$ (these equalities follow from the consistency of our collection of paths $\{\pi_{s, t}\}_{(s, t) \in TC(G)}$). 
Let $Z_1 \in [0, z]$ denote the (random-valued) index such that $\pi_1 \subseteq G_{Z_1}$ and $\pi_1 \not \subseteq G_{Z_1+1}$. Let $Z_2 \in [0, z]$ denote the (random-valued) index such that $\pi_2 \subseteq G_{Z_2}$ and $\pi_2 \not \subseteq G_{Z_2+1}$. Recall that we use $x \leq_D y$ to denote that either $x <_D y$ or $x = y$. 
\begin{itemize}
    \item \textbf{Case 1: $x_{j} \leq_D u_{j+1}$ and $v_{j+1} \leq_D y_j$.} In this case we will  apply the following triangle inequality 
    \begin{equation*} 
        \dist_{D_1}(x_j, y_j)\leq \dist_{D_1}(x_j, u_{j+1}) + \dist_{D_1}(u_{j+1}, v_{j+1}) +  \dist_{D_1}(v_{j+1}, y_j).
    \end{equation*}

    Since $(u_{j+1}, v_{j+1}) \in E(D_1)$ and $(u_{j+1}, v_{j+1}) \in \pi_{s, t}$, it follows that
    $$\dist_{D_1}(u_{j+1}, v_{j+1}) = w_G((u_{j+1}, v_{j+1})) = \dist_G(u_{j+1}, v_{j+1}).$$
    
    We conclude that    
    \begin{equation} \label{new eq tri 1}
        X_j \leq \dist_{D_1}(x_j, u_{j+1}) + \dist_{G}(u_{j+1}, v_{j+1}) +  \dist_{D_1}(v_{j+1}, y_j).
    \end{equation}
    We will use inequality \ref{new eq tri 1} to handle the case where  $x_{j} \leq_D u_{j+1}$ and $v_{j+1} \leq_D y_j$. 
    
    \item  \textbf{Case 2: $x_j \leq_D u_{j+1}$ and $v_{j+1} >_D y_j$.}
    In this case we will apply the following triangle inequality
    \begin{equation*} 
         \dist_{D_1}(x_j, y_j) \leq \dist_{D_1}(x_j, u_{j+1}) + \dist_{D_1}(u_{j+1}, r^1_{Z_2}(v_{j+1})) +  \dist_{D_1}(r^1_{Z_2}(v_{j+1}), y_j).
    \end{equation*}
    We will make use of the fact that $$\dist_{D_1}(u_{j+1}, r^1_{Z_2}(v_{j+1})) \leq w_{D_1}((u_{j+1}, r^1_{Z_2}(v_{j+1}))) = \dist_G(u_{j+1}, v_{j+1}) + d_{Z_2},$$ since $(u_{j+1}, r^1_{Z_2}(v_{j+1})) \in E(D_1)$, to obtain the following inequality:
    \begin{equation} \label{new eq tri 2}
          X_j \leq \dist_{D_1}(x_j, u_{j+1}) + \dist_{G}(u_{j+1}, v_{j+1}) + d_{Z_2} +  \dist_{D_1}(r^1_{Z_2}(v_{j+1}), y_j).
    \end{equation}
        We will use inequality \ref{new eq tri 2} to handle the case where   $x_j \leq_D u_{j+1}$ and $v_{j+1} >_D y_j$. 
    \item \textbf{Case 3: $x_j >_D u_{j+1}$ and $v_{j+1} \leq_D y_j$.} 
    In this case we will apply the following triangle inequality
    \begin{equation*}
        \dist_{D_1}(x_j, y_j)\leq \dist_{D_1}(x_j, r^2_{Z_1}(u_{j+1})) + \dist_{D_1}( r^2_{Z_1}(u_{j+1}), v_{j+1}) +  \dist_{D_1}(v_{j+1}, y_j).
    \end{equation*}
We will make use of the fact that $\dist_{D_1}( r^2_{Z_1}(u_{j+1}), v_{j+1}) \leq \dist_{G}(u_{j+1}, v_{j+1}) + d_{Z_1}$ since $( r^2_{Z_1}(u_{j+1}), v_{j+1}) \in E(D_1)$, to obtain the following inequality
        \begin{equation} \label{new eq tri 3}
        X_j\leq \dist_{D_1}(x_j, r^2_{Z_1}(u_{j+1})) + d_{Z_1} +\dist_{G}(u_{j+1}, v_{j+1})  +  \dist_{D_1}(v_{j+1}, y_j).
    \end{equation}
    We will use inequality \ref{new eq tri 3} to handle the case where $x_j >_D u_{j+1}$ and $v_{j+1} \leq_D y_j$. 
    \item \textbf{Case 4: $x_j >_D u_{j+1}$ and $v_{j+1} >_D y_j$.}
    In this case we will apply the following triangle inequality
    \begin{equation*}      
    \dist_{D_1}(x_j, y_j)\leq \dist_{D_1}(x_j, r^2_{Z_1}(u_{j+1})) + \dist_{D_1}( r^2_{Z_1}(u_{j+1}), r^1_{Z_2}(v_{j+1})) +   \dist_{D_1}(r^1_{Z_2}(v_{j+1}), y_j).
    \end{equation*} 
We will make use of the fact that $\dist_{D_1}( r^2_{Z_1}(u_{j+1}), r^1_{Z_2}(v_{j+1})) \leq \dist_{G}(u_{j+1}, v_{j+1}) + d_{Z_1} + d_{Z_2}$ since $( r^2_{Z_1}(u_{j+1}), r^1_{Z_2}(v_{j+1})) \in E(D_1)$, to obtain the following inequality
    \begin{equation}      \label{new eq tri 4}  
    X_j \leq \dist_{D_1}(x_j, r^2_{Z_1}(u_{j+1})) + d_{Z_1} +  \dist_G(u_{j+1}, v_{j+1}) + d_{Z_2} +   \dist_{D_1}(r^1_{Z_2}(v_{j+1}), y_j).
    \end{equation}
    We will use inequality \ref{new eq tri 4} to handle the case where  $x_j >_D u_{j+1}$ and $v_{j+1} >_D y_j$.
\end{itemize}

In order to apply our four case analysis formally, we will  use the following equality, which can be proved using the fact that $\mathbbm{1}[x \leq_D y] + \mathbbm{1}[x >_D y] = 1$ for any $x, y \in V(G)$. 
\begin{multline}\label{new eq X q}
    X_{j} = X_{j} \mathbbm{1}[x_j \le_D u_{j+1}]  \mathbbm{1}[v_{j+1} \le_D y_j] + X_{j} \mathbbm{1}[x_j \le_D u_{j+1}]  \mathbbm{1}[v_{j+1} >_D y_j] \\
   + X_{j} \mathbbm{1}[x_j >_D u_{j+1}]  \mathbbm{1}[v_{j+1} \le_D y_j] +X_{j} \mathbbm{1}[x_j >_D u_{j+1}] \mathbbm{1}[v_{j+1} >_D y_j] 
\end{multline}

Inserting inequalities \ref{new eq tri 1}, \ref{new eq tri 2}, \ref{new eq tri 3}, and \ref{new eq tri 4} into equation \ref{new eq X q}, we obtain
\begin{multline} \label{new eq large}
    X_{j} \leq  \left(\dist_{D_1}(x_j, u_{j+1}) + \dist_{G}(u_{j+1}, v_{j+1}) +  \dist_{D_1}(v_{j+1}, y_j)\right) \mathbbm{1}[x_j \le_D u_{j+1}]  \mathbbm{1}[v_{j+1} \le_D y_j]\\
    + \left(\dist_{D_1}(x_j, u_{j+1}) + \dist_G(u_{j+1}, v_{j+1}) + d_{Z_2} +  \dist_{D_1}(r^1_{Z_2}(v_{j+1}), y_j) \mathbbm{1}[x_j \le_D u_{j+1}]\right)  \mathbbm{1}[v_{j+1} >_D y_j] \\
    + \left(  \dist_{D_1}(x_j, r^2_{Z_1}(u_{j+1})) +\dist_{G}(u_{j+1}, v_{j+1}) + d_{Z_1} +  \dist_{D_1}(v_{j+1}, y_j) \right)  \mathbbm{1}[x_j >_D u_{j+1}]  \mathbbm{1}[v_{j+1} \le_D y_j] \\
    + \left( \dist_{D_1}(x_j, r^2_{Z_1}(u_{j+1})) +\dist_{G}(u_{j+1}, v_{j+1}) + d_{Z_1} + d_{Z_2} +   \dist_{D_1}(r^1_{Z_2}(v_{j+1}), y_j) \right) \mathbbm{1}[x_j >_D u_{j+1}] \mathbbm{1}[v_{j+1} >_D y_j].
\end{multline}
We can rearrange terms on the right side of inequality \ref{new eq large} to get
\begin{multline} \label{new eq large 2}
 X_{j} \leq     \dist_{D_1}(x_j, u_{j+1}) \mathbbm{1}[x_j \le_D u_{j+1}]  \cdot (\mathbbm{1}[v_{j+1} \le_D y_j] + \mathbbm{1}[v_{j+1} >_D y_j]) \\
 + \dist_{D_1}(v_{j+1}, y_j) \mathbbm{1}[v_{j+1} \le_D y_j] \cdot ( \mathbbm{1}[x_j \le_D u_{j+1}] + \mathbbm{1}[x_j >_D u_{j+1}])
 \\
 + \dist_{G}(u_{j+1}, v_{j+1}) \cdot ( \mathbbm{1}[x_j \le_D u_{j+1}] + \mathbbm{1}[x_j >_D u_{j+1}] ) \cdot (\mathbbm{1}[v_{j+1} \le_D y_j] + \mathbbm{1}[v_{j+1} >_D y_j]) \\
 +  \left( (\dist_{D_1}(r_{Z_2}^1(v_{j+1}), y_j) + d_{Z_2}) \mathbbm{1}[v_{j+1} >_D y_j]  \right) \cdot ( \mathbbm{1}[x_j \le_D u_{j+1}] + \mathbbm{1}[x_j >_D u_{j+1}] ) \\
 + \left((\dist_{D_1}(x_j, r^2_{Z_1}(u_{j+1})) + d_{Z_1}) \mathbbm{1}[x_j >_D u_{j+1}]    \right) \cdot ( \mathbbm{1}[v_{j+1} \le_D y_j] +  \mathbbm{1}[v_{j+1} >_D y_j]).
\end{multline}
We can simplify inequality \ref{new eq large 2} using the fact that for any $x, y \in V(G)$, $\mathbbm{1}[x \leq_D y] + \mathbbm{1}[x >_D y] = 1$. 
\begin{multline} \label{new eq large 3}
 X_{j} \leq     \dist_{D_1}(x_j, u_{j+1}) \mathbbm{1}[x_j <_D u_{j+1}]   \\
 + \dist_{D_1}(v_{j+1}, y_j) \mathbbm{1}[v_{j+1} <_D y_j] \\
 + \dist_{G}(u_{j+1}, v_{j+1}) \\
  + (\dist_{D_1}(x_j, r^2_{Z_1}(u_{j+1})) + d_{Z_1}) \mathbbm{1}[x_j >_D u_{j+1}] \\
 +   (\dist_{D_1}(r_{Z_2}^1(v_{j+1}), y_j) + d_{Z_2}) \mathbbm{1}[v_{j+1} >_D y_j]. 
\end{multline}

Let $W_1, W_2, W_3, W_4, W_5$ be the five terms on the right side of inequality \ref{new eq large 3}, written in the order they appear. 
To finish the proof of Lemma \ref{lem new X j}, it will be sufficient to upper bound $\mathbb{E}[W_p \mid C_i^H]$ for each $p \in [1, 5]$.

We will upper bound the conditional expectations of $W_1$ and $W_2$ by appealing to Lemma \ref{lem:warmup}. We will focus on $W_1$ first. Notice that $\pi_1 = \pi_{x_j, u_{j+1}} \subseteq \pi_{s, t} \subseteq H$ and nodes $x_j$ and $u_{j+1}$ lie in the same SCC in $H$. Then we can directly apply Lemma \ref{lem:warmup} to obtain the inequality
\begin{equation}  \label{eq w1}
\mathbb{E}[W_1 \mid  C_i^H]  = \mathbb{E}[\dist_{D_1}(x_j, u_{j+1}) \cdot \mathbbm{1}[x_{j} <_D u_{j+1}] \mid  C_i^H]   \leq (c \cdot (z-i)^2 + c_1) \cdot \alpha \cdot \dist_G(x_j, u_{j+1}),
\end{equation}
where $c_1 \geq 1 $ is a universal constant that does not depend on induction constant $c$. 
Likewise, since $\pi_2 = \pi_{v_{j+1}, y_j} \subseteq \pi_{s, t} \subseteq H$ and nodes $v_{j+1}$ and $y_j$ lie in the same SCC in $H$, we directly apply Lemma \ref{lem:warmup} to prove the following inequality. 
\begin{equation} \label{eq w2}
    \mathbb{E}[W_2 \mid  C_i^H] = \mathbb{E}[\dist_{D_1}(v_{j+1}, y_j) \cdot \mathbbm{1}[x_j <_D u_{j+1}] \mid C_i^H] \leq  (c \cdot (z-i)^2 + c_1) \cdot \alpha \cdot \dist_G(v_{j+1}, y_j)
\end{equation}

Now we will upper bound the conditional expectations of $W_4$ and $W_5$ by appealing to Lemma \ref{lem:d_X}. We will focus on $W_4$ first. 
Let $A$ be the event that $x_j >_D u_{j+1}$. Suppose that event $A \cap C_i^H$ occurs. Recall that $Z_1$  is the random-valued index such that $\pi_1 \subseteq G_{Z_1}$ and $\pi_1 \not \subseteq G_{Z_1+1}$. Since $x_j >_D u_{j+1}$, there must exist an SCC $S \in \mathcal{F}_{Z_1}$ of $G_{Z_1}$ such that $x_j, u_{j+1} \in S$; if this were not true, then since $x_j \leadsto u_{j+1}$ path $\pi_1$ satisfies $\pi_1 \subseteq G_{Z_1}$ it would imply that $x_j <_D u_{j+1}$, a contradiction.  In particular, this means that  $r^2_{Z_1}(u_{j+1}) \in S$. By the above discussion and Claim \ref{clm:single_scc}, we obtain the inequality
 \begin{equation} \label{eq:z 1}
     \mathbb{E}[\dist_{D_1}(x_j, r^2_{Z_1}(u_{j+1})) \mid A \cap C_i^H] \leq \mathbb{E}[4d_{Z_1} \mid A \cap C_i^H].
 \end{equation}

We observe the following sequence of inequalities:
 \begin{align*} 
     & \mathbb{E}[W_4 \mid  C_i^H]  \\
    & = \mathbb{E}[(\dist_{D_1}(x_j, r^2_{Z_1}(u_{j+1})) + d_{Z_1}) \mathbbm{1}[x_j >_D u_{j+1}] \mid C_i^H ]  \\
    & = \mathbb{E}[\dist_{D_1}(x_j, r^2_{Z_1}(u_{j+1})) + d_{Z_1}  \mid A \cap C_i^H ] \cdot \Pr[A \mid C_i^H]  \\ 
    & \leq \mathbb{E}[5d_{Z_1} \mid A \cap C_i^H ] \cdot \Pr[A \mid C_i^H] & \text{ by inequality \ref{eq:z 1}}  \\
    & \leq 5\mathbb{E}[d_{Z_1} \mid C_i^H] \\
    & \leq c_2 \cdot  (z-i+1) \cdot \alpha \cdot \dist_G(x_j, u_{j+1}), & \text{ by Lemma \ref{lem:d_X}}
 \end{align*}
where $c_2 \geq 1$ is a sufficiently large universal constant that does not depend on the induction constant $c$. We have shown
\begin{equation} \label{eq w4}
    \mathbb{E}[W_4 \mid  C_i^H] \leq c_2 \cdot  (z-i+1) \cdot \alpha \cdot \dist_G(x_j, u_{j+1}).
\end{equation}

An identical argument will work for bounding $\mathbb{E}[W_5 \mid C_i^H]$ as well; we present this argument for completeness. Let $A$ be the event that $v_{j+1} >_D y_j$. Suppose that event $A \cap C_i^H$ occurs. Recall that $Z_2$  is the random-valued index such that $\pi_2 \subseteq G_{Z_2}$ and $\pi_2 \not \subseteq G_{Z_2+1}$. Since $v_{j+1} >_D y_j$, there must exist an SCC $S \in \mathcal{F}_{Z_2}$ of $G_{Z_2}$ such that $v_{j+1}, y_j \in S$; if this were not true, then since $v_{j+1} \leadsto y_j$ path $\pi_2$ satisfies $\pi_2 \subseteq G_{Z_2}$ it would imply that $v_{j+1} <_D y_{j}$, a contradiction.  In particular, this means that  $r^1_{Z_2}(v_{j+1}) \in S$. By the above discussion and Claim \ref{clm:single_scc}, we obtain the inequality
 \begin{equation} \label{eq:z 2}
     \mathbb{E}[\dist_{D_1}(r^1_{Z_2}(v_{j+1}), y_j) \mid A \cap C_i^H] \leq \mathbb{E}[4d_{Z_2} \mid A \cap C_i^H].
 \end{equation}

We observe the following sequence of inequalities:
 \begin{align*} 
     & \mathbb{E}[W_5 \mid  C_i^H]  \\
    & = \mathbb{E}[(\dist_{D_1}(r^1_{Z_2}(v_{j+1}), y_j) + d_{Z_2}) \mathbbm{1}[v_{j+1} >_D y_{j}] \mid C_i^H ]  \\
    & = \mathbb{E}[\dist_{D_1}(r^1_{Z_2}(v_{j+1}), y_j) + d_{Z_2} \mid A \cap C_i^H ] \cdot \Pr[A \mid C_i^H]  \\ 
    & \leq \mathbb{E}[5d_{Z_2} \mid A \cap C_i^H ] \cdot \Pr[A \mid C_i^H] & \text{ by inequality \ref{eq:z 2}}  \\
    & \leq 5\mathbb{E}[d_{Z_2} \mid C_i^H] \\
    & \leq c_2 \cdot  (z-i+1) \cdot \alpha \cdot \dist_G(v_{j+1}, y_{j}), & \text{ by Lemma \ref{lem:d_X}}
 \end{align*}
where $c_2 \geq 1$ is a sufficiently large universal constant that does not depend on the induction constant $c$. We have shown
\begin{equation} \label{eq w5}
    \mathbb{E}[W_5 \mid  C_i^H] \leq c_2 \cdot  (z-i+1) \cdot \alpha \cdot \dist_G(v_{j+1}, y_{j}).
\end{equation}

We are ready to finish the proof of Lemma \ref{lem new X j}. We will make use of the fact that 
\begin{equation} \label{eq dist g}
    \dist_G(x_j, y_j) = \dist_G(x_j, u_{j+1}) + \dist_G(u_{j+1}, v_{j+1}) + \dist_G(v_{j+1}, y_j),
\end{equation}
since path $\pi_{s, t}$ is a $s \leadsto t$ shortest path in $G$. We complete the proof of Lemma \ref{lem new X j} with the following sequence of inequalities. 
\begin{align*}
    \mathbb{E}[X_j \mid C_i^H] & \leq \sum_{p=1}^5 \mathbb{E}[W_p \mid C_i^H] &  \text{ by inequality \ref{new eq large 3}} \\
    &  \leq \sum_{p\in \{1, 2, 4, 5\}} \mathbb{E}[W_p \mid C_i^H] + \dist_G(u_{j+1}, v_{j+1}) & \text{since } W_3 = \dist_G(u_{j+1}, v_{j+1})\\
\end{align*}
Applying inequalities  \ref{eq w1}, \ref{eq w2}, \ref{eq w4}, and \ref{eq w5}, we continue
\begin{multline*}
    \mathbb{E}[X_j \mid C_i^H] \leq (c \cdot (z-i)^2 + c_1 + c_2(z-i+1)) \cdot \alpha \cdot \dist_G(x_j, u_{j+1}) + \dist_G(u_{j+1}, v_{j+1})\\
    +  (c \cdot (z-i)^2 + c_1 + c_2(z-i+1)) \cdot \alpha \cdot \dist_G(v_{j+1}, y_{j}).  
\end{multline*}
Using equation \ref{eq dist g}, this yields
\begin{align*}
\mathbb{E}[X_j \mid C_i^H] 
    & \leq (c \cdot (z-i)^2 + c_1 + c_2(z-i+1)) \cdot \alpha \cdot \dist_G(x_j, y_j) & \text{ by equation \ref{eq dist g}} \\
    & \leq (c \cdot (z-i)^2 + (c_1 + c_2)(z-i+1)) \cdot \alpha \cdot \dist_G(x_j, y_j) \\
    & \leq c \cdot (z-i+1)^2 \cdot \alpha \cdot \dist_G(x_j, y_j) & \text{by letting $c \geq c_1 + c_2$}\\
    & \leq c \cdot k^2 \cdot \alpha \cdot \dist_G(x_j, y_j).
\end{align*}

Notice that we require induction constant $c$ to be at least $c \geq c_1 + c_2$, where $c_1$ is the constant in the statement of Lemma \ref{lem:warmup}, and $c_2$ is the constant in inequality \ref{eq w4}. 
\end{proof}

With Lemma \ref{lem new X j} proved, we can now complete the inductive step of our induction proof of Lemma \ref{lem:induction}.

\inductivestep*
\begin{proof}
We are ready to finish our proof of the inductive step by combining all of our proven inequalities.
\begin{align*}
    \mathbb{E}[X \mid  C_i^H] & \leq \sum_{j=0}^{q-2} \mathbb{E}[X_j \mid  C_i^H] & \text{ by inequality \ref{new eq X bound}} \\
    & \leq \left(\sum_{j=0}^{q-2} c \cdot k^2 \cdot \alpha \cdot \dist_G(x_j, y_{j})  \right)  & \text{by Lemma \ref{lem new X j}} \\
    & = c \cdot k^2 \cdot \alpha \cdot \dist_G(s, t),
\end{align*}
where the last equality follows from the fact that $\dist_G(s, t) = \sum_{j=0}^{q-2}\dist_G(x_j, y_{j})$, since $\pi_{s, t}$ is a shortest path in $G$. Note that
Lemma \ref{lem new X j} requires that the induction constant $c$ in statements $S_{k-1}$ and $S_k$ has size at least $c \geq c_1 + c_2$, where $c_1$ is the constant in the statement of Lemma \ref{lem:warmup}, and $c_2$ is the constant in inequality \ref{eq w4}. 
We have completed the inductive step of the proof of Lemma \ref{lem:induction}. 
\end{proof}

 By  Claim \ref{clm:base_induction} and Lemma \ref{lem:inductive_step} we have proven that Inductive Statement $S_k$ holds for all $k \in [1, z+1]$. This immediately implies Lemma \ref{lem:induction}.

\induction*

\subsection{Finishing the proof of Theorem \ref{thm:dag_cover}}
\label{subsec:final}

With Lemma \ref{lem:induction} in hand, we can now easily finish our proof of Theorem \ref{thm:dag_cover}.

\begin{claim}
    For all $s, t \in V(G)$,
    $$
    \mathbb{E}[\dist_{D_1}(s, t) \cdot \mathbbm{1}[s <_D t]] = O(\log^3n \cdot  \log\log n) \cdot \dist_G(s, t).
    $$
    \label{clm:s_leq_t}
\end{claim}
\begin{proof}
Recall that $C_0^G$ is the event that $G_0 = G$. Notice that by our construction, $G_0 = G$ and shortest path $\pi_{s, t}$ satisfies $\pi_{s, t} \subseteq G_0$, unconditionally. Then by Lemma \ref{lem:induction}, 
$$
\mathbb{E}[\dist_{D_1}(s, t) \cdot \mathbbm{1}[s <_D t]] = \mathbb{E}[\dist_{D_1}(s, t) \cdot \mathbbm{1}[s <_D t] \mid C_0^G] \leq O(z^2 \alpha) \cdot \dist_G(s, t). 
$$
Since $z = O(\log n)$ and $\alpha = O( \log n \cdot \log \log n)$, the claim follows. 
\end{proof}

We are ready to prove Lemma \ref{lem:exp_err}, which we restate below.

\expected*
\begin{proof}
\begin{multline*}
  \mathbb{E}[\min(\dist_{D_1}(s, t), \dist_{D_2}(s, t))]  \leq \mathbb{E}[\dist_{D_1}(s, t) \cdot \mathbbm{1}[s <_D t]] + \mathbb{E}[\dist_{D_2}(s, t) \cdot \mathbbm{1}[s >_D t]] \\
  \leq O(\log^3 n \cdot \log \log n) \cdot \dist_G(s, t),
\end{multline*}
where the final inequality follows from 
Lemma \ref{lem:s_geq_t} and Claim \ref{clm:s_leq_t}.
\end{proof}

We can now finish the distortion analysis and the proof of Theorem \ref{thm:dag_cover}. 

\begin{proof}[Proof of Theorem \ref{thm:dag_cover}]
    We have finished the running time analysis and size analysis for Theorem \ref{thm:dag_cover}. What remains is to prove that with high probability, for all $s, t \in V(G)$, there exists $D \in \mathcal{D}$ such that
    $$
    \dist_D(s, t) \leq O(\log^3 n \cdot \log \log n) \cdot \dist_G(s, t).
    $$
    Fix a pair of nodes $s, t \in V(G)$. 
    Since $\mathbb{E}[\min(\dist_{D_1}(s, t), \dist_{D_2}(s, t))] = O(\log^3 n \cdot \log \log n) \cdot \dist_G(s, t)$ by Lemma \ref{lem:exp_err}, by Markov's inequality we conclude that with probability at least $1/2$, 
    $$
    \min(\dist_{D_1}(s, t), \dist_{D_2}(s, t)) = O(\log^3 n \cdot \log \log n) \cdot \dist_G(s, t).
    $$
    Then with probability at least
    $$
    1 - \left(\frac{1}{2}\right)^{10\log n} \geq 1 - n^{-10},
    $$
    there exists a DAG $D \in \mathcal{D}$ such that $\dist_D(s, t) \leq \dist_G(s, t)$. Then by applying the union bound over all pairs of nodes $s, t \in V(G)$, we conclude that  $\mathcal{D}$ is an $O(\log^3 n \cdot \log \log n)$-distance-preserving DAG cover of $G$ with high probability. 
\end{proof}

We will also quickly prove Theorem \ref{thm:embed}, which we restate below.
\embed*
\begin{proof}
    We define a random graph $D$ from our distribution $\mathcal{D}$ as follows. Construct random graphs $D_1$ and $D_2$ as in section \ref{subsec:const_D}, and choose graph $D$ from set $\{D_1, D_2\}$ uniformly at random. Notice that the first property  and the third property of Theorem \ref{thm:embed} holds immediately from Theorem \ref{thm:dag_cover}.

    Let $(s, t) \in TC(G)$. Notice that by our construction of DAGs $D_1$ and $D_2$,  $s$ can reach $t$ in DAG $D_1$ if and only if $ s <_D t$, where $<_D$ is the total order on $V(G)$ defined in subsection \ref{subsec:const_D}. Likewise, $s$ can reach $t$ in DAG $D_2$ if and only if $s >_D t$. Then it immediately follows that $s$ can reach $t$ in sampled DAG $D$ with probability $1/2$ if $s \neq t$. 

    To bound the (conditional) expected distortion, we now consider two cases:
    \begin{itemize}
        \item \textbf{Case 1:} DAG $D = D_1$. Then
        $$
        \mathbb{E}[\dist_{D}(s, t) \mid s \leadsto_D t] = \mathbb{E}[\dist_{D_1}(s, t) \mid s <_D t] = \mathbb{E}[\dist_{D_1}(s, t) \cdot \mathbbm{1}[s <_D t]]=O(\log^3 \cdot \log \log n) \cdot \dist_G(s, t),
        $$
        by Claim \ref{clm:s_leq_t}.
        \item \textbf{Case 2:} DAG $D = D_2$. Then
        $$
\mathbb{E}[\dist_{D}(s, t) \mid s \leadsto_D t]  = \mathbb{E}[\dist_{D_2}(s, t) \mid s >_D t]  = \mathbb{E}[\dist_{D_1}(s, t) \cdot \mathbbm{1}[s >_D t]]=O(\log^2 n \cdot \log \log n) \cdot \dist_G(s, t),
        $$
        by Lemma \ref{lem:s_geq_t}.
    \end{itemize}
\end{proof}

\section{A lower bound for DAG covers with $\widetilde{O}(m)$ additional edges }

The goal of this section is to prove the following lower bound.

\lbm*

The proof of this theorem is quite involved, and we  split it into five different phases:
\begin{enumerate}
    \item In subsection \ref{subsec:base}, we construct a base graph $G$ and associated collection of shortest paths $\Pi$. Graph $G$ and paths $\Pi$ essentially correspond to a collection of points and lines in the Euclidean plane. We prove that $G, \Pi$ have a special ``expansion'' property (Lemma \ref{lem:base_degree}).  
    \item In subsection \ref{subsec:prod}, we construct a \textit{product graph } $G^{\times}$ and a collection of shortest paths $\Pi^{\times}$, constructed by taking a graph  product of two copies of our base graph $G$ and set of paths $\Pi$. This graph product has been used in prior work \cite{hesse2003directed}. We prove that the special expansion property of $G, \Pi$  is preserved in $G^{\times}, \Pi^{\times}$ (Lemma \ref{lem:prod_degree}). 
    \item In subsection \ref{subsec:match}, we define what we call a \textit{matching graph} $J$ of $G^{\times}, \Pi^{\times}$ that will be useful for analysis purposes later. We prove that the edges of graph $J$ can be decomposed into a collection of large matchings (Lemma \ref{lem:match_decomp}). The proof of this lemma will follow from the special expansion property of $G^{\times}, \Pi^{\times}$.
    \item In subsection \ref{subsec:random}, we define a random graph $G^*$ and a collection of shortest paths $\Pi^*$ constructed from $G^{\times}, \Pi^{\times}$. Graph $G^*$ and paths $\Pi^*$ will be constructed in a similar manner as in our lower bound for DAG covers with $f(n)$ additional edges in subsection \ref{subsec:fn_lb}.  
    \item In subsection \ref{subsec:analysis}, we finish the proof of Theorem \ref{thm:lb_m} using  $G^*, \Pi^*$. At a high level, we will use our large matchings in matching graph $J$ to argue that no DAG in any DAG cover can cover many of the shortest paths in $\Pi^*$ simultaneously.
\end{enumerate}

\subsection{Base Graph Construction $G, \Pi$}

We begin our lower bound construction by constructing a base graph $G$ and associated collection of shortest paths $\Pi$.

\label{subsec:base}
\begin{lemma}[cf. \cite{coppersmith2006sparse}]
\label{lem:base_graph}
    There exists an $n$-node directed, edge-weighted graph $G $, with weight function $w_G(\cdot )$,  and a set $\Pi$ of $|\Pi| = p$ paths, where $p \in [n^{2/3}, n^{3/2}]$, satisfying the following conditions:
    \begin{enumerate}
        \item Graph $G$ has maximum degree $d = \Theta\left(\frac{p^{2/3}}{n^{1/3}}\right)$, and $|E(G)| = \Theta(n^{2/3}p^{2/3})$, 
        \item Graph $G$ has $\ell$ layers $L_1, \dots, L_{\ell}$, where $\ell = \Theta\left(\frac{n^{2/3}}{p^{1/3}}\right)$, and each layer has $n / \ell$ nodes,
        \item Every path $\pi \in \Pi$ contains $\ell$ nodes, and the $i$th node of $\pi$ is in layer $L_i$,
        \item Every path $\pi \in \Pi$ is the unique shortest path between its endpoints in $G$, and
        \item Paths in $\Pi$ are pairwise edge-disjoint.
    \end{enumerate}
\end{lemma}

\paragraph{Construction of $G, \Pi$.} 
We  construct graph $G$ and set of paths $\Pi$ as follows:
\begin{itemize}
    \item For each index $i \in [1, \ell]$, where $\ell = \frac{n^{2/3}}{2p^{1/3}}$,  we let the nodes in layer $L_i$ correspond to a column of points in a two-dimensional grid. Specifically, let $L_i$ be
    $$
    L_i = \{i\} \times \left[1, \frac{n}{\ell}\right].
    $$
    \item We define a set of source vertices $S \subseteq L_1$ as follows:
    $$
    S = \{1\} \times \left[1, \frac{n^{1/3}}{2}\right].
    $$
    \item Let $r = \frac{n}{2\ell^2}$.  For each source vertex $s \in S$ and every integer $x$ in the interval $x \in \left[1, r\right]$, we define a sequence of nodes $\pi_{s, x}$ as follows. For $i \in [1, \ell]$, we let the $i$th node of $\pi_{s, x}$ be
    $$
    s+ (i-1) \cdot \left(1,  x\right) \in L_i.
    $$
    Note that the first node of $\pi_{s, x}$ is node $s \in S$, and $\pi_{s, x}$ contains exactly one node in each layer. We  treat node sequence $\pi_{s, x}$ as a path, and we will assign edge weights to this path. Let $\pi_{s, x} = (v_1, \dots, v_{\ell})$, where $v_1 = s$ and $v_i \in L_i$ for $i \in [1, \ell]$. We assign each directed edge $(v_i, v_{i+1})$ in $\pi_{s, x}$ weight $x^2$, so that $w_G((v_i, v_{i+1})) = x^2$. 
    We define our collection of paths $\Pi$ to be
    $$
    \Pi = \{ \pi_{s, x} \mid s \in S, x \in [1, r] \}.
    $$
    \item We define graph $G$ to be the union of the paths in $\Pi$ over the vertex set $V = \cup_i L_i$. 
    Formally, we let $G = \graph(V, \Pi)$. 
    We let $w_G:E(G) \mapsto \mathbb{R}$ be our weight function for graph $G$. Notice that weight function $w_G$ will be well-defined since paths in $\Pi$ are pairwise edge-disjoint by Property 5 of Lemma \ref{lem:base_graph}. 
    This completes the construction of $G$, $w_G$,  and $\Pi$. 
\end{itemize}
\begin{proof}[Proof of Lemma \ref{lem:base_graph}]
    Conditions 2 and 3 are immediate from the construction of $G$ and $\Pi$. Condition 5 follows from observing that paths in $\Pi$ correspond to distinct straight lines in $\mathbb{R}^2$, and therefore intersect at no more than one point. By our construction and Condition 5, the total number of edges in $G$ is 
    $$
    |E(G)| = \sum_{\pi \in \Pi}|\pi| = \Theta(\ell p) = \Theta(n^{2/3}p^{2/3}).
    $$
    
    What remains is to verify Condition 4. Fix an $s\leadsto t$ path $\pi \in \Pi$ and suppose towards contradiction that there exists an $s \leadsto t$  path $\pi'$ such that $\pi \neq \pi'$ and $w_G(\pi') \leq w_G(\pi)$. 
    Let $\pi = \pi_{s, x}$ for some $x \in [1, r]$. Since graph $G$ is a layered graph, path $\pi'$ must have $|\pi'| = \ell - 1$ edges. Let $\pi' = (v_1, \dots, v_{\ell})$, where $v_1=s$,  $v_{\ell}=t$, and $v_i \in L_i$ for $i \in [1, \ell]$. 
    Let the $i$th edge of path $\pi'$ correspond to the vector $v_{i+1} - v_i = (1, x_i)$ for some $x_i \in [1, r]$.

    % Let $\pi' = (v_1, \dots, v_{\ell})$, where $v_1=s$ and $v_{\ell}=t$ and $v_i = (v_{i, 1}, v_{i, 2})$ for all $i \in [1, \ell]$.     
    % For each $i \in [1, \ell-1]$, we define the vector $u_i$ to be
    % $$
    % u_i = (v_{i+1, 1} - v_{i, 1}, \quad v_{i+1, 2} - v_{i, 2}, \quad w_G(v_{i}, v_{i+1}) ) = (1, \quad  v_{i+1, 2} - v_{i, 2}, \quad w_G(v_{i}, v_{i+1}) ).
    % $$

    We quickly observe that since $\pi$ and $\pi'$ are both $s \leadsto t$ paths and $w_G(\pi') \leq w_G(\pi)$,
    $$
  t - s =   (\ell - 1) \cdot (1, x) = \sum_{i=1}^{\ell - 1}(1, x_i) \quad \text{ and } \quad (\ell - 1) \cdot x^2 \leq \sum_{i=1}^{\ell - 1}x_i^2
    $$
    The above equality and inequality can be restated as
    $$
  x = \frac{1}{\ell - 1} \sum_{i=1}^{\ell - 1}x_i \quad \text{ and } \quad x^2 \leq \frac{1}{\ell - 1}\sum_{i=1}^{\ell - 1}x_i^2
    $$
    By Jensen's inequality,  the above equality and inequality simultaneously hold only if $x_i = x$ for all $i \in [1, \ell - 1]$. We conclude that $\pi = \pi'$, a contradiction.
\end{proof}

We will prove that base graph $G$ and collection of paths $\Pi$ have a certain special ``expansion'' property in Lemma \ref{lem:base_degree}. Before we can state this property we need to introduce several definitions. 

\begin{definition}[Path-node incidences]
    We define the set of path-node incidences $\mathcal{I}_{V, \Pi}$ between a set of nodes $V$ and a set of paths $\Pi$ as
    $$
    \mathcal{I}_{V, \Pi} := \{(v, \pi) \mid v\in V, \pi \in \Pi, \text{ and } v \in \pi \}.
    $$
\end{definition}

Let $G, \Pi$ be the $n$-node graph and set of $|\Pi| = p $ paths described in Lemma \ref{lem:base_graph}. Then the number of path-node incidences between $V(G), \Pi$ is 
$$
|\mathcal{I}_{V(G), \Pi}| = \ell p = \Theta(n^{2/3}p^{2/3}).
$$
We also define a variation of the notion path-node incidences, where we are only interested in incidences at low-degree nodes. 

\begin{definition}[Degree-bounded path-node incidences]
Let $V$ be a set of nodes and $\Pi$ be a set of paths, and let $G = \graph(V, \Pi)$. 
We define the set of $d$-degree-bounded path-node incidences between a set of nodes $V$ and a set of paths $\Pi$ as 
    $$
    \mathcal{I}^{\leq d}_{V(G), \Pi} := \{(v, \pi) \mid v\in V(G), \pi \in \Pi, \deg_G(v) \leq d, \text{ and } v \in \pi \}.
    $$
    Likewise, we let $\mathcal{I}^{> d}_{V(G), \Pi}$ denote the set of path node incidences that are \textit{not} $d$-degree-bounded.
\end{definition}

We are now able to state Lemma \ref{lem:base_degree}, which corresponds to our special expansion property of $G, \Pi$. 

\begin{lemma}
\label{lem:base_degree}
Let $G, \Pi$ be the $n$-node graph and set of $|\Pi| = p$ paths described in  \cref{lem:base_graph}. Let $\Pi' \subseteq \Pi$ be a set of $|\Pi'| \leq p \cdot n^{-\varepsilon}$ paths in $\Pi$, where $\varepsilon > 0$ is a sufficiently small constant. Let $G' = \graph(V(G), \Pi')$ be the graph $G$ induced on the set of paths $\Pi'$, and let $d' := \frac{p^{2/3}}{n^{1/3 + \varepsilon/10}}$. Then  the number of $d'$-degree-bounded path-node incidences between $G' $ and $\Pi'$ is at least
$$
|\mathcal{I}^{\leq d'}_{V(G'), \Pi'}| \geq \left(1 - \frac{1}{16}\right) \cdot |\mathcal{I}_{V(G'), \Pi'}|.
$$
\end{lemma}

Notice that Lemma \ref{lem:base_degree} is a slightly weaker version of the following statement: ``The maximum degree of graph $G'$ is at most $d'$.'' This quoted statement is false in general, but Lemma \ref{lem:base_degree} is true as we will soon prove. We can view Lemma \ref{lem:base_degree} as an expansion property in the following sense. Let $H$ be the  incidence graph between nodes $V(G)$ and paths $\Pi$. Then Lemma \ref{lem:base_degree} implies that for any small set $\Pi' \subseteq \Pi$, the neighborhood $N_H(\Pi')$ of $\Pi'$ in $H$ is polynomially larger than the trivial lower bound. 

The proof of Lemma \ref{lem:base_degree} will make use of the Szemerédi–Trotter Theorem, a classical theorem from incidence geometry. 

\begin{theorem}[Szemerédi–Trotter Theorem \cite{szemeredi1983extremal}]
    Given $n$ distinct points and $m$ distinct lines in the Euclidean plane, the number of point-line incidences is $O(n^{2/3}m^{2/3}+n+m)$. 
    \label{thm:ST}
\end{theorem}

\begin{proof}[Proof of Lemma \ref{lem:base_degree}]
    Suppose towards contradiction that
    $$
|\mathcal{I}^{> d'}_{V(G'), \Pi'}| > \frac{1}{16} \cdot |\mathcal{I}_{V(G'), \Pi'}|. 
    $$
    Now let $V' \subseteq V(G')$ denote the set of nodes in $G'$ of degree greater than $d'$. Formally, let
    $$
    V' = \{v \in V(G') \mid \deg_{G'}(v) > d'\}. 
    $$
    We observe that the number of path-node incidences between $V'$ and $\Pi'$ is large: 
    $$
    |\mathcal{I}_{V', \Pi'}| = |\mathcal{I}^{> d'}_{V(G'), \Pi'}| = \Theta\left(\frac{1}{16} \cdot |\mathcal{I}_{V(G'), \Pi'}|\right) = \Theta(\ell p \cdot n^{-\varepsilon}) = \Theta(n^{2/3- \varepsilon}p^{2/3}).
    $$
    Additionally, since every node in $V'$ participates in at least $d'$ incidences, the number of nodes in $V'$ is small:
    $$
    |V'| \leq \frac{|\mathcal{I}_{V', \Pi'}|}{d'} = \Theta\left(  \frac{n^{2/3 - \varepsilon}p^{2/3}}{p^{2/3} / n^{1/3 + \varepsilon/10}}\right) = \Theta\left( n^{1-9\varepsilon/10}\right).
    $$
    Recall that each path in $\Pi' \subseteq \Pi$ corresponds to a distinct line in the Euclidean plane, and each node in $V' \subseteq V(G)$ corresponds to a distinct point in the Euclidean plane. Since we have $|\Pi'| = p \cdot n^{-\varepsilon}$ lines and $|V'| = O(n^{1 - 9\varepsilon/10})$ points, applying Theorem \ref{thm:ST}, we obtain the upper bound
    $$
    |\mathcal{I}_{V', \Pi'}| = O(|\Pi'| + |V'| + |\Pi'|^{2/3} \cdot |V'|^{2/3}) = O( p \cdot n^{-\varepsilon} + n^{1 - 9\varepsilon/10} + p^{2/3}n^{-2\varepsilon/3} \cdot n^{2/3 - 3\varepsilon/5}) = O(n^{2/3-19\varepsilon/15}p^{2/3}),
    $$
    where the final equality follows from our choice of $p \in [n^{2/3}, n^{3/2}]$ and by taking $\varepsilon>0$ to be sufficiently small. 

    The above upper bound on $|\mathcal{I}_{V', \Pi'}|$ contradicts our earlier bound of 
    $|\mathcal{I}_{V', \Pi'}| = \Theta(n^{2/3-\varepsilon}p^{2/3})$. We conclude that our claimed lemma holds.     
\end{proof}

\subsection{Product Graph Construction $G^{\times}, \Pi^{\times}$}
\label{subsec:prod}
In this subsection, we construct a \textit{product graph } $G^{\times}$ and a collection of shortest paths $\Pi^{\times}$, constructed by taking a graph  product of two copies of our base graph $G$ and set of paths $\Pi$. Then we prove that the special expansion property of $G, \Pi$  is preserved in $G^{\times}, \Pi^{\times}$ (Lemma \ref{lem:prod_degree}). 

Let $G, \Pi$ be the $n$-node, $\ell$-layer directed graph with weights $w_G$ and set of $|\Pi| = p$ (weighted) paths described in  \cref{lem:base_graph}. We construct a weighted, directed product graph $G^{\times}$, with associated weight function $w^{\times}$ and  set  of paths $\Pi^{\times}$ as follows. 
\begin{itemize}
    \item The product graph has $2\ell$ layers, which we denote as $L^{\times}_1, \dots, L^{\times}_{2\ell}$. 
    \item The $(2i - 1)$th layer, $L^{\times}_{2i-1}$, of $G^{\times}$ is defined as
    $$
    L^{\times}_{2i-1} := L_i \times L_i,
    $$
    where $L_i$ is the $i$th layer of base graph $G$.
    \item The $2i$th layer, $L^{\times}_{2i}$, of $G^{\times}$ is defined as
    $$
    L^{\times}_{2i} := L_{i+1} \times L_i.
    $$
    \item  For each pair of paths $(\pi_1, \pi_2) \in \Pi \times \Pi$, where $\pi_1 = (u_1, \dots, u_{\ell})$ and $\pi_2 = (v_1, \dots, v_{\ell})$, we define a corresponding path as
    $$
    (u_1, v_1) \rightarrow (u_2, v_1) \rightarrow (u_2, v_2) \rightarrow \cdots \rightarrow (u_{\ell}, v_{\ell -1}) \rightarrow (u_{\ell}, v_{\ell}).
    $$
    We will occasionally denote this new path by the tuple $(\pi_1, \pi_2)$, and we denote the edges of this path by $E((\pi_1, \pi_2))$. Note that the $i$th node of $(\pi_1, \pi_2)$ is in layer $L^{\times}_i$ of $G^{\times}$. 
    \item We define the edge set $E(G^{\times})$ of $G^{\times}$ as
    $$E(G^{\times}) = \bigcup_{(\pi_1, \pi_2) \in \Pi \times \Pi}  E((\pi_1, \pi_2)).$$
    For every edge $e \in E(G^{\times})$ of the form 
    $
    e = (u, v) \rightarrow (u', v),
    $
    we assign edge $e$ weight $w^{\times}(e) := w_G(u, u')$ in $G^{\times}$. Similarly, for every edge $e \in E(G^{\times})$ of the form 
    $
    e = (u, v) \rightarrow (u, v'),
    $
    we assign edge $e$ weight $w^{\times}(e) := w_G(v, v')$ in $G^{\times}$. 
    \item We define the set  $\Pi^{\times}$ of paths in $G^{\times}$ as
    $$
    \Pi^{\times} = \bigcup_{(\pi_1, \pi_2) \in \Pi \times \Pi}  (\pi_1, \pi_2).
    $$
\end{itemize}

\begin{lemma}
    Let $G, \Pi$ be the $n$-node graph and set of $|\Pi| = p$ paths, where $p \in [n^{2/3}, n^{3/2}]$, described in Lemma \ref{lem:base_graph}. Let $G^{\times}$, $\Pi^{\times}$ be the product graph and product paths constructed from $G, \Pi$. Then $G^{\times}, \Pi^{\times}$ satisfy the following conditions:
    \begin{enumerate}
        \item the parameters of $G^{\times}, \Pi^{\times}$ are:
        \begin{itemize}
            \item $|V(G^{\times})| =  \Theta\left(\frac{n^2}{\ell}\right)$, where $\ell = \frac{n^{2/3}}{p^{1/3}}$,
            \item $|E(G^{\times})| = |V(G^{\times})| \cdot \Theta\left( d\right)$ and $G^{\times}$ has maximum degree $\Theta(d)$, where $d = \frac{p^{2/3}}{n^{1/3}}$,
            \item  $|\Pi^{\times}| = \Theta(p^2)$,
        \end{itemize}
        \item graph $G^{\times}$ has $2\ell$ layers $L_1^{\times}, \dots, L_{2\ell}^{\times}$, and each layer has $n^2 / \ell^2$ nodes,
        \item every path $\pi \in \Pi^{\times}$ contains $2\ell$ nodes, and the $i$th node of $\pi$ is in layer $L_i^{\times}$,
        \item every path $\pi \in \Pi^{\times}$ is the unique shortest path between its endpoints in $G^{\times}$, and
        \item intersecting paths in $\Pi^{\times}$ intersect on either one node or one edge in $G^{\times}$.
    \end{enumerate}
    \label{lem:prod_graph}
\end{lemma}
\begin{proof}
Conditions 1, 2, and 3 are immediate from the construction of $G^{\times}$ and $\Pi^{\times}$ and Lemma \ref{lem:base_graph}. To prove Condition 5, observe the following: for any path $(\pi_1, \pi_2) \in \Pi^{\times}$, this path is uniquely identified by any two edges on path $(\pi_1, \pi_2)$. 
Consequently, any two paths in $\Pi^{\times}$ can intersect on either one node or one edge. 

What remains is to verify Condition 4. Let $\pi   = (\pi_1, \pi_2) \in \Pi^{\times}$ be an $s \leadsto t$ path in $\Pi^{\times}$. Suppose towards contradiction that there exists an $s \leadsto t$ shortest path $\pi'$ such that $\pi' \neq \pi$. By the construction of graph $G^{\times}$, we can write path $\pi'$ as 
$$
\pi' = (u_1, v_1) \rightarrow (u_2, v_1) \rightarrow (u_2, v_2) \rightarrow \cdots \rightarrow (u_{\ell}, v_{\ell-1}) \rightarrow (u_{\ell}, v_{\ell}),
$$
where $s=(u_1, v_1)$ and $t = (u_{\ell}, v_{\ell})$. From path $\pi'$, we can define paths $\pi_1'$ and $\pi_2'$ as follows:
$$
\pi_1' = (u_1, u_2, \dots, u_{\ell}) \text{ and } \pi_2' = (v_1, v_2, \dots, v_{\ell}).
$$
Observe that $\pi_1', \pi_2' \subseteq G$ are valid paths in $G$ that share the same endpoints as $\pi_1$ and $\pi_2$, respectively.
Additionally, by our construction of $G^{\times}$,
$$
w_G(\pi_1') + w_G(\pi_2') = w^{\times}(\pi') \leq w^{\times}((\pi_1, \pi_2)) = w_G(\pi_1) + w_G(\pi_2).
$$
Since $\pi_1$ and $\pi_2$ are unique shortest paths in $G$ by Lemma \ref{lem:base_graph}, it follows that $w_G(\pi_1') \geq w_G(\pi_1)$ and $w_G(\pi_2') \geq w_G(\pi_2)$, which combined with the above sequence of inequalities immediately implies
$$
w_G(\pi_1') = w_G(\pi_1) \text{ and } w_G(\pi_2') = w_G(\pi_2).
$$
Then since paths $\pi_1$ and $\pi_2$ are unique shortest paths in $G$, it follows that $\pi_1 =\pi_1'$ and $\pi_2 = \pi_2'$. However, this implies that $\pi' = \pi$, a contradiction. 
\end{proof}

Additionally, we will make use of the following claim about the degrees of nodes in $G^{\times}$. 

\begin{claim}
Let $(u, v) \in V(G) \times V(G)$ be a node in $V(G^{\times})$. Then the degree of $(u, v)$ in $G^{\times}$ is at most
$$
\deg_{G^{\times}}((u, v)) \leq \deg_{G}(u) + \deg_{G}(v).
$$
Moreover, the number of paths in $\Pi^{\times}$ containing node $(u, v)$ is at most
$$
\left|\{(\pi_1, \pi_2) \in \Pi^{\times} \mid (u, v) \in (\pi_1, \pi_2)\}\right| \leq \deg_G(u) \cdot \deg_G(v).
$$
\label{clm:prod_degree}
\end{claim}
\begin{proof}
First, we make an observation about our base graph $G$. Since paths in $\Pi$ are pairwise edge-disjoint in $G$ by Condition 5 of \ref{lem:base_graph}, we immediately obtain that for all $v \in V(G)$, 
$$
|\{ \pi \in \Pi \mid v \in \pi \}| \leq \deg_G(v).
$$
As a direct consequence of our construction, for all $(u, v) \in V(G) \times V(G)$,
$$
\left|\{(\pi_1, \pi_2) \in \Pi^{\times} \mid (u, v) \in (\pi_1, \pi_2)\}\right| = |\{ \pi \in \Pi \mid u \in \pi \}|  \cdot |\{ \pi \in \Pi \mid v \in \pi \}| \leq \deg_G(u) \cdot \deg_G(v),
$$
establishing the second part of the claim. To prove the first part of our claim, we observe that node $(u, v)$ is only adjacent to nodes in $V(G^{\times})$ of the form $(u', v)$ and $(u, v')$, where node $u'$ is adjacent to node $u$ in $G$ and node $v'$ is adjacent to node $v$ in $G$. There can be at most $\deg_G(u) + \deg_G(v)$ many nodes of this form. 
\end{proof}

Let $G^{\times}$ and $\Pi^{\times}$ be the product graph and product paths constructed from a base graph $G$ with $n$ nodes and $|\Pi| = p$ paths. Then by Lemma \ref{lem:prod_graph}, the number of path-node incidences of $G^{\times}, \Pi^{\times}$ is 
$$
|\mathcal{I}_{G^{\times}, \Pi^{\times}}| = 2\ell \cdot |\Pi^{\times}| =  \Theta(\ell p^2) = \Theta(n^{2/3}p^{5/3}). 
$$

We will now prove that graph $G^{\times}$ and paths $\Pi^{\times}$ inherit special ``expansion'' property that we proved graph $G$ and paths $\Pi$ have in Lemma \ref{lem:base_degree}. 

\begin{lemma}
\label{lem:prod_degree}
Let $G^{\times}$ and $\Pi^{\times}$ be the product graph and product paths constructed from a base graph $G$ with $n$ nodes and a set $\Pi$ of $|\Pi| = p$ paths. Let $\Pi' \subseteq \Pi$ be a set of $|\Pi'| = p \cdot n^{-\varepsilon}$ paths in $\Pi$, for some sufficiently small constant $\varepsilon > 0$. Now let $\Pi'' \subseteq \Pi^{\times}$ be a set of paths in $\Pi^{\times}$ contained in the rectangle $\Pi' \times \Pi'$, i.e.,
$$
\Pi'' \subseteq \Pi' \times \Pi' \subseteq \Pi^{\times}.
$$ 
Let $d'' := \frac{2p^{2/3}}{n^{1/3 + \varepsilon/10}}$. Then there exists a subset $\Pi^{(3)} \subseteq \Pi''$ of $\Pi''$ such that the set of all $d''$-degree-bounded path-node incidences between $\Pi^{(3)}$ and $V(G^{\times})$ is at least
$$
\left|\mathcal{I}^{\leq d''}_{V(G^{\times}), \Pi^{(3)}}\right| = \Theta\left(\left|\mathcal{I}_{V(G^{\times}), \Pi^{(3)}}\right|\right) =  \Omega\left( \frac{1}{\log^2 n} \cdot  |\mathcal{I}_{V(G^{\times}), \Pi''}|\right).
$$
\end{lemma}
% \begin{proof}
% The proof of Lemma \ref{lem:prod_degree} straightforwardly follows from Lemma \ref{lem:base_degree} using an intersection argument.
% We begin with some new notations: 
% \begin{itemize}
%     \item Let $\Pi_1 \subseteq \Pi'$ denote the set
%     $$
%     \Pi_1 = \{\pi_1 \in \Pi' \mid (\pi_1, \pi_2) \in \Pi''\}.
%     $$
%     Likewise, let $\Pi_2 \subseteq \Pi'$ denote the set 
%         $$
%     \Pi_2 = \{\pi_2 \in \Pi' \mid (\pi_1, \pi_2) \in \Pi''\}.
%     $$
%     Observe that $\Pi'' \subseteq \Pi_1 \times \Pi_2$. 
% \end{itemize}

% \end{proof}e
\begin{proof}The proof of Lemma \ref{lem:prod_degree} straightforwardly follows from Lemma \ref{lem:base_degree} using an intersection argument.
We begin with some new notations: 
\begin{itemize}
    \item Let $\Pi_1 \subseteq \Pi'$ denote the set
    $$
    \Pi_1 = \{\pi_1 \in \Pi' \mid (\pi_1, \pi_2) \in \Pi''\}.
    $$
    Likewise, let $\Pi_2 \subseteq \Pi'$ denote the set 
        $$
    \Pi_2 = \{\pi_2 \in \Pi' \mid (\pi_1, \pi_2) \in \Pi''\}.
    $$
    Observe that $\Pi'' \subseteq \Pi_1 \times \Pi_2$. 
\end{itemize}
Let $k_1 = \lceil \log p \rceil = \Theta(\log n)$. We will bucket the paths in $\Pi_1$ as follows. For $i \in [1, k_1]$, let $\Pi_1^i \subseteq \Pi_1$ denote the set
    $$
    \Pi_1^i = \{\pi_1 \in \Pi_1 \mid  2^{i - 1} \leq |\Pi'' \cap (\{\pi_1\} \times \Pi_2 )| \leq 2^{i}\}. 
    $$
    By the pigeonhole principle, there must exist an $i \in [1, k_1]$ such that
    $|\Pi'' \cap (\Pi_1^i \times \Pi_2)| \geq \frac{1}{k_1} \cdot |\Pi''| \geq  \frac{1}{2\log n}\cdot |\Pi''| $. 
Now we will need to introduce some additional notation:
\begin{itemize}
    \item Let $G_1 = \graph(V(G), \Pi_1^i)$. 
    \item Given a graph $H$ and a node $v \in V(H)$, We say that $v$ is $H$\textit{-good} if $\deg_{H}(v) \leq d'$,  where $d' = \frac{p^{2/3}}{n^{1/3 + \varepsilon/10}}$ is as defined in Lemma \ref{lem:base_degree}. 
    \item Given a graph $H$, we say that a path $\pi_1 \in \Pi_1^i$ is $H$\textit{-good} if at least $\frac{2}{3}\cdot \ell$ nodes in $\pi_1$ are $H$-good. 
\end{itemize}
Let $\Pi_1' \subseteq \Pi_1^i$ denote the subset
$$
\Pi_1' = \{\pi_1 \in \Pi_1^i \mid \pi_1 \text{ is $G_1$-good}\}.
$$
We claim that $|\Pi_1'| \geq |\Pi_1^i|/4$. To prove this, we will apply Lemma \ref{lem:base_degree} on set $\Pi_1^i$. Observe that $|\Pi_1^i| \leq |\Pi'| = p \cdot n^{-\varepsilon}$, for a sufficiently small constant $\varepsilon > 0$. Consequently, by Lemma \ref{lem:base_degree},
$$
|\mathcal{I}_{V(G), \Pi_1^i}^{\leq d'}| \geq (1 - 1/16) \cdot |\mathcal{I}_{V(G), \Pi_1^i}| = (1 - 1/16) \cdot |\Pi_1^i| \cdot \ell, 
$$
where $d' = d''/2$. Additionally, we can upper bound $|\mathcal{I}_{V(G), \Pi_1^i}^{\leq d'}|$ in terms of $|\Pi_1'|$, as follows:
$$
|\mathcal{I}_{V(G), \Pi_1^i}^{\leq d'}| \leq |\Pi_1'| \cdot \ell + |\Pi_1^i \setminus \Pi_1'| \cdot \frac{2}{3} \cdot \ell.
$$
Putting these inequalities together, we get that
$$
|\Pi_1'| \cdot \ell + |\Pi_1^i \setminus \Pi_1'| \cdot \frac{2}{3} \cdot \ell \geq (1 - 1/16) \cdot |\Pi_1^i| \cdot \ell,
$$
and solving for $|\Pi_1'|$, we get
$$
|\Pi_1'| \geq (1 - 1/16 - 2/3) \cdot |\Pi_1^i| \geq |\Pi_1^i| / 4.
$$
Let $\Pi^{(3)} \subseteq \Pi''$ denote the set
$$
\Pi^{(3)} = \Pi'' \cap (\Pi_1' \times \Pi_2). 
$$
We observe from our choice of $\Pi_1^i$ that
$$
|\Pi^{(3)}| \geq \Pi_1' \cdot 2^i \geq |\Pi_1^i|/4 \cdot 2^i = \frac{1}{8} \cdot |\Pi_1^i| \cdot 2^{i+1} \geq \frac{1}{8} \cdot 
|\Pi'' \cap (\Pi_1^i \times \Pi_2)| \geq \frac{1}{16 \log n}\cdot |\Pi''|.$$
We have shown that given a set $\Pi'' \subseteq \Pi_1 \times \Pi_2$, where $|\Pi_1| \leq p \cdot n^{-\varepsilon}$, we can find a subset $\Pi^{(3)} = \Pi'' \cap (\Pi_1' \times \Pi_2)$  of size $|\Pi^{(3)}| \geq \frac{1}{16 \log n} \cdot |\Pi''|$, such that every path $\pi_1 \in \Pi_1'$ is $G_1'$-good, for graph $G_1' = \graph(V(G), \Pi_1')$.

Now note that $|\Pi_2| \leq |\Pi'| = p \cdot n^{-\varepsilon}$. Then by an argument symmetric to the above argument, we can show that there exists a subset $\Pi^{(4)} \subseteq \Pi^{(3)}$ defined as
$$
\Pi^{(4)} = \Pi^{(3)} \cap (\Pi_1' \times \Pi_2'),
$$
(where $\Pi_2' \subseteq \Pi_2$) such that subset $\Pi^{(4)}$ has size
$$
|\Pi^{(4)}| \geq \frac{1}{16 \log n}\cdot |\Pi^{(3)}|,
$$
and every path $\pi_2 \in \Pi_2'$ is $G_2'$-good, for graph $G_2' = \graph(V(G), \Pi_2')$. 

Putting it all together, we observe that subset $\Pi^{(4)} \subseteq \Pi''$ has size
$$
|\Pi^{(4)}| \geq \frac{1}{16^2 \log^2 n} \cdot |\Pi''|,
$$
and for every path $(\pi_1, \pi_2) \in \Pi^{(4)}$, path $\pi_1$ is $G_1'$-good and path $\pi_2$ is $G_2'$-good.

We will show that subset $\Pi^{(4)} \subseteq \Pi^{(3)}$ is the subset claimed in the statement of Lemma \ref{lem:prod_degree}. Let $G^{(4)} = \graph(V(G^{\times}), \Pi^{(4)})$.

Fix a path $(\pi_1, \pi_2) \in \Pi^{(4)}$. We claim that path $(\pi_1, \pi_2)$ contributes at least $\ell/3$ $d''$-degree-bounded path-node incidences to $\mathcal{I}_{V(G^{\times}), \Pi^{(4)} }^{\leq d''}$. Let $\pi_1 = (u_1, u_2, \dots, u_{\ell})$, and let $\pi_2 = (v_1, v_2, \dots, v_{\ell})$. Since at least $2/3 \cdot \ell$ nodes in $\pi_1$ are $G_1'$-good and at least $2/3 \cdot \ell$ nodes in $\pi_2$ are $G_2'$-good, we conclude that there are at least $1/3 \cdot \ell$ indices $i \in [1, \ell]$ such that $u_i$ is $G_1'$-good and $v_i$ is $G_2'$-good. Let $J \subseteq [1, \ell]$ denote all indices $i \in [1, \ell]$ such that $u_i$ is $G_1'$-good and $v_i$ is $G_2'$-good.  By Claim \ref{clm:prod_degree}, for all $i \in J$, node $(u_i, v_i) \in V(G^{\times})$ has degree at most
$$
\deg_{G^{(4)} }(u_i, v_i) \leq \deg_{G_1'}(u_i) + \deg_{G_2'}(v_i) \leq d''
$$
in $G^{(4)} $. Finally, observe that for all $i \in J$, the node $(u_i, v_i)$ is contained in path $(\pi_1, \pi_2)$, i.e.,  $(u_i, v_i) \in V((\pi_1, \pi_2))$. Then path $(\pi_1, \pi_2) \in \Pi^{(4)}$ contributes at least $1/3 \cdot \ell$ $d''$-degree-bounded path-node incidences to $\mathcal{I}_{V(G^{\times}), \Pi^{(4)} }^{\leq d''}$.

We conclude that
$$
\left|\mathcal{I}_{V(G^{\times}), \Pi^{(4)} }^{\leq d''}\right| \geq |\Pi^{(4)} | \cdot \frac{\ell}{3} \geq \frac{1}{3 \cdot 16^2 \log^2 n} \cdot |\Pi''| \cdot \ell = \frac{1}{3 \cdot 16^2 \log^2 n} \cdot \left|\mathcal{I}_{V(G^{\times}), \Pi''}\right|,
$$
as claimed. 
\end{proof}

\subsection{Matching Decomposition of $G^{\times}, \Pi^{\times}$}

 In this subsection we define what we call a \textit{matching graph} $J$ of $G^{\times}, \Pi^{\times}$ that will be useful for analysis purposes later. We prove that the edges of graph $J$ can be decomposed into a collection of large matchings (Lemma \ref{lem:match_decomp}). The proof of this lemma will follow from the special expansion property of $G^{\times}, \Pi^{\times}$ (Lemma \ref{lem:prod_graph}).  

\label{subsec:match}

\begin{definition}[Matching Graph $J$.]
Let $G^{\times}, \Pi^{\times}$ be the product graph and product paths constructed from a base graph $G$ with $n$ nodes and a set $\Pi$ of $|\Pi| = p$ paths. Fix a subset $\Pi'' \subseteq \Pi^{\times}$ of the product paths. We define the matching graph $J = J(n, p, \Pi'')$ as follows.  

Let $G'' = \graph(V(G^{\times}), \Pi'')$, and fix a node $v \in V(G'')$. Let $U_v$ be the in-neighbors of $v$ in $G''$, and let $W_v$ be the out-neighbors of $v$ in $G''$. We define the graph $H_v$ associated with node  $v$ of $G''$ as follows:
\begin{itemize}
    \item $H_v$ is a bipartite graph with partite sets $U_v$ and $W_v$, and 
    \item $H_v$ contains edge $(u, w) \in U_v \times W_v$ iff there exists a path $\pi$ in $\Pi'$  such that $u, v,  w \in \pi$. 
\end{itemize}
We define the  \textit{matching graph} $J=J(n, p, \Pi'')$ of $G''$ as the \textit{disjoint} union of graphs $\{H_v\}_{v \in V(G'')}$.
Formally, $V(J) = \cup_{v \in V(G'')} (U_v \cup W_v)$ and $E(J) = \cup_{v \in V(G'')} E(H_v)$. 
Note that $J$ is the union of $|V(G'')|$ bipartite subgraphs that are pairwise vertex-disjoint. 
\label{def:matching_graph}
\end{definition}

% Let $G^{\times}, \Pi^{\times}$ be the product graph and product paths constructed from a base graph $G$ with $n$ nodes and a set $\Pi$ of $|\Pi| = p$ paths. Fix a subset $\Pi'' \subseteq \Pi^{\times}$ of the product paths, and let $G'' = \graph(V(G^{\times}), \Pi')$.

% We define the $d$-degree-bounded intersection graph $H^{\leq d}$ of $G'$ as the disjoint union of graphs $H_v$, where $v \in V(G')$ is a node of degree at most $\deg_{G'}(v) \leq d$. 

The key feature of matching graph $J$ that will be useful in our analysis later is summarized in Lemma \ref{lem:match_decomp}. At a high level, Lemma \ref{lem:match_decomp} states that for certain choices of $\Pi''$, matching graph $J(n, p, \Pi'')$ contains a collection of unusually large matchings.

\begin{lemma}
\label{lem:match_decomp}
Fix a sufficiently small constant $\varepsilon > 0$. 
Let $G^{\times}$ and $\Pi^{\times}$ be the product graph and product paths constructed from a base graph $G$ with $n$ nodes and a set $\Pi$ of $|\Pi| = p \leq n^{1+\varepsilon/100}$ paths. Let $\Pi' \subseteq \Pi$ be a set of $|\Pi'| = p \cdot n^{-\varepsilon}$, $\varepsilon > 0$, paths in $\Pi$. Now let $\Pi'' \subseteq \Pi^{\times}$ be a set paths in $\Pi^{\times}$ contained in the rectangle $\Pi' \times \Pi'$, i.e.,
$
\Pi'' \subseteq \Pi' \times \Pi' \subseteq \Pi^{\times}.
$
We require that $|\Pi''| \geq n^{1/3}p^{5/6}$.\footnote{This assumption will be necessary for the proof of Property 2 of Lemma \ref{lem:match_decomp}.}

Additionally, let $G'' = \graph(V(G^{\times}), \Pi'')$ be the graph $G^{\times}$ induced on paths in $\Pi''$, and let $d'' := \frac{2p^{2/3}}{n^{1/3 + \varepsilon/10}}$. Let $J = J(n, p, \Pi'')$ be the matching graph of $G''$ as defined in Definition \ref{def:matching_graph}.  
Then the matching graph $J$ admits a collection of  $k$ matchings, $\mathcal{M} = \{M_1, \dots, M_k\}$, with the following properties:
\begin{enumerate}
    \item For each $i \in [1, k]$, $M_i \subseteq E(J)$ is a valid matching in graph $J$, 
    \item For each $i \in [1, k]$, matching $M_i$ is of size 
    $$
    |M_i| = \Theta\left( \frac{|\mathcal{I}_{V(G''), \Pi''}|}{d''\cdot \log^2 n} \right) = \Theta \left( \frac{|\Pi''| \cdot \ell}{d'' \cdot \log^2 n} \right) = \Theta\left(|\Pi''| \cdot \frac{n^{1+\varepsilon/10}}{p\cdot \log^2 n}\right),
    $$
    \item The number of matchings in our collection is $k = \Theta(d'') = \Theta\left( \frac{p^{2/3}}{n^{1/3 + \varepsilon/10}}\right)$, and
    \item Matchings $M_1, \dots, M_k$ are pairwise edge-disjoint. 
\end{enumerate}
\end{lemma}

For the remainder of this section, we focus on proving \Cref{lem:match_decomp}. We will make use of the following somewhat standard claim about the existence of large matchings in graphs with large average degree. 

\begin{claim}
Let $H$ be a bipartite graph with average degree $d \geq 1$. Then there exists a collection of $k$ pairwise edge-disjoint matchings $M_1, \dots, M_k \subseteq E(H)$  of $H$ such that  $|M_i| = \lceil d/4\rceil$ for all $i \in [1, k]$ and $k=\lceil |V(H)|/4 \rceil$. 
\label{clm:konig}
\end{claim}
\begin{proof}
By Kőnig's Theorem, the size of a maximum matching of $H$ equals the number of vertices in a minimum vertex cover of $H$.  Since $|E(H)| = |V(H)| \cdot d$, and the maximum degree of any node in $H$ is at most $|V(H)|-1$, any vertex cover of $H$ must have size at least
$$
\frac{|E(H)|}{|V(H)| - 1} \geq d. 
$$
Consequently, $H$ contains a matching $M_1$ of size $|M_1| \geq \lceil d/4 \rceil$. To obtain the remaining pairwise edge-disjoint matchings $M_2, \dots, M_k$, we can delete the edges in $M_1$ from $H$ and repeat the above argument. After deleting $i$ matchings $M_1, \dots, M_i$ from $H$, the average degree of $H$ is at least
$$
\frac{|E(H)| - i\cdot \lceil d/4 \rceil}{|V(H)|}=\frac{d|V(H)| - i\cdot \lceil d/4 \rceil}{|V(H)|} \geq \frac{d \cdot (|V(H)| - i )}{|V(H)|}  \geq d/4,
$$
when $i \leq  \lceil |V(H)|/4 \rceil$. Then we can repeat the above argument $k$ times to obtain $k$ pairwise edge-disjoint matchings, each of size at least $\lceil d/4 \rceil$. 
\end{proof}

We now specify our construction of the collection of matchings $\mathcal{M}$.

\paragraph{Constructing the Collection of Matchings $\mathcal{M} = \{M_1, \dots, M_k\}$.}
\begin{itemize}
    \item Let $k = \Theta(d'') = \Theta\left( \frac{p^{2/3}}{n^{1/3 + \varepsilon/10}}\right)$.
    
    \item Initially, we will let each matching $M_i$ be empty, so that $M_i = \emptyset$ for all $i \in [1, k]$. 

    \item By Lemma \ref{lem:prod_degree}, there exists a subset $\Pi^* \subseteq \Pi''$ of $\Pi''$ such that the set of all $d''$-degree-bounded path-node incidences between $\Pi^*$ and $V(G^{\times})$ is at least
$$
\left|\mathcal{I}^{\leq d''}_{V(G^{\times}), \Pi^{*}}\right| =  \Omega\left( \frac{1}{\log^2 n} \cdot  |\mathcal{I}_{V(G^{\times}), \Pi''}|\right).
$$
Let $G^* = \graph(V(G^{\times}), \Pi^*)$, and let $J^* = J(n, p, \Pi^*)$ be the matching graph of $G^*$. Notice that $G^*$ is a subgraph of $G''$, with $V(G^*) = V(G'')$, and $J^* \subseteq J$, where $J = J(n, p, \Pi'')$ is the matching graph of $G''$. As a consequence, any matching of $J^*$ is a matching of $J$ as well. Our matchings in $\mathcal{M}$ will be matchings in graph $J^*$. 
    
    \item Fix a vertex $v \in V(G^*)$, and let $d_v$ denote the average degree of bipartite graph $H_v \subseteq J^*$ as defined in Definition \ref{def:matching_graph} with respect to matching graph $J^* = J(n, p, \Pi^*)$. 
    \item Let $M^v_1, \dots, M^v_{k_v} \subseteq E(H_v)$ denote a collection of $k_v = \min(k, \lceil |V(H_v)|/4 \rceil)$ pairwise edge-disjoint matchings in $H_v$, each of size $|M^v_i| = \lceil d_v/4 \rceil$ for $i \in [1, k_v]$. Note that such a collection of matchings must exist by Claim \ref{clm:konig}. 
    \item Choose a subset $\mathcal{M}' \subseteq \mathcal{M}$ of size $|\mathcal{M}'| = k_v$ uniformly at random from ${\mathcal{M}  \choose k_v}$. Let $\mathcal{M}' = \{M_{i_1}, M_{i_2}, \dots, M_{i_{k_v}}\}$, where $i_1, \dots, i_{k_v} \in [1, k]$. For all $j \in [1, k_v]$, add matching $M_j^v$ to matching $M_{i_j}$, i.e., let 
    $$
    M_{i_j} \leftarrow M_{i_j} \cup M_j^v.
    $$
    \item We complete our construction of the collection of matchings $\mathcal{M} = \{M_1, \dots, M_k\}$ by repeating the above construction for all $v \in V(G^*)$. 
\end{itemize}

\paragraph{Proof of Lemma \ref{lem:match_decomp}.}
We will now analyze our construction of $\mathcal{M}$ to prove  Lemma \ref{lem:match_decomp}. We will verify each property stated in Lemma \ref{lem:match_decomp} individually. 
\begin{enumerate}
    \item To see that each matching $M_i \in \mathcal{M}$ is a valid matching of $J^*$, observe that for each vertex $v \in V(G^*)$, we add at most one matching from our collection of matchings $M^v_1, \dots, M^v_{k_v} \subseteq E(H_v)$ to $M_i$. Then $M_i \cap E(H_v)$ is a matching for all $v \in V(G^*)$. Moreover, subgraphs $\{H_v\}_{v \in V(G^*)}$ of graph $J^*$ are pairwise vertex-disjoint. We conclude that $M_i$ is a valid matching in graph $J^*$. Property 1 is established. 
    \item Property 2 will follow from our construction and Lemma \ref{lem:prod_degree}.
    Fix a matching $M_i$, $i \in [1, k]$. As a first step, we will show that
    $$
    \mathbb{E}[|M_i|] = \Omega\left( \frac{|\mathcal{I}_{V(G''), \Pi''}|}{d'' \cdot \log^2 n} \right).
    $$

    Let $X \subseteq V(G^*)$ denote the set of vertices $x \in V(G^*)$ of degree at most $d''$ in $G^*$. Formally,
    $$
    X = \{ x \in V(G^*) \mid \deg_{G^*}(x) \leq d''\}. 
    $$
    By Lemma \ref{lem:prod_degree},
    $$
    \sum_{x \in X}|\mathcal{I}_{\{x\}, \Pi^*}| = \left|\mathcal{I}^{\leq d''}_{V(G^*), \Pi^*}\right| = \Theta\left( \left|\mathcal{I}_{V(G^*), \Pi^*}\right| \right) = \Omega\left( \frac{1}{\log^2 n} \cdot \left|\mathcal{I}_{V(G''), \Pi''}\right| \right). 
    $$
    Additionally, observe that by Definition \ref{def:matching_graph}, for all $x \in X$,
    \begin{itemize}
        \item  $|V(H_x)| \leq d''$, and
        \item  $|E(H_x)| = |\mathcal{I}_{\{x\}, \Pi^*}|$.
    \end{itemize}
    Combining these two facts, we can bound the expected contribution of graph $H_x$ to the size of matching $M_i$ as follows:
    $$\mathbb{E}[|M_i \cap E(H_x)|] = \Pr[M_i \cap E(H_x) \neq \emptyset] \cdot \lceil d_x/4 \rceil = \frac{k_x}{k} \cdot \lceil d_x/4 \rceil  = \frac{\lceil|V(H_x)|/4\rceil}{k} \cdot \lceil d_x/4 \rceil = \Theta\left( \frac{|\mathcal{I}_{\{x\}, \Pi^*}|}{d''} \right),$$
    where $d_x$ and $k_x$ are as defined in the construction of $\mathcal{M}$. Putting everything together, we get that
    \begin{align*}
       \mathbb{E}[|M_i|] & \geq \sum_{x \in X}\mathbb{E}[|M_i \cap E(H_x)|] =  \sum_{x \in X}\Theta\left( \frac{|\mathcal{I}_{\{x\}, \Pi^*}|}{d''} \right) = \Theta\left( \frac{|\mathcal{I}^{\leq d''}_{V(G''), \Pi^*}|}{d''} \right) = \Theta\left( \frac{|\mathcal{I}_{V(G''), \Pi^*}|}{d''} \right)\\
       & = \Omega\left( \frac{|\mathcal{I}_{V(G''), \Pi''}|}{d'' \cdot \log^2 n} \right).
    \end{align*}

    We will now give a lower bound on $|M_i|$ that holds with high probability, using Hoeffding's inequality. For all $x \in X$, let $Y_x$ denote the random variable $Y_x = |M_i \cap E(H_x)|$. We observe that $$0 \leq Y_x \leq \lceil d_x/4 \rceil \leq |V(H_x)| \leq  d'',$$ where the second to last inequality follows from the fact that the average degree $d_x$ of $H_x$ is at most $|V(H_x)|$.     
    Additionally, we observe that $|X| \leq |V(G'')| =  \Theta(n^{4/3}p^{1/3})$. 
    Let $Z = \sum_{x \in X}Y_x$, and note that $Z \leq |M_i|$. In particular, our previous lower bound on $\mathbb{E}[|M_i|]$ implies that $\mathbb{E}[Z] =\Theta\left( \frac{|\mathcal{I}_{V(G''), \Pi^*}|}{d''} \right) $
as well.  Then by our lower bound on $\mathbb{E}[Z]$ and by Hoeffding's inequality, 
    $$
    \Pr[Z \leq \mathbb{E}[Z]/2] \leq 2\exp\left( - \frac{\mathbb{E}[Z]^2}{8\cdot |X| \cdot d''^2} \right).
    $$
    As we have already proved,
    $$
    \mathbb{E}[Z] =\Theta\left( \frac{|\mathcal{I}_{V(G''), \Pi^*}|}{d''} \right) =  \Theta \left( \frac{|\Pi^*| \cdot \ell}{d''} \right) = \Theta\left(|\Pi^*| \cdot \frac{n^{1+\varepsilon/10}}{p}\right) = \Omega\left( |\Pi^*| \cdot n^{\varepsilon/20} \right), 
    $$
    where the final equality follows from our assumption that $p \leq n^{1+\varepsilon/100}$. Likewise,
    $$
    8 \cdot |X| \cdot d''^2 \leq 8 \cdot |V(G'')| \cdot \frac{p^{4/3}}{n^{2/3}} = \Theta(n^{2/3}p^{5/3}).
    $$
    Continuing with our calculations,
    $$
    \Pr[Z \leq \mathbb{E}[Z]/2] \leq 2\exp\left( - \frac{\mathbb{E}[Z]^2}{8\cdot |X| \cdot d''^2} \right) \leq \exp\left(- \Omega\left( \frac{|\Pi^*|^2n^{\varepsilon/10}}{n^{2/3}p^{5/3}} \right)\right) \leq \exp(-\Omega(n^{\varepsilon/20})),
    $$
    where the final inequality follows the fact that $$|\Pi^*| = \Omega\left( \frac{1}{\log^2 n}|\Pi''|\right) = \Omega\left(\frac{n^{1/3}p^{5/6}}{\log^2 n}\right),$$ 
     from our assumption in the statement of Lemma \ref{lem:match_decomp} that $|\Pi''| \geq n^{1/3}p^{5/6}$. 

    We have shown that matching $|M_i|$ has size at least $\mathbb{E}[Z]/2 = \Omega\left( \frac{|\mathcal{I}_{V(G''), \Pi''}|}{d'' \cdot \log^2 n} \right) $ with probability at least $1 - \exp(-\Omega(n^{\varepsilon/20}))$. By union bounding over all $k = \Theta(d'') \leq n$ matchings in $\mathcal{M}$, we conclude that all of our  matchings $M \in \mathcal{M}$ have size $|M| = \Theta\left( \frac{|\mathcal{I}_{V(G''), \Pi''}|}{d'' \cdot \log^2 n} \right)$ with high probability (specifically, with probability at least $1 - n \cdot \exp(-\Omega(n^{\varepsilon/20}))$).

    \item Property 3 is immediate from our construction.
    \item Property 4 immediately follows from the fact that for all $v \in V(G^*)$, our collection of matchings $M^v_1, \dots, M^v_{k_v} \subseteq E(H_v)$ in $H_v \subseteq J^*$ are pairwise edge-disjoint. 
\end{enumerate}

\subsection{Random Graph Construction $G^*, \Pi^*$}
\label{subsec:random}
 In subsection \ref{subsec:random}, we define a random graph $G^*$ and a collection of shortest paths $\Pi^*$ constructed from $G^{\times}, \Pi^{\times}$. Graph $G^*$ and paths $\Pi^*$ will be constructed in a similar manner as in our lower bound for DAG covers with $f(n)$ additional edges in subsection \ref{subsec:fn_lb}. 

Let $G^{\times} = (V(G^{\times}), E(G^{\times}), w^{\times})$ and $\Pi^{\times}$ be the product graph and product paths constructed from a base graph $G = (V(G), E(G), w_G)$ with $n$ nodes and a set $\Pi$ of $|\Pi| = p$ paths. Let $c\geq 2$ be a positive integer parameter.  
We define a random cyclic graph $G^* = G^*(n, p, c)$ with associated weight function $w^*$ as follows:
\begin{itemize}
    \item Replace each node $v$ in $V(G^{\times})$ with a copy of the clique graph $K_c$ with bidirectional edges and $c$ nodes. Denote the clique replacing node $v$ as $K_c^v$.
    \item  Fix an edge $(u, v) \in E(G^{\times})$. Let $u^* \in V(K_c^u)$ be a node  sampled uniformly at random from $K_c^u$. Likewise, let $v^* \in V(K_c^v)$ be a node sampled uniformly at random from $K_c^v$. Add edge $(u^*, v^*)$ to graph $G^*$.  Repeat this random process for every edge in $E(G^{\times})$. 
    \item We will now assign a weight $w^*(e)$ to every edge $e$ in $E(G^*)$.   For every edge $e \in E(K_c^v)$ in a clique subgraph $K_c^v$, assign edge $e$ weight $w^*(e) := 1$ in $G^*$. For every edge $e \in E(V(K_c^u), V(K_c^v))$ between distinct clique subgraphs $K_c^u$ and $K_c^v$, where $u \neq v$, assign edge $e$ weight $w^*(e) := 10 \ell \cdot w^{\times}((u, v))$, where $\ell = \frac{n^{2/3}}{p^{1/3}}$ as in Lemma \ref{lem:prod_graph}.
\end{itemize}
This completes the construction of the random graph $G^*$. We now define an associated set of paths $\Pi^*$ in $G^*$. Each path in $\Pi^*$ will be an image of a path in $\Pi^{\times}$. 
Note that each edge $(u, v) \in E(G^{\times})$ corresponds to a unique edge $(u^*, v^*)$ in $E(G^*)$ between cliques $K_c^u$ and $K_c^v$. Let $\phi:E(G^{\times}) \mapsto E(G^*)$ denote this injective mapping. 

Now fix a path $\pi$ in $\Pi^{\times}$, and let $\pi = (v_1, \dots, v_{2\ell})$, where $\ell = \frac{n^{2/3}}{p^{1/3}}$  as in Lemma \ref{lem:prod_graph}. For each edge $(v_i, v_{i+1}) \in E(\pi)$, let $e_i := \phi((v_i, v_{i+1}))$. Note that $e_i := (x_i, y_i) \subseteq V(K_c^{v_i}) \times V(K_c^{v_{i+1}})$, so $y_i, x_{i+1} \in V(K_c^{v_{i+1}})$, and in particular, either $(y_i, x_{i+1}) \in E(K_c^{v_{i+1}})$ or $y_i = x_{i+1}$. Then we define the image path $\pi^*$ in $G^*$ as
$$
\pi^* = (x_1, y_1, x_2, y_2, \dots, x_i, y_i, \dots, x_{2\ell-1}, y_{2\ell-1}).
$$
We will use $\psi(\pi)$ to denote the image path $\pi^*$ of $\pi$ in $G^*$. Given a subset $\Pi' \subseteq \Pi^{\times}$, we will use $\psi[\Pi']$ to denote the image of set $\Pi'$ under function $\psi$, i.e., $\psi[\Pi'] = \{\psi(\pi) \mid \pi \in \Pi'\}$.
We define $\Pi^*$ as the collection of image paths $\psi(\pi)$ of paths $\pi$ in $\Pi^{\times}$, so
$
\Pi^* = \psi[\Pi^{\times}].
$

\begin{lemma}
    Let $G^* = G^*(n, p, c)$ be the random cyclic graph described above, and let $\Pi^*$ be the associated set of  paths.  Then $G^{*}, \Pi^{*}$ satisfy the following conditions:
    \begin{enumerate}
        \item the parameters of $G^{*}, \Pi^{*}$ are:
        \begin{itemize}
            \item $|V(G^{*})| =  \Theta\left(c \cdot \frac{n^2}{\ell}\right)$, where $\ell = \frac{n^{2/3}}{p^{1/3}}$,
            \item $|E(G^{*})| = \Theta\left(\frac{n^2}{\ell} \cdot (d+c^2)\right)$, where $d = \frac{p^{2/3}}{n^{1/3}}$,
            \item  $|\Pi^{*}| = \Theta(p^2)$,
        \end{itemize}
        \item every path $\pi \in \Pi^{*}$ contains $\Theta(\ell)$ nodes,
        \item every path $\pi \in \Pi^*$ is a unique shortest path in $G^*$, 
        and
        \item intersecting paths in $\Pi^{*}$ intersect in either one clique $K_c^v$ or two adjacent cliques $K_c^u$ and $K_c^v$, where $(u, v) \in E(G^{\times})$.
    \end{enumerate}
    \label{lem:rand_graph}
\end{lemma}
\begin{proof}
Properties 1, 2, and 4 follow immediately from the construction of $G^*, \Pi^*$ and Properties 1, 3, and 5 of Lemma \ref{lem:prod_graph}, respectively. What remains is to prove Property 3, which will follow from Property  4 of Lemma \ref{lem:prod_graph} and our weighting scheme $w^*$ of $G^*$. 

Fix an $s \leadsto t$ path $\pi^* \in \Pi^*$, and suppose towards contradiction that there exists an $s \leadsto t$ path $\pi_1^*$ such that $\pi_1^* \neq \pi^*$ and $w^*(\pi_1^*) \leq w^*(\pi^*)$. 

Contract each clique $K_c^v$ of $G^*$ back into a single node $v$. The resulting graph will be identical to product graph $G^{\times}$.  Let $\pi$ (and $\pi_1$, respectively) denote the image of path $\pi^*$ (and path $\pi_1^*$) after performing this contraction operation. We now have two cases:
\begin{itemize}
    \item \textbf{Case 1: $\pi_1 = \pi$.} In this case, we know that paths $\pi^*$ and $\pi_1^*$ travel through exactly the same sequence of cliques $K_c^v$ in $G^*$. Let $K_c^{v_1}, \dots, K_c^{v_{2\ell}}$ denote the sequence of cliques passed through by paths $\pi^*$ and $\pi_1^*$ in $G^*$. By our construction of $G^*$, there is exactly one edge between clique $K_c^{v_i}$ and $K_c^{v_{i+1}}$ in $G^*$. Formally,
    $$
    |E(G^*) \cap (V(K_c^{v_i}) \times V(K_c^{v_{i+1}}))| = 1,
    $$
    for all $i \in [1, k-1]$. As a consequence, paths $\pi^*$ and $\pi_1^*$ must take exactly the same edge $E(V(K_c^{v_i}), V(K_c^{v_{i+1}}))$ between cliques $K_c^{v_i}$ and $K_c^{v_{i+1}}$.
    We will argue that if $w^*(\pi_1^*) \leq w^*(\pi^*)$, then path $\pi^*$ and path $\pi_1^*$ will use exactly the same set of edges inside each clique $K_c^v$ of $G^*$, and so $\pi_1^* = \pi^*$. For any clique $K_c^v$ in $G^*$, the subpath $\pi^* \cap V(K_c^v)$ of path $\pi^*$ restricted to clique $K_c^v$ is either empty, or consists of a single node in $V(G^*)$ or an edge in $E(G^*)$. Then by our construction of $G^*$ and clique subgraphs $K_c^v$,  any nonempty subpath $\pi^* \cap V(K_c^v)$ of path $\pi^*$ is a unique shortest path between its endpoints in $G^*$. Then if $w^*(\pi_1^*) \leq w^*(\pi^*)$, we must have that $\pi_1^* \cap V(K_c^v) = \pi^* \cap V(K_c^v)$ for all $v \in V(G^{\times})$. Then in fact, $\pi_1^* = \pi^*$, contradicting our assumption that $\pi_1^* \neq \pi^*$. We conclude that in this case, $w^*(\pi_1^*) > w^*(\pi^*)$.
    \item     \textbf{Case 2: $\pi_1 \neq \pi$.} Now observe that by Lemma \ref{lem:prod_graph}, path $\pi$ is a unique shortest path in $G^{\times}$, so we have that $w^{\times}(\pi) \leq w^{\times}(\pi_1) - 1$, since $\pi_1 \neq \pi$ (recall that  $w^{\times}$ is the weight function for product graph $G^{\times}$). We will use the above observation to show that $w^*(\pi_1^*) > w^*(\pi^*)$, as desired. We partition the edges in $E(G^*)$ into two sets, $E_1$ and $E_2$, as follows. Let
    $$
    E_1 = \{ e \in E(G^*) \mid e \in E(K_c^v) \text{ for some $v \in V(G^{\times})$} \}
    $$
    denote the set of edges in $E(G^*)$ inside clique subgraphs $K_c^v$ of $G^*$. Let
    $$
    E_2 = \{ e \in E(G^*) \mid e \in E(V(K_c^{u}), V(K_c^{v})) \text{ where $u, v \in V(G^{\times})$ and $u \neq v$} \}
    $$
    denote the set of edges in $E(G^*)$ between distinct clique subgraphs $K_c^u, K_c^v$ in $G^*$. We observe that $E_1, E_2$ partitions the edges of $G^*$ as desired. We now measure the contribution of each edge set to the weights $w^*(\pi_1^*), w^*(\pi^*)$ of paths $\pi_1^*, \pi^*$.  We first observe that paths $\pi_1^*$ and $\pi^*$ pass through $2\ell$ cliques, and they use at most one clique edge in each clique. Consequently,
    $$
    0 \leq w^*(E(\pi_1^*) \cap E_1) \leq 2\ell  \text{\quad and \quad} 0 \leq w^*(E(\pi^*) \cap E_1) \leq 2\ell.
    $$
    Likewise, using our earlier observation and our choice of weight function $w^*$,
    $$
    w^*(E(\pi^*) \cap E_2) = 10 \ell \cdot w^{\times}(\pi) \leq 10 \ell \cdot (w^{\times}(\pi_1) - 1) = 10 \ell \cdot w^{\times}(\pi_1) - 10 \ell = w^*(E(\pi_1^*) \cap E_2) - 10\ell.
    $$
    We conclude that
\begin{align*}
    w^*(\pi^*) & = w^*(E(\pi^*) \cap E_1) + w^*(E(\pi^*) \cap E_2) \\
    & \leq 2\ell +  w^*(E(\pi_1^*) \cap E_2) - 10\ell \\
    & \leq w^*(\pi_1^*) - 8\ell,
\end{align*}
so $\pi^*$ is a unique shortest path in $G^*$, as claimed.
\end{itemize}
\end{proof}

Like in \Cref{subsec:fn_lb}, we now define a family of DAGs $\mathcal{D}$ associated with our distribution of random graphs $G^*$ that will be useful in our analysis.

\paragraph{DAG Family $\mathcal{D}$.} We begin by defining a graph $\hat{G} = (V^*, \hat{E})$  over the same vertex set as $G^*$, and containing every possible edge in $G^*$. Formally, we define edge set $\hat{E}$ to be 
$$
\hat{E} = \{ V(K_2^u) \times V(K_2^v) \mid (u, v) \in E(G^{\times}) \} \cup \{ E(K_2^v) \mid  v \in V \}.
$$
Note that $G^* \subseteq \hat{G}$ for each graph $G^*$ in the support of our random graph distribution.

Our DAG family $\mathcal{D}$ will be the set of all edge-maximal acyclic subgraphs of $\hat{G}$. 
% For analysis purposes, it will be useful to define a family of DAGs $\mathcal{D}$ which roughly correspond to DAGs  total orderings of the vertex set $V(G^*)$. 
We define a DAG $D$ of $\hat{G}$ as follows. For each clique $K_c^v$, fix a total order on its $c$ vertices, and delete the edges in $K_c^v$ that do not respect this total order. After repeating this process for all cliques $K_c^v$, the resulting graph will be a DAG, since the product graph $G^{\times}$ is a DAG. Define $\mathcal{D}$ to be the family of DAGs of $G^*$ that can be constructed from the above process. Note that since there are  $c!$ choices of total orderings for each clique $K_c$, it follows that the size of $\mathcal{D}$ is 
$$
|\mathcal{D}| = (c!)^{|V(G^{\times})|} = 2^{\Theta(n^2 / \ell \cdot c \log c)} = 2^{\Theta(n^{4/3}p^{1/3} \cdot c \log c)}.
$$

For each DAG $D \in \mathcal{D}$, the graph $D \cap G^*$ is an acyclic subgraph of random graph $G^*$.\footnote{We use $D \cap G^*$ to  denote the graph with vertex set $V^*$ and edge set $ E(D) \cap E(G^*)$.} This allows us to reason about DAG covers of $G^*$ using the DAGs in graph family $\mathcal{D}$. We note that while $G^*$ is a random graph, the DAGs in $\mathcal{D}$ are fixed and depend only on product graph $G^{\times}$.  
The next claim  follows directly from the structure of graph $G^*$ as a DAG of strongly connected components $K_2^v$.

\begin{claim}
For each acyclic subgraph $D$ of $G^*$, there is a DAG $D' \in \mathcal{D}$ such that $D \subseteq D' \cap G^*$. \label{clm:subdag2}
\end{claim}
\begin{proof}
The proof of this claim follows from an identical argument as in the proof of Claim \ref{clm:subdag}.
\end{proof}

% We now verify that our random graph $G^*$ inherits the unique shortest path property of our original product graph $G^{\times}$ (Property 4 of Lemma \ref{lem:prod_graph}).

% \begin{claim}
%     Every path $\pi^* \in \Pi^*$ is a unique shortest path in $G^*$. 
% \end{claim}
% \begin{proof}
%     This claim follows from an identical argument as in Claim \ref{clm:uniqueness}, using the unique path property of  $G^{\times}, \Pi^{\times}$, as stated in Property 4 of Lemma \ref{lem:prod_graph}.
% \end{proof}

\subsection{Finishing the Proof}
\label{subsec:analysis}
Let $G^* = G^*(n, p, c)$ be our random cyclic graph constructed in the previous subsection, with construction parameters $p, c$ to be specified later. Let $w^*$ be the edge weight function associated with $G^*$, and let $\Pi^*$ be the collection of critical paths associated with $G^*$. Likewise, let $\mathcal{D}$ denote our collection of DAGs satisfying \Cref{clm:subdag2}. 

Graph $G^*$  and paths $\Pi^*$ are constructed from a corresponding product graph $G^{\times}$ and set of paths $\Pi^{\times}$, which are in turn constructed from a base graph $G$ and set of paths $\Pi$. In this subsection, we will occasionally reference $G^{\times}, \Pi^{\times}, G$, and  $\Pi$, as they relate to the construction of $G^*, \Pi^*$.

Recall that our goal in proving Theorem \ref{thm:lb_m} is the following. For any exact distance-preserving DAG cover $D_1, \dots, D_g$ of $G^*$ that uses a set $E' \subseteq TC(G^*)$ of additional edges of size $|E'| = O(m^{1+\varepsilon})$, for a sufficiently small $\varepsilon > 0$, we have that $g = \Omega(n^{1/6})$. Towards this goal, it will be useful in our analysis to split the set of additional edges $E' \subseteq TC(G^*)$ into two categories: ``short edges'' and ``long edges''. We may assume without loss of generality that $E' \subseteq TC(G^*) \setminus E(G^*)$.
\begin{itemize}
    \item \textbf{Short Edges:} We say an edge $e \in TC(G^*) \setminus E(G^*)$ is a \textit{short} edge if
    edge $e$ is between two clique subgraphs $K_c^u$ and $K_c^v$ that are adjacent in $G^*$. Formally, if $e$ is short, then
    $e \in V(K_c^u) \times V(K_c^v)$ where  $(u, v) \in E(G^{\times})$.
    \item \textbf{Long Edges:} We say an edge $e \in TC(G^*) \setminus E(G^*)$ is a \textit{long} edge if $e$ is between two cliques $K_c^u$ and  $K_c^v$ that do not have an edge between them in $E(G^*)$. Equivalently, $e$ is long if it is not short. 
\end{itemize}

The reason it is helpful to distinguish between short edges and long edges is that their intersection behavior with the paths in $\Pi^*$ is very different. As stated in Claim \ref{obs:long_edge}, every long edge in $E'$ is contained in the transitive closure of at most one path in $\Pi^*$. On the other hand,  short edges in $E'$ are contained in the transitive closure of up to $d/c$ distinct paths in $\Pi^*$ in general (see Claim \ref{obs:short}).

Since long edges have this simple intersection behavior with paths in $\Pi^*$, it is easier to handle them in our analysis. With this in mind, we first give a proof of Theorem \ref{thm:lb_m} for the special case where our set of additional edges consists only of long edges. 

\subsubsection{Proof of Theorem \ref{thm:lb_m} when $E'$ consists only of long edges}

The goal of this subsubsection is to prove the following special case of Theorem \ref{thm:lb_m}. 
\begin{theorem}[Long edges only]
\label{thm:long}
    Let $D_1, \dots, D_k$ be an exact distance-preserving DAG cover of $G^*$ with a set of additional edges $E' \subseteq TC(G^*)$ of size $|E'| = m^{1+\varepsilon}$ for a sufficiently small constant $\varepsilon > 0$. Assume that $E'$ consists only of long edges. Then $k = \Omega(|V(G^*)|^{1/6})$. 
\end{theorem}

Our proof of Theorem \ref{thm:long} will rely on the following property of edges in $E'$. 

\begin{claim}
 Every long edge $e \in E'$ is contained in the transitive closure of at most one path in $\Pi^*$. 
\label{obs:long_edge}
\end{claim}
\begin{proof}
This claim follows immediately from  Property 4 of Lemma \ref{lem:rand_graph}.
\end{proof}

We will reuse the following claim from our  proof of  Theorem \ref{thm:dag_cover_n_setting}.

\begin{claim}
    Fix a DAG $D$ in $\mathcal{D}$. Fix two edges $(u, v), (v, w) \in E(G^{\times})$ in graph $G^{\times}$, and let $\phi(u, v) = (s, t)$ and $\phi(v, w) = (x, y)$, where function $\phi:E(G^{\times}) \mapsto E(G^*)$ is as defined in the previous subsection. Then path $\pi^* = (s, t, x, y)$ is contained in DAG $D \cap G^*$ with probability at most $3/4$.
    \label{clm:small_paths2}
\end{claim}
\begin{proof}
    This claim follows from an argument identical to that of Claim \ref{clm:small_paths}. 
\end{proof}

Claim \ref{clm:small_paths2} states that certain short paths in $G^*$ survive in $D \cap G^*$, where $D \in \mathcal{D}$, with probability at most $3/4$. In the following claim, we will use Claim \ref{clm:small_paths2} to argue that  $D \cap G^*$ contains  many paths in $\Pi^*$ with very low probability. 

\begin{claim}
Fix a sufficiently small $\varepsilon > 0$, and assume that construction parameter $p$ is chosen so that $p \leq n^{1 + \varepsilon/100}$.
    Fix a DAG $D \in \mathcal{D}$. Let $\Pi' \subseteq \Pi$ be a set of $|\Pi'| = p \cdot n^{-\varepsilon}$ paths in $\Pi$, where $\Pi$ is the set of base paths constructed in subsection \ref{subsec:base}. Let $\Pi'' \subseteq \Pi^{\times}$ be a set of $|\Pi''| = k \geq n^{1/3}p^{5/6}$ paths in $\Pi^{\times}$ contained in the rectangle $\Pi' \times \Pi'$, i.e.,
    $$
    \Pi'' \subseteq \Pi' \times \Pi' \subseteq \Pi^{\times}. 
    $$
    Finally, let $\Pi^{(3)} = \psi[\Pi''] \subseteq \Pi^*$ denote the image of set $\Pi''$ under function $\psi$, as defined in the previous subsection. The probability that DAG $D \cap G^*$ contains all $|\Pi^{(3)}| = k$ paths in $\Pi^{(3)}$ as subgraphs is at most
    $$
    \left(\frac{3}{4}\right)^{\Theta\left(k \cdot \frac{n^{1+\varepsilon/10}}{p \cdot \log^2 n}\right)}.
    $$
    \label{clm:matching_bound}
\end{claim}
\begin{proof}
Let graph $G'' = \graph(V(G^{\times}), \Pi'')$ be the graph $G^{\times}$ induced on paths in $\Pi''$, and let $J = J(n, p, \Pi'')$ be the matching graph of $G''$ as defined in Definition \ref{def:matching_graph}. By Lemma \ref{lem:match_decomp}, the matching graph $J$ admits a collection of $\alpha$ matchings $M_1, \dots, M_{\alpha}$, with each matching of size $|M_i| = \beta$, where $\alpha = \Theta\left(\frac{p^{2/3}}{n^{1/3+\varepsilon/10}}\right)$ and $\beta = \Theta\left( k \cdot \frac{n^{1+\varepsilon/10}}{p \cdot \log^2 n}  \right)$. For the remainder of the proof of this claim, we will only use the single matching $M_1$ from our collection of matchings.

Let $(u, w) \in  M_1$ be an edge in matching $M_1$ that is contained in a subgraph $H_v$ of the matching graph $J$. That is $(u, w) \in E(U_v, W_v) = E(H_v) \subseteq E(J)$. Then by the definition of matching graph $J$, this implies that edges $(u, v), (v, w) \in E(G'') \subseteq E(G^{\times})$ are contained in a single path in $\Pi''$. 
We will let $\pi \in \Pi''$ denote the path such that  $(u, v), (v, w) \in E(\pi)$. Let $\phi(u, v) = (s, t)$ and $\phi(v, w) = (x, y)$, where function $\phi:E(G^{\times}) \mapsto E(G^*)$ is as defined in the previous subsection; note that $\psi(\pi) \in \Pi^{(3)}$.  Now observe that path $(s, t, x, y)$ is a subpath of the path $\psi(\pi)$, where function $\psi:\Pi^{\times} \mapsto \Pi^*$ is as defined in the previous subsection. Fix a DAG $D \in \mathcal{D}$. If DAG $D \cap G^*$ contains the path $\psi(\pi) \in \Pi^{(3)}$, it must also contain the subpath $(s, t, x, y)$. This event can happen with probability at most $3/4$ by Claim \ref{clm:small_paths2}.

We can repeat this argument for every edge in matching $M_1$. 
Let $e_1, \dots, e_{\beta} \in M_1$ denote the edges in  $M_1$. For each edge $e_i \in M_1$, where $i \in [1, \beta]$, we can identify two corresponding edges $(u_i, v_i), (v_i, w_i) \in E(G'')$ such that $e_i = (u_i, w_i)$. Additionally, we can identify a path $(s_i, t_i, x_i, y_i)$ such that
\begin{enumerate}
    \item $\phi(u_i, v_i) = (s_i, t_i) \in E(G^*)$ and $\phi(v_i, w_i)= (x_i, y_i) \in E(G^*)$, and
    \item $(s_i, t_i, x_i, y_i)$ is a subpath of a path $\pi^*_i \in \Pi^{(3)}$.
\end{enumerate}
Then if  DAG $D \cap G^*$ contains every path in $\Pi^{(3)}$, it must contain every subpath $(s_i, t_i, x_i, y_i)$ for $i \in [1, \beta]$, as well.

For each $i \in [1, \beta]$, let $A_i$ be the event that DAG $D \cap G^*$ contains subpath $(s_i, t_i, x_i, y_i)$. We claim that events $A_i$ and $A_j$ are independent for $i \neq j$. Edges $(s_i, t_i), (x_i, y_i), (s_j, t_j), (x_j, y_j)$ are always contained in DAG $D \cap G^*$, by the construction of $G^*$ and \Cref{clm:subdag2}. 
Then we must show that the event that subpath $(t_i, x_i)$ survives in  $D \cap G^*$  is independent of the event that subpath $(t_j, x_j)$ survives in  $D \cap G^*$. Since $M_1$ is a matching in $J$, the edges $(u_i, v_i), (v_i, w_i), (u_j, v_j), (v_j, w_j) \in E(G'')$ in $G''$ are all distinct. Recall that the random mapping $\phi:E(G^{\times}) \mapsto E(G^*)$ maps each edge $(u, v) \in E(G^{\times})$ to an edge $\phi((u, v)) = (u^*, v^*) \in E(G^*)$, where $(u^*, v^*)$ is chosen uniformly at random from $V(K_c^u) \times V(K_c^v)$. As a consequence, since edges $(u_i, v_i), (v_i, w_i), (u_j, v_j), (v_j, w_j) \in E(G'')$ are all distinct, the event that random mapping $\phi$ preserves subpath $(t_i, x_i)$ in  $D \cap G^*$  is independent of the event that random mapping $\phi$ preserves subpath $(t_j, x_j)$ in  $D \cap G^*$.  We conclude that events $A_i$ and $A_j$ are independent for $i \neq j$. 
Putting it all together, we find that
$$
\Pr[\pi^* \subseteq  D \cap G^* \text{ for all $\pi^* \in \Pi^{(3)}$}] \leq \Pr[\cap_i A_i] = \left( \frac{3}{4}\right)^{\beta} =  \left(\frac{3}{4}\right)^{\Theta\left(k \cdot \frac{n^{1+\varepsilon/10}}{p \cdot \log^2 n}\right)},
$$
as claimed.
\end{proof}

Notice how in the proof of Claim \ref{clm:matching_bound}, we only used the matching $M_1$ from our collection of matchings $M_1, \dots, M_{\alpha}$. In our final proof of Theorem \ref{thm:lb_m} where we incorporate short edges, we will make use of all the matchings in the matching decomposition.

We will now combine the very strong probability bound proved in Claim \ref{clm:matching_bound} with a union bound to argue that with nonzero probability there does not exist a DAG $D \in \mathcal{D}$  such that $D \cap G^*$  contains many paths in $\Pi^*$ as subgraphs.

\begin{lemma}
    Fix a sufficiently small $\varepsilon > 0$, and assume that construction parameter $p$ is chosen so that $p \leq n^{1 + \varepsilon/100}$.
    With nonzero probability the following statement holds. 
    For every DAG $D \in \mathcal{D}$, the graph $D \cap G^*$  contains at most
    $$
    n^{1/3 +2\varepsilon}p^{4/3} c \log c
    $$
    distinct paths in $\Pi^*$ as subgraphs. 
    \label{lem:path_cover2}
\end{lemma}
\begin{proof}
    Our strategy for proving this lemma is to use Claim \ref{clm:matching_bound} and apply the union bound. Before we can do this, we need to develop some additional notation. 

    Let $\alpha = n^{\varepsilon}$. Let $\Pi_1, \dots, \Pi_{\alpha}$ be a partition of the set of base paths $\Pi$ into $\alpha$ sets $\Pi_i$ each of size $|\Pi_i| \leq |\Pi|/\alpha \leq p \cdot n^{\varepsilon}$. For all $i, j \in [1, \alpha]$, we define the subset $\Pi_{i, j}^{\times} \subseteq \Pi^{\times}$ as follows:
    $$
    \Pi_{i, j}^{\times} := \Pi_i \times \Pi_j.
    $$
    For all $i, j \in [1, \alpha]$, we define $\Pi^*_{i, j} \subseteq \Pi^*$ to be the image of set $\Pi_{i, j}^{\times}$ under function $\psi$. Formally, $\Pi^*_{i, j} = \psi[\Pi^{\times}_{i,j}]$. Notice that the collection of sets $\{\Pi^*_{i, j}\}_{i, j \in [1, \alpha]}$ is a partition of set $\Pi^*$.  
    Additionally, since each set $\Pi^*_{i, j}$ is contained inside a small rectangle $\psi[\Pi_{i} \times \Pi_j]$, we can directly apply Claim \ref{clm:matching_bound} on each set $\Pi^*_{i, j}$. 

    Towards proving Lemma \ref{lem:path_cover2}, we will show the following. With nonzero probability, for every $D \in \mathcal{D}$, graph $D \cap G^*$ contains at most 
    $$
    n^{1/3}p^{4/3}c \log c
    $$
    distinct paths from a set $\Pi^*_{i, j}$ (where $i, j \in [1, \alpha]$) as subgraphs. Once we prove this, Lemma \ref{lem:path_cover2} will follow straightforwardly from an averaging argument.

    By Claim \ref{clm:matching_bound} and the union bound, the probability that there exists:
    \begin{enumerate}
        \item a DAG $D \in \mathcal{D}$,
        \item a set $\Pi^*_{i, j}$, $i, j \in [1, \alpha]$, and
        \item  a subset $\Pi' \subseteq \Pi^*_{i, j}$ of $\Pi^*_{i, j}$ of size $|\Pi'| = k$
    \end{enumerate}
    such that DAG $D \cap G^*$ contains every path in $\Pi'$ as a subpath is at most
    $$
    |\mathcal{D}| \cdot \alpha^2 \cdot {|\Pi_{i, j}^*| \choose k} \cdot \left( \frac{3}{4}\right)^{\Theta\left( k \cdot \frac{n^{1+\varepsilon/10}}{p \cdot \log^2 n} \right)} \leq 2^{\Theta(n^{4/3}p^{1/3} \cdot c \log c)} \cdot n^{2\varepsilon} \cdot n^{\Theta(k)} \cdot 2^{-\Theta\left( k \cdot \frac{n^{1+\varepsilon/10}}{p \cdot \log^2 n} \right)} < 1,
    $$
    when $k \geq n^{1/3}p^{4/3}c \log c$. Then by the probabilistic method we may assume that no DAG $D \cap G^*$, where $D \in \mathcal{D}$, contains more than $n^{1/3}p^{4/3}c \log c$ distinct paths from any set $\Pi^*_{i, j}$, where $i, j \in [1, \alpha]$. 

    Now to see why Lemma \ref{lem:path_cover2} follows, suppose towards contradiction that there exists a DAG $D \in \mathcal{D}$ such that $D \cap G^*$ contains more than $n^{1/3 + 2\varepsilon}p^{4/3}c \log c$ distinct paths in $\Pi^*$ as subpaths. Then by an averaging argument, there must exist a set $\Pi^*_{i, j} \subseteq \Pi^*$, $i, j \in [1, \alpha]$, such that DAG $D \cap G^*$ contains more than 
    $$
    \frac{n^{1/3 + 2\varepsilon}p^{4/3}c \log c}{\alpha^2} = n^{1/3}p^{4/3}c \log c
    $$
    distinct paths in $\Pi^*_{i, j}$, contradicting our earlier discussion. 
\end{proof}

We are finally ready to prove our special case of Theorem \ref{thm:lb_m}.

\begin{proof}[Proof of Theorem \ref{thm:long}]
Let $\varepsilon > 0$ be a sufficiently small constant. We will choose our construction parameters $p$ and $c$ of $G^*(n, p, c)$ to be $p = n^{1+\varepsilon/100}$ and $c = 2$.  Notice that under this choice of construction parameters, graph $G^*$ and paths $\Pi^*$ have the following properties:
\begin{itemize}
    \item $\ell = \frac{n^{2/3}}{p^{1/3}} = \Theta\left( n^{1/3 -\varepsilon/300} \right)$ and $d = \frac{p^{2/3}}{n^{1/3}} = \Theta(n^{1/3+\varepsilon/150})$, 
    \item $|V(G^*)| = \Theta\left(\frac{n^2}{\ell}\right) = \Theta(n^{5/3 + \varepsilon/300})$,
    \item $|E(G^*)| = \Theta\left( \frac{n^2}{\ell} \cdot d \right) = \Theta(n^{2+\varepsilon/100})$, and
    \item $|\Pi^*| = \Theta(p^2) = \Theta(n^{2+\varepsilon/50})$. 
\end{itemize}
Let $N = |V(G^*)|$ and let $m = |E(G^*)|$. Let $D_1, \dots, D_k$ be an exact distance-preserving DAG cover with a set $E'$ of at most $|E'| \leq m^{1+\varepsilon/300} \leq |\Pi^*|/2$ additional edges.  Recall that we are assuming all additional edges $E' \subseteq TC(G^*)$ are long edges.  We define a subset $\Pi_1^* \subseteq \Pi^*$ of paths in $\Pi^*$ as follows
 $$\Pi_1^* = \{\pi^* \in \Pi^* \mid TC(\pi^*) \cap E' = \emptyset \}.$$ 
 We observe that $|\Pi_1^*| \geq |\Pi^*| - |E'| \geq |\Pi^*|/2$, since 
every long edge $e \in E'$ lies in the transitive closure of at most one path in $\Pi^*$  by Claim \ref{obs:long_edge}.

Now fix an $s \leadsto t$ path $\pi^* \in \Pi_1^*$. Since path $\pi^*$ is the unique shortest $s \leadsto t$ path in $G^*$,  $E' \cap TC(\pi^*) = \emptyset$,
and $D_1, \dots, D_k$ is an exact distance-preserving DAG cover of $G^*$,  path $\pi^*$ must be  contained in  DAG $D_i$ as a subgraph, for some $i \in [1, k]$. There are $|\Pi_1^*| = |\Pi^*|/2 = n^{2+\varepsilon/50}$ distinct paths in $\Pi_1^*$. By Lemma \ref{lem:path_cover2}, every DAG $D_i$ in our DAG cover $\{D_i\}_{i \in [1, k]}$ contains at most $O(n^{1/3+2\varepsilon}p^{4/3})$ distinct paths in $\Pi_1^*$. Then  the number of DAGs in our DAG cover is at least
$$
k \geq \frac{|\Pi_1^*|}{O(n^{1/3+2\varepsilon}p^{4/3})} = \frac{n^{2+\varepsilon/50}}{O(n^{1/3+2\varepsilon}p^{4/3})} \geq n^{1/3 - 3\varepsilon} \geq N^{1/6},
$$
where the final inequality follows by taking $\varepsilon>0$ to be sufficiently small. 
\end{proof}

Notice that our proof of Theorem \ref{thm:long} did not use the flexibility of construction parameter $c$, nor did it use the full power of the matching decomposition claimed in Lemma \ref{lem:match_decomp}. Both of these details will become relevant in the next subsubsection, where we will complete the proof of the full version of Theorem \ref{thm:lb_m}.

% \begin{claim}
%     Fix a DAG $D$ in $\mathcal{D}$. Fix a collection of $k$ pairs of edges $(e_1, e_1'), (e_2, e_2')  \dots, (e_k, e_k')$ from the set $E(G^{\times}) \times E(G^{\times})$ with the following properties:
%     \begin{enumerate}
%         \item edges $e_1, e_1', e_2, e_2', \dots, e_k, e_k' \in E(G^{\times})$ are all distinct, and
%         \item $e_i = (u_i, v_i)$ and $e_i' = (v_i, w_i)$ for some $u_i, v_i, w_i \in V(G^{\times})$, for all $i \in [1, k]$.  
%     \end{enumerate}
%     For all $i \in [1, k]$ let $\phi$
% \end{claim}

\subsubsection{Proof of Theorem \ref{thm:lb_m} in the general case}
We restate Theorem \ref{thm:lb_m} below for reference.

\lbm*

A key concept we will need in our analysis of the short edges in our DAG cover is the concept of path-edge incidences.

\begin{definition}[Path-edge incidences]
    Given a collection of paths $\Pi' \subseteq \Pi^*$ and a collection of edges $E' \subseteq E(G^*)$, we define the set of path-edge incidences $\mathcal{I}_{E', \Pi'}$ between $E'$ and $\Pi'$ as
    $$
    \mathcal{I}_{E', \Pi'} = \{(e, \pi) \in E' \times \Pi' \mid e \in TC(\pi)\}.
    $$
    \label{def:path-edge}
\end{definition}

At a high level, the reason that the concept of path-edge incidences will be useful for analyzing short edges is that 
 a short edge $e$ can only shortcut a path $\pi \in \Pi^*$ if $e \in TC(\pi)$. As a consequence, if we can argue that there exists many paths in $\Pi^*$ that have few path-edge incidences with short edges, then we can hopefully argue that the short edges will not affect our lower bound that much. 

  In the following claim, we will use Claim \ref{clm:small_paths2} to argue that given a  DAG $D \in \mathcal{D}$ and a set of additional edges $E' \in TC(D)$, if the number of path-edge incidences between a large subset $\Pi' \subseteq \Pi^*$ of $\Pi^*$ and the set of edges $E'$ is small, then $(D \cap G^*) \cup E'$ preserves the distances between the endpoints of paths in $\Pi^*$ with very low probability.

\begin{claim}[Similar to Claim \ref{clm:matching_bound}]
Fix a sufficiently small $\varepsilon > 0$, and assume that construction parameter $p$ is chosen so that $p \leq n^{1 + \varepsilon/100}$.
    Fix a DAG $D \in \mathcal{D}$. 
    Let $\Pi' \subseteq \Pi$ be a set of $|\Pi'| = p \cdot n^{-\varepsilon}$ paths in $\Pi$, where $\Pi$ is the set of base paths constructed in subsection \ref{subsec:base}. Let $\Pi'' \subseteq \Pi^{\times}$ be a set of $|\Pi''| = k \geq n^{1/3}p^{5/6}$  paths in $\Pi^{\times}$ contained in the rectangle $\Pi' \times \Pi'$, i.e.,
    $$
    \Pi'' \subseteq \Pi' \times \Pi' \subseteq \Pi^{\times}. 
    $$
    Additionally, let $\Pi^{(3)} = \psi[\Pi''] \subseteq \Pi^*$ denote the image of set $\Pi''$ under function $\psi$. Let $\gamma > 0$ be a sufficiently large constant, and let $A$ be the event that there exists a set of additional edges $E' \subseteq TC(D \cap G^*) \setminus E(D \cap G^*)$ such that:
    \begin{enumerate}
        \item Every edge in $E'$ is a short edge, 
        \item $|\mathcal{I}_{E', \Pi^{(3)}}| \leq \frac{\ell}{n^{\eps/50}} \cdot |\Pi^{(3)}|$, and
        \item  For every $s \leadsto t$ path $\pi \in \Pi^{(3)}$, DAG $(D \cap G^*) \cup E'$ exactly preserves the distance between $s$ and $t$ in $G^*$, i.e., 
        $$
        \dist_{(D \cap G^*) \cup E'}(s, t) = \dist_{G^*}(s, t).
        $$ 
    \end{enumerate}
    Then event $A$ occurs with probability at most
    $$
    \left(\frac{3}{4}\right)^{\Theta\left(k \cdot \frac{n^{1+\varepsilon/10}}{p \cdot \log^2 n}\right)}.
    $$
    \label{clm:matching_bound2}
\end{claim}
\begin{proof}
% Note conditions one and two are deterministic (is that true?)
Fix a DAG $D \in \mathcal{D}$ and sets $\Pi', \Pi'',$ and $\Pi^{(3)}$ as described in the statement of Claim \ref{clm:matching_bound2}. 
Let graph $G'' = \graph(V(G^{\times}), \Pi'')$ be the graph $G^{\times}$ induced on paths in $\Pi''$, and let $J = J(n, p, \Pi'')$ be the matching graph of $G''$ as defined in Definition \ref{def:matching_graph}. By Lemma \ref{lem:match_decomp}, the matching graph $J$ admits a collection of $\alpha$ matchings $M_1, \dots, M_{\alpha}$, with each matching of size $|M_i| = \beta$, where $\alpha = \Theta\left(\frac{p^{2/3}}{n^{1/3+\varepsilon/10}}\right)$ and $\beta = \Theta\left( k \cdot \frac{n^{1+\varepsilon/10}}{p \cdot \log^2 n}  \right)$.

We can interpret each matching $M_i$, $i \in [1, \alpha]$, as a collection of paths each of length one in $J$. Under the interpretation of $M_i$ as a collection of paths, we can reason about the set of path-edge incidences $\mathcal{I}_{E', M_i}$ between $M_i$ and a set of short edges $E' \subseteq TC(D \cap G^*) \setminus E(D \cap G^*)$. We claim that for every set of short edges $E' \subseteq  TC(D \cap G^*) \setminus E(D \cap G^*)$ such that $|\mathcal{I}_{E', \Pi^{(3)}}| \leq \frac{\ell}{n^{\varepsilon/50}} \cdot |\Pi^{(3)}|$, 
there must exist a matching $M_i$, $i \in [1, \alpha]$, such that $|\mathcal{I}_{E', M_i}| \leq |M_i|/2$. This claim emerges from the following observations:
\begin{enumerate}
    \item By the definition of matching graph $J$, we must have that for all $i \in [1, \alpha]$, $\mathcal{I}_{E', M_i} \subseteq \mathcal{I}_{E', \Pi^{(3)}}$.
    \item Since matchings $M_i, M_j$ are edge-disjoint for $i \neq j$, we must have that 
    $$
    \left| \cup_{i \in [1, \alpha]}M_i \right| = \alpha \beta = \Theta\left( 
 \frac{\ell}{\log^2 n} \cdot |\Pi^{(3)}|\right) \gg|\mathcal{I}_{E', \Pi^{(3)}}|.
    $$
\end{enumerate}
By combining these two observations and using a simple averaging argument, we can  show there exists a matching $M_i$, $i \in [1, \alpha]$, such that $|\mathcal{I}_{E', \Pi^{(3)}}| \leq \frac{\ell}{n^{\varepsilon/50}} \cdot |\Pi^{(3)}|$. 
Let $M_i' \subseteq M_i$ denote the set 
$$
M_i' = \{e \in M_i \mid e \not \in E'\}.
$$
By the above discussion, $$|M_i'| \geq |M_i|/2 \geq \beta/2 =  \Theta\left(k \cdot \frac{n^{1+\varepsilon/10}}{p \cdot \log^2 n}\right).$$
Now note that no edge in matching $M_i'$ coincides with a short edge in $E'$. 
As discussed in the proof of Claim \ref{clm:matching_bound}, each edge $e$ in matching $M_i'$ corresponds to a path $\pi$ on four nodes in graph $G^*$. 
Moreover, by our choice of $M_i'$, $TC(\pi) \cap E' = \emptyset$.
Then we must necessarily have that path $\pi$ is a subgraph of DAG $D \cap G^*$, which holds with probability at most $3/4$ by Claim \ref{clm:small_paths2}. Then just as in the proof of Claim \ref{clm:matching_bound}, we have that 
$$
\Pr[A] \leq \left( \frac{3}{4} \right)^{|M_i'|} = \left( \frac{3}{4} \right)^{\Theta\left(k \cdot \frac{n^{1+\varepsilon/10}}{p \cdot \log^2 n}\right)}.
$$
\end{proof}

We will now combine the very strong probability bound proved in Claim \ref{clm:matching_bound2} with a union bound to argue that there does not exist a DAG $D \cap G^*$, where $D \in \mathcal{D}$, that when unioned with a set of additional edges $E' \subseteq TC(D \cap G^*)$,  can exactly preserve many distances in $G^*$, if we assume that paths in $\Pi^*$ have few of 
a certain set of path-edge incidences with $E'$.

\begin{lemma}[Similar to Lemma \ref{lem:path_cover2}]
    Fix a sufficiently small $\varepsilon > 0$, and assume that construction parameter $p$ is chosen so that $p \leq n^{1 + \varepsilon/100}$. 
    With nonzero probability, there does not exist DAG $D \in \mathcal{D}$, a collection of edges $E' \subseteq TC(D \cap G^*) \setminus E(D \cap G^*)$,  
    and a collection of paths $\Pi' \subseteq \Pi^*$ such that 
    \begin{enumerate}
        \item Every edge in $E'$ is a short edge,
        \item $|\Pi'| = n^{1/3 +2\varepsilon}p^{4/3} c \log c$,  
        \item $|\mathcal{I}_{E', \{\pi\}}| \leq \frac{\ell}{n^{\varepsilon/50}}$ for all $\pi \in \Pi'$,  and
        \item For every $s \leadsto t$ path $\pi \in \Pi'$, DAG $(D \cap G^*) \cup E'$ exactly preserves the distance between $s$ and $t$ in $G^*$, i.e., 
        $$
        \dist_{(D \cap G^*) \cup E'}(s, t) = \dist_{G^*}(s, t). 
        $$
    \end{enumerate}

    \label{lem:path_cover3}
\end{lemma}
\begin{proof}
 Our strategy for proving this lemma is to use Claim \ref{clm:matching_bound2} and apply the union bound. Before we can do this, we need to develop some additional notation. 

    Let $\alpha = n^{\varepsilon}$. Let $\Pi_1, \dots, \Pi_{\alpha}$ be a partition of the set of base paths $\Pi$ into $\alpha$ sets $\Pi_i$ each of size $|\Pi_i| \leq |\Pi|/\alpha \leq p \cdot n^{\varepsilon}$. For all $i, j \in [1, \alpha]$, we define the subset $\Pi_{i, j}^{\times} \subseteq \Pi^{\times}$ as follows:
    $$
    \Pi_{i, j}^{\times} := \Pi_i \times \Pi_j.
    $$
    For all $i, j \in [1, \alpha]$, we define $\Pi^*_{i, j} \subseteq \Pi^*$ to be the image of set $\Pi_{i, j}^{\times}$ under function $\psi$. Formally, $\Pi^*_{i, j} = \psi[\Pi^{\times}_{i,j}]$. Notice that the collection of sets $\{\Pi^*_{i, j}\}_{i, j \in [1, \alpha]}$ is a partition of set $\Pi^*$.  
    Additionally, since each set $\Pi^*_{i, j}$ is contained inside a small rectangle $\psi[\Pi_{i} \times \Pi_j]$, and each path $\pi \in \Pi'$ has few path-edge incidences with $E'$, we can directly apply Claim \ref{clm:matching_bound2} on each set of paths $\Pi^*_{i, j}$. 

    For all $i, j \in [1, \alpha]$, let $A_{i, j}$ be the event that
    \begin{enumerate}
        \item there exists a DAG $D \in \mathcal{D}$,
        \item there exists a subset of paths $\Pi' \subseteq \Pi^*_{i, j}$ of size $|\Pi'|= k$, and 
        \item there exists a collection of short edges $E' \subseteq TC(D \cap G^*) \setminus E(D \cap G^*)$,
    \end{enumerate}
    such that
    \begin{enumerate}
  \item $|\mathcal{I}_{E', \Pi'}| \leq \frac{\ell}{ n^{\eps/50}} \cdot |\Pi'|$, and
        \item  For every $s \leadsto t$ path $\pi \in \Pi'$, DAG $(D \cap G^*) \cup E'$ exactly preserves the distance between $s$ and $t$ in $G^*$, i.e., 
        $$
        \dist_{(D \cap G^*) \cup E'}(s, t) = \dist_{G^*}(s, t).
        $$
    \end{enumerate}
By Claim \ref{clm:matching_bound2} and the union bound, we have that
\begin{align*}
\Pr\left[ \bigcup_{i, j \in [1, \alpha]} A_{i, j} \right] & \leq \sum_{i, j \in [1, \alpha]}\Pr[A_{i, j}] \leq \alpha^2 \cdot |\mathcal{D}| \cdot {|\Pi_{i, j}^*| \choose k} \cdot  \left(\frac{3}{4}\right)^{\Theta\left(k \cdot \frac{n^{1+\varepsilon/10}}{p \cdot \log^2 n}\right)}\\
& \leq n^{2\varepsilon}  \cdot 2^{\Theta(n^{4/3}p^{1/3} \cdot c \log c)} \cdot n^{\Theta(k)} \cdot 2^{-\Theta\left( k \cdot \frac{n^{1+\varepsilon/10}}{p \cdot \log^2 n} \right)} \\
& < 1,
\end{align*}
when $k \geq n^{1/3}p^{4/3}c \log c$. 
Then by the probabilistic method, we may assume that event $A_{i, j}$ does not hold for any $i, j \in [1, \alpha]$, when we specifically choose  $k = n^{1/3}p^{4/3}c \log c$ in the definition of $A_{i, j}$.

    Now to see why Lemma \ref{lem:path_cover3} follows, suppose towards contradiction that there exists a DAG $D \in \mathcal{D}$ and a set of short edges $E' \subseteq TC(D \cap G^*) \setminus E(D \cap G^*)$ such that $(D \cap G^*) \cup E'$ exactly preserves the distances between the endpoints of a set $\Pi'$ of $|\Pi'| > n^{1/3 + 2\varepsilon}p^{4/3}c \log c$ distinct paths in $\Pi^*$, where each path $\pi \in \Pi'$ has at most $\frac{\ell}{n^{\varepsilon/50}}$ incidences with $E'$. Then by an averaging argument, there must exist a set $\Pi^*_{i, j} \subseteq \Pi^*$, with $i, j \in [1, \alpha]$, such that DAG $(D \cap G^*) \cup E'$ preserves more than
    $$
    \frac{n^{1/3 + 2\varepsilon}p^{4/3}c \log c}{\alpha^2} = n^{1/3}p^{4/3}c \log c
    $$
    distinct shortest path distances between the endpoints of paths  in $\Pi^*_{i, j} \cap \Pi'$, contradicting our assumption that event $A_{i, j}$ does not occur. 
\end{proof}

Lemma \ref{lem:path_cover3} gives us a very powerful tool to argue that DAGs in our DAG cover cannot cover many exact shortest path distances simultaneously. In order to apply Lemma \ref{lem:path_cover3}, we will need to bound the number of possible path-edge incidences between the set of short edges in our DAG cover and our set of paths $\Pi^*$.

For the remainder of the proof, we will assume that construction parameter $c$ is chosen so that $c \leq \frac{d^{1/2}}{\log n}$. 

\begin{claim}
    With high probability, every every node $v \in V(G^*)$ has degree at most $\deg_{G^*}(v) \leq O(d/c)$ in $G^*$. 
    \label{clm:low_degree}
\end{claim}
\begin{proof}
    Node $v$ is in a clique subgraph $K_c^u$ in $G^*$ for some $u \in V(G^{\times})$.
    Node $v$ has at most $c-1 \leq d/c$ edges to other nodes in $K_c^u$. Additionally, by Lemma \ref{lem:prod_graph}, $\deg_{G^{\times}}(u) = O(d)$. Each edge $(u, w)$ incident to node $u$ in $G^{\times}$ is reattached to node $v \in V(K_c^u)$ in $G^*$ independently with probability $1/c$, by the construction of $G^*$. Then the expected size of $\deg_{G^*}(v)$ is $\deg_{G^{\times}}(u) / c = O(d/c)$. If $\deg_{G^{\times}}(u) / c \leq 10 \log n$, then $$\deg_{G^*}(v) \leq \deg_{G^{\times}}(u) + d/c \leq c \cdot 10 \log n + d/c  = O(d/c).$$
    Otherwise, if $\deg_{G^{\times}}(u) / c \geq  10 \log n$, then we can apply the Chernoff bound to argue that $\deg_{G^*}(v) = O(d/c)$ with high probability. We complete the proof by union bounding over all $v \in V(G^*)$. 
\end{proof}

We will use Claim \ref{clm:low_degree} to argue that every short edge $e \in TC(G^*) \setminus E(G^*)$ cannot appear in the transitive closure of too many paths in $\Pi^*$.

\begin{claim}
    With high probability, every short edge $e \in TC(G^*) \setminus E(G^*)$ is contained in the transitive closure of at most $O(d/c)$ distinct paths in $\Pi^*$. 
    \label{obs:short}
\end{claim}
\begin{proof}
    Let $(x, y) \in TC(G^*) \setminus E(G^*)$ be a short edge. Then $(x, y) \in V(K_c^u) \times V(K_c^v)$ for nodes $u, v \in V(G^{\times})$ such that $(u, v) \in E(G^{\times})$. Moreover, since $(x, y) \not \in E(G^*)$, we must have that $(x, y) \neq \phi(u, v) \in E(G^*)$. In particular, let $\phi(u, v) = (u^*, v^*) \in E(G^*)$. We must have that either $x \neq u^*$ or $y \neq v^*$. Assume without loss of generality that $y \neq v^*$ (the other case is identical). 

    Let $\pi_1, \dots \pi_k$ denote the set of paths in $\Pi^{\times}$ that contain edge $(u, v)$. Let $(v, w_i) \in E(G^{\times})$ denote the edge that path $\pi_i$ takes after reaching node $v$ in $G^{\times}$. 
    By Property 5 of Lemma \ref{lem:prod_graph}, any two paths $\pi_i, \pi_j$ with $i \neq j$  must travel along different edges $(v, w_i)$ and $(v, w_j)$, where $w_i \neq w_j$, from node $v$. 
    This in turn implies that any two paths $\psi(\pi_i), \psi(\pi_j) \in \Pi^*$ with $i \neq j$ must travel along different edges $\phi((v, w_i))$ and $\phi((v, w_j))$ from clique $K_c^v$ in $G^*$. 
    Additionally, if edge $(x, y)$ is in the transitive closure of the image $\psi(\pi_i)$ of a path $\pi_i \in \Pi^{\times}$, then necessarily $\phi((v, w_i)) \in \{y\} \times V(K_c^{w_i})$ because $y \neq v^*$ and every path $\psi(\pi_i) \in \Pi^*$ contains at most two nodes in any clique subgraph $K_c^v$. 
    
    By Claim \ref{clm:low_degree}, we have that $\deg_{G^*}(y) = O(d/c)$ with high probability. Then at most $O(d/c)$ distinct paths in $\Pi^*$ can use edge $(x, y)$.
\end{proof}

The following claim is immediate from Definition \ref{def:path-edge} and \Cref{obs:short}. 

\begin{claim}
With high probability, for every collection  $E' \subseteq TC(G^*) \setminus E(G^*)$ of short edges in $TC(G^*)$,  the number of path-edge incidences between $\Pi^*$ and $E'$ is at most
$$
| \mathcal{I}_{E', \Pi^*}| \leq |E'| \cdot O(d/c). 
$$
\label{clm:path-edge-inc}
\end{claim}

We are almost ready to prove Theorem \ref{thm:lb_m}.

\paragraph{Specifying our construction parameters.}
Let $\varepsilon > 0$ be a sufficiently small constant. We will choose our construction parameters $p$ and $c$ of $G^*(n, p, c)$ to be $p = n^{1+\varepsilon/100}$ and $c =  n^{3\varepsilon/100}$.  Notice that under this choice of construction parameters, graph $G^*$ and paths $\Pi^*$ have the following properties:
\begin{itemize}
    \item $\ell = \frac{n^{2/3}}{p^{1/3}} = \Theta\left( n^{1/3 -\varepsilon/300} \right)$,  $d = \frac{p^{2/3}}{n^{1/3}} = \Theta(n^{1/3+\varepsilon/150})$, and $c \leq \frac{d^{1/2}}{\log n}$ as assumed, 
    \item $|V(G^*)| = \Theta\left(c \cdot \frac{n^2}{\ell}\right) = \Theta(n^{5/3 + \varepsilon/30})$,
    \item $|E(G^*)| = \Theta\left( \frac{n^2}{\ell} \cdot (d+c^2) \right) = \Theta(n^{2+\varepsilon/100})$, and
    \item $|\Pi^*| = \Theta(p^2) = \Theta(n^{2+\varepsilon/50})$. 
\end{itemize}
Let $N = |V(G^*)|$ and let $m = |E(G^*)|$. Let $D_1, \dots, D_k$ be an exact distance preserving DAG cover with a set $E'$ of at most $|E'| \leq m^{1+\varepsilon/300} \leq |\Pi^*|/4$ additional edges. Let $E_L'$ denote the edges in $E'$ that are long, and let $E_S'$ denote the edges in $E'$ that are short. 
We define a subset $\Pi_1^* \subseteq \Pi^*$ of paths in $\Pi^*$ as follows:
 $$\Pi_1^* = \left\{\pi^* \in \Pi^* \mid TC(\pi^*) \cap E_L' = \emptyset \text{ and } |TC(\pi^*) \cap E_S'| \leq \frac{\ell}{ n^{\varepsilon/50} }  \right\}.$$ 

 We now prove that $\Pi_1^*$ contains a constant fraction of the paths in $\Pi^*$. 

\begin{claim}
    $|\Pi_1^*| \geq |\Pi^*|/4$. 
\end{claim}
\begin{proof}
By Claim \ref{obs:long_edge}, at most $|E'| \leq |\Pi^*|/4$ paths in $\Pi^*$ can contain a long edge in their transitive closure, so the set
$$
\Pi_0^* = \left\{\pi^* \in \Pi^* \mid TC(\pi^*) \cap E_L' = \emptyset  \right\}
$$
has size $|\Pi_0^*| \geq 3/4 \cdot |\Pi^*|=3/4 \cdot p^2$. By Claim \ref{clm:path-edge-inc}, the number of path-edge incidences between $\Pi^*$ and $E'_S$ is at most
$$
| \mathcal{I}_{E'_S, \Pi^*}| \leq |E'_S| \cdot O(d/c) \leq |\Pi^*|/4 \cdot O(d/c) = O(d/c \cdot p^2) = O\left( \frac{\ell p^2}{n^{7\varepsilon/300}} \right), 
$$
where the final inequality follows from our choice of $c$ and the fact that $d/\ell = n^{\varepsilon/100}$.
Notice that $\Pi_1^* \subseteq \Pi_0^*$. By an averaging argument, we can show that $|\Pi_1^*| \geq |\Pi_0^*|/3$. Suppose towards contradiction that instead $|\Pi_1^*| < |\Pi_0^*|/3$. Then 
$$
| \mathcal{I}_{E'_S, \Pi^*}| \geq 2/3 \cdot |\Pi_0^*| \cdot \frac{\ell}{ n^{\varepsilon/50}} = \Omega\left( \frac{\ell p^2}{n^{\varepsilon/50}}\right),
$$
contradicting our previous upper bound on $| \mathcal{I}_{E'_S, \Pi^*}|$, since $\varepsilon/50 < 7\varepsilon/300$. We conclude that
$$
|\Pi_1^*| \geq |\Pi_0^*|/3 \geq |\Pi^*|/4,
$$
as claimed.
\end{proof}

We are now  ready to prove Theorem \ref{thm:lb_m}. 

\begin{proof}[Proof of Theorem \ref{thm:lb_m}]
For every $s \leadsto t$ path $\pi^* \in \Pi_1^*$,  we observe the following:
\begin{itemize}
    \item Path $\pi^*$ is the unique shortest $s \leadsto t$ path in $G^*$, and 
    \item By our choice of $\Pi_1^*$,  $E_L' \cap TC(\pi^*) = \emptyset$.
\end{itemize} 
Then every $s \leadsto t$ path $\pi^* \in \Pi_1^*$ can only be helped by short edges in $E'_S$. Additionally, by our choice of $\Pi_1^*$ we know that every path $\pi \in \Pi_1^*$ has at most 
$$
|\mathcal{I}_{E_S', \{\pi\}}| \leq \frac{\ell}{n^{\varepsilon/50}}
$$
path-edge incidences with the set of short edges $E_S'$.
Let $$P = \{(s, t) \in V(G^*) \times V(G^*) \mid \text{there exists an $s \leadsto t$ path $\pi \in \Pi_1^*$}\}$$ denote the set of endpoints of paths in $\Pi_1^*$. 
Then we can apply Lemma \ref{lem:path_cover3} to argue that every DAG $D_1, \dots, D_k$ preserves exact distances between at most $$ n^{1/3}p^{4/3}c \log c = O(n^{1/3+3\varepsilon}p^{4/3})$$ 
distinct pairs of vertices $(s, t) \in P$. Then  the number of DAGs in our DAG cover is at least
$$
k \geq \frac{|\Pi_1^*|}{O(n^{1/3+3\varepsilon}p^{4/3})} = \frac{n^{2+\varepsilon/50}}{O(n^{1/3+3\varepsilon}p^{4/3})} \geq n^{1/3 - 4\varepsilon} \geq N^{1/6},
$$
where the final inequality follows by taking $\varepsilon>0$ to be sufficiently small. 
\end{proof}

\section{DAG covers with no additional edges}

In this section we will prove the following theorem: \diam*

\subsubsection*{Construction}

To build the graph $G$ we begin with two vertices $s,t$. We partition the rest of the $n-2$ vertices into two groups $A=a_1,\dots,a_{n/2-1}$ and $B=b_1,\dots,b_{n/2-1}$. We add a directed path $P_A$ from $s$ to $t$ whose internal vertices are exactly $a_1,\dots,a_{n/2-1}$ (in that order). We also add a directed path $P_B$ from $t$ to $s$ whose internal vertices are exactly $b_1,\dots,b_{n/2-1}$ (in that order). 

Next we add nested ``shortcut'' edges to the paths $P_A$ and $P_B$ in the following way. We impose a binary tree structure on each path $P_A$, $P_B$. The root node, which is on level 1, represents the entire path, and the nodes in level $i$ of the tree represent edge-disjoint directed subpaths on $n/2^i$ edges (whose union is the entire path). There are $\log_2n$ levels total, and each leaf node represents a single edge of the path. Let $T_A$ and $T_B$ be the trees corresponding to $P_A$ and $P_B$ respectively. For every node in each of the two trees, we add an edge from the first to last vertex of the corresponding subpath.

\subsubsection*{Analysis}
First, we claim that the diameter of $G$ is $O(\log n)$. This is evident from the binary tree structure of the shurtcut edges. Specifically, every vertex in $A$ is reachable from $s$ and can reach $t$ by a path of length $O(\log n)$, and every vertex in $B$ is reachable from $t$ and can reach $s$ by a path of length $O(\log n)$.

It remains to show that any reachability-preserving DAG cover of $G$ with no additional edges requires $n$ DAGs. Let the \emph{predecessor} of a vertex $v$, denoted $\text{pred}(v)$ be the vertex that falls before $v$ on the cycle composed of $P_A$ and $P_B$. We will prove the following lemma:

\begin{lemma}\label{lem:diam}
    For any pair $u,v$ of vertices, the union of (1) any path from $u$ to $\text{pred}(u)$, and (2) any path from $v$ to $\text{pred}(v)$, contains a cycle.
\end{lemma} 

Proving \cref{lem:diam} completes the analysis because it implies that any DAG cover requires a separate DAG to capture the reachability from each vertex to its predecessor, for a total of $n$ DAGs. 

We first make the following observation, which is true by the construction of $G$:

\begin{observation}\label{obs:st}
    For every vertex $v$, any path from $v$ to $\text{pred}(v)$ contains both $s$ and $t$. Specifically, if $v\in A \cup \{t\}$, the path visits $t$ before $s$, and if $v\in B\cup \{s\}$, the path visits $s$ before $t$.
\end{observation}

Since the union of a path from $s$ to $t$ and a path from $t$ to $s$ form a cycle, \cref{obs:st} implies that \cref{lem:diam} is true if $u\in A \cup \{t\}$ and $v\in B\cup \{s\}$. It remains to consider the case where both $u$ and $v$ are in $A \cup \{t\}$; the case where both $u$ and $v$ are in $B\cup \{s\}$ is symmetric. Fix $u,v\in A \cup \{t\}$, where $u$ comes before $v$ on $P_A$.

The following notation will be useful. For any vertex $w\in A \cup \{t\}$, let $\ell_w$ be the leaf node in $T_A$ corresponding to the edge $(\text{pred}(w), w)$. For any ancestor of $\ell_w$ in the tree $T_A$, call the corresponding subpath an \emph{ancestor subpath} of the edge $(\text{pred}(w), w)$. We will refer to the \emph{least common ancestor subpath} of $(\text{pred}(u), u)$ and $(\text{pred}(v), v)$. (For these definitions a tree node is said to be its own ancestor.)

We now state a strengthening of \cref{obs:st}, which follows from the nested structure of the shortcut edges:

\begin{observation}\label{obs:stgen}
    For any vertex $w\in A \cup \{t\}$, any path $P$ from $w$ to $\text{pred}(w)$ contains both the start and end vertex of every ancestor subpath of $(\text{pred}(w), w)$. The ordering of vertices on $P$ is as follows: each subpath end vertex (besides $t$) is before $t$, which is before $s$, which is before each subpath start vertex (besides $s$).
\end{observation}

Consider the least common ancestor subpath $P_{anc}$ of $(\text{pred}(u), u)$ and $(\text{pred}(v), v)$. Let $P_1$ and $P_2$ be the two children subpaths of $P_{anc}$. That is, the vertex $z$ that is both the end of $P_1$ and the start of $P_2$ falls between $u$ and $v$ on $P_A$. Furthermore, $z$ is the end vertex of an ancestor subpath of $(\text{pred}(u), u)$ and the start vertex of an ancestor subpath of $(\text{pred}(v), v)$. By \cref{obs:stgen}, any path from $u$ to $\text{pred}(u)$ contains $z$, $t$, and $s$, in that order, and any path from $v$ to $\text{pred}(v)$ contains $t$, $s$, and $z$ in that order. Together, these $z$-to-$t$-to-$s$ and $t$-to-$s$-to-$z$ paths form a cycle. Thus, the union of (1) any path from $u$ to $\text{pred}(u)$, and (2) any path from $v$ to $\text{pred}(v)$, contains a cycle. This completes the proof of \cref{lem:diam}.

\section{Open Problems}

The main problem left open by our work is closing the gap between our upper and lower bounds for DAG covers with $\tilde{O}(m)$ additional edges. Our upper bound is polylogarithmic in both the distortion and the number of DAGs, while our lower bound is for exact distances and a polynomial number of DAGs. We conjecture that the lower bound can be improved to handle approximate distances, as well as a larger polynomial number of DAGs.

Another open problem is to find concrete applications of our DAG cover algorithm. We suspect that such applications exist, given the wide-reaching applications of the analogous undirected constructions.
Another open problem is to extend our DAG cover algorithm to various settings such as the distributed, parallel, online, dynamic, and streaming settings, since such settings have been fruitful for the analogous undirected constructions. 
Lastly, it would be interesting to consider the ``Steiner'' version of this problem where vertices, in addition to edges, are allowed to be added.

\section{Acknowledgements}
We would like to thank Aaron Bernstein, Shimon Kogan, and Merav Parter for the conversation that initiated this work, as well as additional fruitful conversations about how to define the problem. We would like to thank Arnold Filtser for offering detailed comments and questions on an earlier version of this manuscript. 

\bibliographystyle{alpha}
\bibliography{ref.bib}

\end{document}